\pdfoutput=1
\documentclass[10pt]{report}
\usepackage[utf8]{inputenc}

\usepackage{amsmath}
\DeclareMathOperator\arctanh{arctanh}

\usepackage{tikz}

\usepackage{amsfonts}

\usepackage{amsthm}

\usepackage{enumitem}

\usepackage{cite}

\theoremstyle{definition}
\newtheorem{definition}{Definition}[section]
\newtheorem{theorem}{Theorem}[chapter]
\newtheorem{corollary}{Corollary}[theorem]

\title{Analytical Advances about the Apex Field Enhancement Factor of a Hemisphere on a Post Model}
\author{Adson Soares}

\date{\today}

\begin{document}
%\begin{titlepage}
%\maketitle
%\end{titlepage}
%\newpage

\begin{titlepage}
\begin{center}
\vspace*{1cm}
\large
\textbf{Universidade Federal da Bahia} \\
\textbf{Instituto de Física} \\
\textbf{Programa de Pós-graduação em Física}\\

\vspace{1.5cm}
\textbf{Adson Soares de Souza}

\vspace{1.5cm}
\LARGE
\textbf{Analytical Advances about the Apex Field Enhancement Factor of a Hemisphere on a Post Model}

\normalsize
\vfill
\vspace{0.8cm}
\vspace{1.5cm}
Salvador,\\
2019
\end{center}
\end{titlepage}

\begin{titlepage}
\begin{center}
\vspace*{1cm}
\large
\textbf{Universidade Federal da Bahia} \\
\textbf{Instituto de Física} \\
\textbf{Programa de Pós-graduação em Física}\\

\vspace{1.5cm}
\textbf{Adson Soares de Souza}

\vspace{1.5cm}
\LARGE
\textbf{Analytical Advances about the apex Field Enhancement Factor of a Hemisphere on a Post Model}

\normalsize
\vfill
Dissertation submitted to Programa de Pós-graduação em Física, Instituto de Física, Universidade Federal da Bahia, as a partial requirement for obtaining a Master's Degree in Physics.
% Dissertação apresentada ao Programa de Pós-graduação em Física, Instituto de Física, Universidade Federal da Bahia, como requisito parcial para obtenção do grau de Mestre em Física.

\vspace{0.8cm}
Supervisor: Prof. Dr. Thiago Albuquerque de Assis.

\vspace{1.5cm}
Salvador,\\
2019
\end{center}
\end{titlepage}

\newpage

% To define HCP before the method section.
% To speak about the importance of calculating the FEF.
% To read paper Richard? 2003?, that has several FEF applications.
% To treat equations as if they were part of text: dots, commas, etc, into the actual mathematical expressions.

% Ellipsoid chapter.
% To calculate second order FEF.
% To compare expressions of G1, G3, I11, I13, A1, A3, etc with the exact numerical results.
% To write exact FEF expression for ellipsoid (and cite it, of course!)
% To compare the expressions of the FEF with exact analytical expressions. To show it becomes non accurate for ee>>1.
% To plot a graph about the error falling exponentially (semilog plot).
% To check why first order FEF expression decreases if h>1. FEF=3 for h=1, and FEF<3 for h>1. Error? Come up with explanation! Might be asked!
% To complete slow convergence chapter. It is VERY incomplete. Either delete or complete.
% To derive slow convergence: numerically reference that charge is concentrated on the top. Axial multipoles for charge on the top grows a^k. Meaning, series are harmonic 1/n (slow).

% HCP Chapter.
% Check cylindrical integrals. Numerically calculate them.
% Check convergence of two dimensional array summation. Calculate its value. Delete section it sum diverges.

% To generate pictures for the introduction. Pictures for several models used.
% To generate triangle pictures with the quantities, to derive the quantities in the shapes.

% To write more comprehensive appendix shortcuts [instead of (A.4.45) to do: As in equation (A.4.45) in appendix A]

\chapter*{Abstract}
In this dissertation, we analytically study the apex field enhancement factor (FEF), $\gamma_a$, by constructing a method which consists in minimizing an error function defined as to measure the deviation of the potential at the boundary, yielding approximate axial multipole coefficients of a general axial-symmetric conducting emitter shape, on which the apex FEF can be calculated by summing its respective Legendre series. Such method is analytically applied for a conducting hemisphere on a flat plate, confirming the known result of $\gamma_a = 3$. Also, it is applied on a hemi-ellipsoid on a plate where the values of the apex FEF are compared with the ones extracted from the analytical expression. Then, the method is applied for the hemisphere on a cylindrical post (HCP) model. In this case, to analytically estimate the apex FEF from first principles is a problem of considerable complexity.

Despite the slow convergence of the apex FEF, useful analytical conclusions are drawn and explored, such as, it is confirmed that all even multipole contributions of the HCP model are zero, which in turn leads to restrictions on the charge density distribution: it will be shown the surface charge density must be an odd function with respect to height in an equivalent system. Also, expressions found for the apex FEF depend explicitly on the aspect ratio, that is, the ratio of height by base radius. Using the dominant multipole contribution, the dipole, at sufficient large distances, it is shown that, for two interacting emitters, as their separation distance increases, the fractional change in apex FEF, $\delta$, decreases following a power law with exponent $-3$. The result is extended for conducting emitters having an arbitrary axially-symmetric shape, where it is also shown $\delta$ has a pre-factor depending on geometry, confirming the tendency observed in recent analytical and numerical results.

\chapter*{Resumo}
Nessa dissertação, estudamos analiticamente o fator de amplificação por efeito de campo eletrostático (FEF), $\gamma_a$, por construir um método que consiste em minimizar uma função de erro definida para medir o desvio do potencial na superficie, produzindo coeficientes multipolos axiais aproximados de uma superfície emissora condutora axialmente simetrica, no qual o FEF no ápice pode ser obtido somando a respectiva série de Legendre. Tal método é aplicado analiticamente para uma semi-esfera condutora sobre uma placa plana, confirmando o resultado conhecido $\gamma_a=3$. Ademais, o método foi aplicado a um semi-elipsóide sobre uma placa, onde os valores do FEF no ápice são comparados com os valores extraídos a partir da expressão analítica. Em seguida, o método é aplicado para o modelo de uma semi-esfera sobre um poste cilíndrico (modelo HCP). Neste caso, estimar analiticamente o FEF no ápice a partir de primeiros princípios é um problema de considerável complexidade.

Apesar da lenta convergência do FEF no ápice, conclusões analíticas úteis são tiradas e exploradas, como por exemplo, é confirmado que todas as contribuições de multipolos pares do modelo HCP são nulas, o que por sua vez leva a restrições na distribuição de densidade de carga: será mostrado que a densidade de carga superficial deve ser uma função ímpar em relação à altura em um sistema equivalente. Além disso, expressões encontradas para o FEF no ápice dependem explicitamente da razão de aspecto, ou seja, a razão da altura com relação ao raio da base. Utilizando a contribuição de multipolo dominante, o dipolo, a dist\^{a}ncias suficientemente grandes, mostra-se que, para dois emissores interagentes, à medida que sua distância de separação aumenta, a variação fracionária do FEF no ápice, $\delta$, diminui seguindo uma lei de potência com expoente $-3$. O resultado é estendido para emissores condutores com geometria arbitrária eixo-simétricos, onde também é mostrado que $\delta$ tem um pré-fator dependendo da geometria, confirmando a tendência observada em resultados recentes, analíticos e numéricos.

\chapter*{Acknowledgements}
I'd like to thank a couple people, which, in a way or another, helped me arrive where I am right now.

Jehovah, the God of the Bible, for everything He has done for me.

Professor Thierry Jacques Lemaire, whose lectures (besides really awesome) were the primary reason why I decided to choose the physics course when I was an undergraduate student, still in engineering.

Professors Mario Cezar Ferreira Gomes Bertin, Newton Barros de Oliveira, Gildemar Carneiro dos Santos, Thiago Albuquerque de Assis, and Diego Catalano Ferraioli, for the great lectures when I was an undergraduate student. These were really enjoyable classes.

Professor Diego Catalano Ferraioli for encouraging me and giving me the initial kick, in a matter of speaking, for me to start the master program in the first place.

Professors Julien Chopin and Denis Gilbert Francis David for the awesome lectures while I was in the graduate program, full of really interesting math, calculations, and interesting topics outside the main books.

My supervisor, Thiago Albuquerque de Assis, for exposing me to scientific literature and some open problems, for showing me how things work in the scientific world, for showing me some aspects about what it takes to be a good scientist, for the (almost always) weekly meetings with insightful talks, and especially for always being motivating.

My dad, for his support and encouragement to write and finish this dissertation and the research.

Also, I'd like to thank CNPq for finantial support.

% To create list of abbreviations.
% https://tex.stackexchange.com/questions/86666/how-to-create-both-list-of-abbreviations-and-list-of-nomenclature-using-nomencl

\newpage
\chapter*{List of abbreviations and symbols}
\newlist{abbrv}{itemize}{1}
\setlist[abbrv,1]{label=,labelwidth=1in,align=parleft,itemsep=0.1\baselineskip,leftmargin=!}

\begin{abbrv}
\item[FEF] Field Enhancement Factor
\item[HCP] Hemisphere on a cylindrical post
\end{abbrv}

\begin{abbrv}
\item[$\mathbb{N}$] The set of natural numbers. That is, $\mathbb{N} = \{1, 2, 3, \dots\}$.
\item[$\mathbb{Z}$] The set of integer numbers. That is, $\mathbb{Z} = \{\dots, -3, -2, -1, 0, 1, 2, 3, \dots\}$.
\item[$h$] Height of an emitter.
\item[$R$] Base radius of an emitter.
\item[$\nu$] Aspect ratio of an emitter, defined as $\nu = h/R$.
\item[$\ell$] Defined as $\ell = h - R$.
\item[$V$] Electrostatic potential. See (\ref{potential}).
\item[$\theta$] Azimuthal angle in spherical coordinates.
\item[$r$] Radial distance in spherical coordinates.
\item[$J$] Spherical element of area, defined as $dS = Jd\theta d\phi$.
\item[$J_A$] Component of element of area, defined as $J = J_A r^2\sin\theta$.
\item[$P_l(x)$] Legendre Polynomials.
\item[$E_0$] Magnitude of applied electrostatic field in positive z-direction.
\item[$F_M$] Magnitude of applied electrostatic field in negative z-direction (that is, $F_M = -E_0$).
\item[$\Sigma$] Error function measuring deviation from real potential, defined in (\ref{error_function}).
\item[$A_l$] Legendre coefficients of the potential, defined in (\ref{potential}).
\item[$A_l^{(n)}$] Approximate value of $A_l$, when calculated from linear system (\ref{linear_system}), of dimension $(n+1)\times (n+1)$. Note: by definition, for some fixed $n$, it is true that $A_{l+n}^{(n)} = 0$ for all $l\in\mathbb{N}$.
\item[$\tilde A_l$] Defined as $\tilde A_l = A_l / E_0$.
\item[$\tilde A_l^{(n)}$] Defined as $\tilde A_l^{(n)} = A_l^{(n)} / E_0$.
\item[$\gamma_a$] Field Enhancement Factor at the apex of an emitter.
\item[$\gamma_a^{(n)}$] Approximate value of $\gamma_a$, by plugging $\tilde A_l^{(n)}$ at (\ref{FEF}).
\item[$G_l$] G-Integrals, as defined in (\ref{IG_integrals}).
\item[$I_{ij}$] I-Integrals, as defined in (\ref{IG_integrals}).
\item[$Q_l$] The $l$-th axial multipole moment. See (\ref{Ql_Al}) and (\ref{multipole_moment}).
\item[$M_l$] The $l$-th multipole moment of a line of charge. See (\ref{multipole_moment_line_of_charge}).
\item[$c$] The distance between the symmetry-axial lines of two emitters.
\item[$\delta$] Fractional reduction in apex FEF. %4.5.8.
\item[$\Phi(\mathbf r)$] Newtonian Kernel, defined at (\ref{newtonian_kernel}).
\item[$h(\mathbf r)$] Resolvent Kernel, defined implicitly by (\ref{fredholm_integral_equation}). See also (\ref{liouville_neumann_series}).
\item[$W$] Interacting first order contributions for the resolvent kernel.
\end{abbrv}

% First eccentricity?
% Kronecker delta?
% A,b from: error = Aexb(-bn)?
% \gamma' , new FEF by bringing some other shape close.
% Zeta Function.
% A1, A2: Surface areas of emitters (potential theory part).
% \item[$A_{2n}$] Defined in (\ref{A2n_moment_integrals}).

\newpage
\tableofcontents
\newpage

\chapter{Introduction}
In the context of field emission at low temperatures, often called cold field electron emission (CFE), protrusions in the micro or nano scale over a flat conducting plate, experiences an amplification of the external applied electrostatic field at some parts of its surface \cite{richard}. A common characterization parameter of such phenomena is the local field enhancement factor (FEF) at the apex, denoted by $\gamma_a = F_a/F_M$, which is the ratio between the field at the apex of the surface $F_a$, and the external macroscopic field $F_M$. For an ideal emitter, $F_M$ is taken to be $F_M = V/d$, where $V$ is the voltage difference between the two plates, and $d$ its separation distance. If the distance $d$ is large enough, the apex FEF is known to depend only on the geometry of the protrusion \cite{richard}.

The most common method of measuring emitter characteristics, is by means of a current-voltage (IV) plot: experimental control over the plate voltage $V$ is assumed, and the emitted current $I$ can be measured, for instance under high ultra-vacuum conditions. Field emission can be modeled by a Fowler-Nordheim (FN)-type equation: $J = AF_M^2\exp(-B/F_M)$ \cite{FN1928,Forbes2007,jordanpaper}, where $A$ and $B$ are parameters which may depend on the macroscopic field $F_M$, $J$ is a current density, which in some contexts is taken as a macroscopic current density $J_M$, some characteristic current density at a characteristic point $J_{kc}$, or the local current density $J_{l}$, depending on what is wanted to be analyzed, or measured.

The emission is called orthodox if some conditions are met \cite{richardemission}, such as: (1) uniform voltages on the electrodes, (2) the measured current is only due to CFE phenomena, (3) deep tunneling, that is, the Fermi level is ``much" less than the top of potential barrier in which they tunnel, (4) the work function of the emitter is constant across its surface.

The parameters $A$ and $B$ are relevant. For instance, one can model the local current density by using a one dimensional exact triangular (ET) potential barrier. It will depend on the probability of an electron to tunnel the barrier, and, then, $J$ is given by an FN-type equation, where $B$ depends on $\gamma_a$, and $A$ can be assumed to be a constant dependent on fundamental physical and mathematical constants. To extract the coefficients $A$ and $B$, one can plot $\ln(I/V^2)$ vs $1/V$, which is known as a FN plot. If the emission is orthodox, then, a FN-plot will be a nearly straight line, where the slope will depend on $B$. This procedure allows the FEF to be measured experimentally, by looking at the slope of a nearly linear FN type plot \cite{richardemission, jordanpaper}.

% Perhaps more technological examples: Field emission displays? Etc. TO CITE!
As for a technological example: a protrusion might be engineered into a flat surface, with the aim of causing field emission at lower external fields. For instance, post-like carbon nanotubes (CNTs) may be used as a single or a large area emitter. That indeed has several applications, such as, the manufacturing of microwave amplifiers, electron microscopy, field emission displays and other applications \cite{application1,application2,application3,application4,techcnt}. There has been several works that use a classical model to estimate the apex FEF of arrays of CNTs \cite{cntfef1,cntfef2,cntfef3}, as well as experimental measurements \cite{cntfefexp1,cntfefexp2,cntfefexp3}.

In an attempt to model an isolated post-like protrusion, like a CNT, several models have been made, treating the protrusion as a classical conducting entity, with a specific geometry \cite{richard}. Some of them will be listed as follows.

The \textit{hemisphere on a conducting plate} model assumes a hemisphere of radius $R$ above the flat plate, and attempts to calculate the apex FEF. This case can be analytically solved from first principles \cite{jackson}, yielding a known $\gamma_a = 3$. Using it to model the FEF of a general emitter creates some problems: the model fails to capture the height of the emitter.

The \textit{hemi-ellipsoid on a conducting plate} model, as shown in figure (1.1a),  is a conducting revolution hemi-ellipsoid of base radius $R$ and height $h$, placed over a flat conducting plate. It can also be solved analytically using spheroidal coordinate system \cite{smyth}, thus $\gamma_a$ is known. Such model is more realistic than the former, because it can have any height, however, the greater the height, the greater the Gaussian curvature at the apex, and thus, the greater the charge density, which in turn, leads to large (and possibly unrealistic) apex FEFs.

The \textit{hemisphere on a cylindrical post} (HCP) model, as shown in figure (1.1b), is a conducting hemisphere of radius $R$, over a cylindrical shape of height $\ell$, over a plate. The overall height of the structure is often called $h = R + \ell$. Such model, in turn, has a smooth curvature at the apex, and can have any height. On the other hand, this model lacks an analytical solution in the literature, though accurate numerical computations have been done \cite{thiago}.

In the case of a CNT, there's no reason to believe that classical models would be effective, as quantum effects might be relevant. However, a recent paper \cite{classicalcnt} have numerically explored the possibility for determining a characteristic FEF of a CNT using density functional theory (DFT) \cite{DFT}, a quantum mechanical method, and compared with the the local FEFs of the HCP model. The authors found that, depending on how the FEF is defined in the quantum case, there exists a relation between the FEF calculated from DFT, and the FEF calculated from the classical HCP model. This highlights the relevance of classical models, which will be the focus of this dissertation.

% New paragraph.
It has been focus of several works to find an expression for the FEF in the context of the HCP model. Some of them includes:
\begin{itemize}
\item $\gamma_a \approx 2 + \nu$  (Eq. (16) of \cite{richard})
\item $\gamma_a \approx \nu$ if $\nu\gg 1$ (Eq. (17) of \cite{richard})
\item $\gamma_a \approx 1.2(2.15 + \nu)^{0.90}$ (Eq. (6) of \cite{Edgcombe2001b}) (Eq. (20) of \cite{richard})
\item $\gamma_a \approx 1.0782(4.7 + \nu)^{0.9152}$ (Eq. (2) of \cite{cntfef2})
\item $\gamma_a \approx 5.93 + 0.73\nu - 0.0001\nu^2$ (Eq. (22) of \cite{richard})
\end{itemize}

In this work, a mathematical method is constructed in order to calculate the potential over an emitter shape with axial symmetry from first principles, and the method is applied for the case of a hemisphere on a plate, a hemiellipsoid on a plate, and a HCP on a plate. The method allows analytical conclusions to be drawn, and these will be explored.

In addition, interacting emitters, either in the form of finite arrays, or simply two of them separated at a distance $c$ from each other, experiences a reduction in the FEF: It falls from $\gamma_a$ to $\gamma_a'$, where the fractional change in apex FEF, $\delta = (\gamma_a' - \gamma_a)/\gamma_a$, was attempted to be estimated by several works in different contexts.

\begin{itemize}
\item $-\delta = \exp\left(-1.1586c\frac{h_0}{h}\right)$, in ref \cite{Bonard2001}, where $h_0 = 2\mu m$. This formula is dimensionally incorrect, as the argument inside the exponencial is not adimensional.
\item $-\delta = \exp\left(-2.3172\frac{c}{h}\right)$, in ref \cite{Jo2003}. This fixes the dimensional error from the previous formula. However, ref \cite{Harris2015} says such formula doesn't provide a good fit for arrays of various geometries.
\item $-\delta = \exp\left(a\left[\frac{c}{h}\right]^k\right)$, in ref \cite{Harris2015}. Parameter $k$ was introduced in order to better fit arrays. Furthermore, \cite{Harris2015} claims there's agreement with line of charge (LCM) model.
\item $-\delta \sim 2(r/\ell)\left(\frac{\ell}{c}\right)^3$, from \cite{Richard2016}. This was obtained by a more analytical approach using the floating sphere at emitter plane potential (FSEPP) model, as opposed to curve fitting (as was done with all the others above it). The striking contrast is, instead of an exponential behavior, a power law was predicted, with exponent $-3$.
\item Ref \cite{Assis2018} numerically estimated $\delta$ for pairs of 6 different emitter geometries, and found the $c^{-3}$ law.
\item Ref \cite{DallAgnol2018} used method of images to find $\delta\sim c^{-3}$.
\item Ref \cite{Biswas2018} used LCM models to find $\delta\sim c^{-3}$.
\item Ref \cite{Forbes2018} argues about universality of $\delta\sim c^{-3}$.
\end{itemize}

This work will show the $c^{-3}$ law is valid for any axially-symmetric emitter, under an independent approximation (charge density doesn't change much as $c$ varies), and under approximation $c/h\gg 1$. Using these approximations, it will be shown that $-\delta = K\cdot (h/c)^{3}$, where $h$ is the height of the emitter. Furthermore, it will be shown the pre-factor $K$ is not 1, and in fact, depends on the aspect ratio, that is, $K = K(\nu)$, where the functional form of $K$ will be given.

This dissertation is organized as follows: Chapter 02 explains the method, its connection with the multipole coefficients, and proves how multipole moments scales linearly with the external macroscopic field. In chapter 03, the method is applied to model shapes in which the apex FEF is known: the hemisphere on a plate model, and the hemi-ellipsoid on a plate model. In chapter 04, the method is applied to the HCP model. It will be shown that even multipoles are zero for the system HCP emitter on a plate, and some consequences of it. Additionally, an approximate first-order expression for the FEF will be calculated for the range $h\ll R$. Chapter 05 will prove some theorems for a general emitter, as long as some conditions are met. Chapter 06 will conclude the dissertation, with a review of the main results obtained, and perspectives of future work. For the sake of transparency, there will be several appendixes, in which some technical derivations are shown in detail.

\begin{figure} \label{hcp_spheroid_figure}
\includegraphics[scale=0.20]{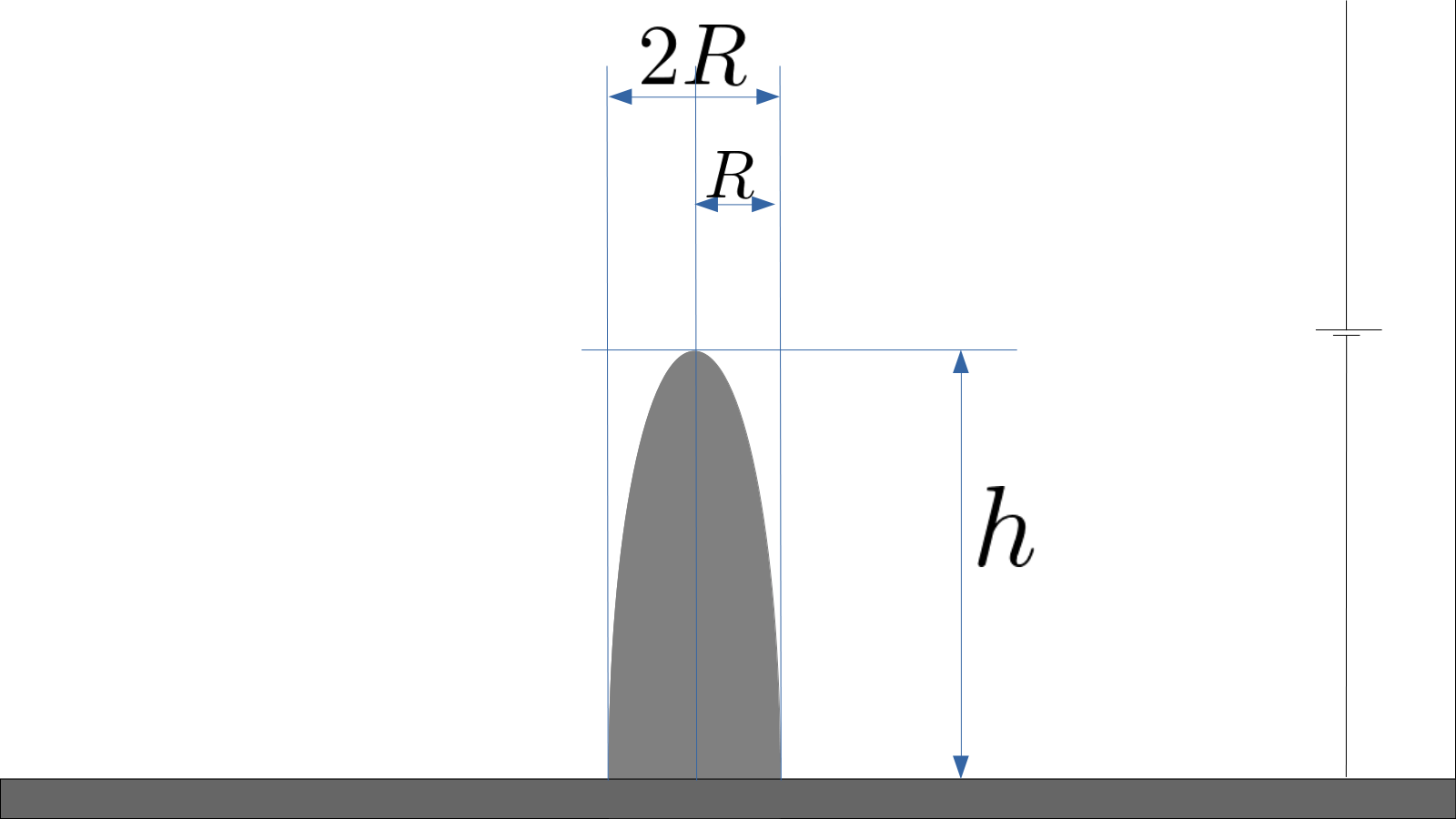}
\includegraphics[scale=0.20]{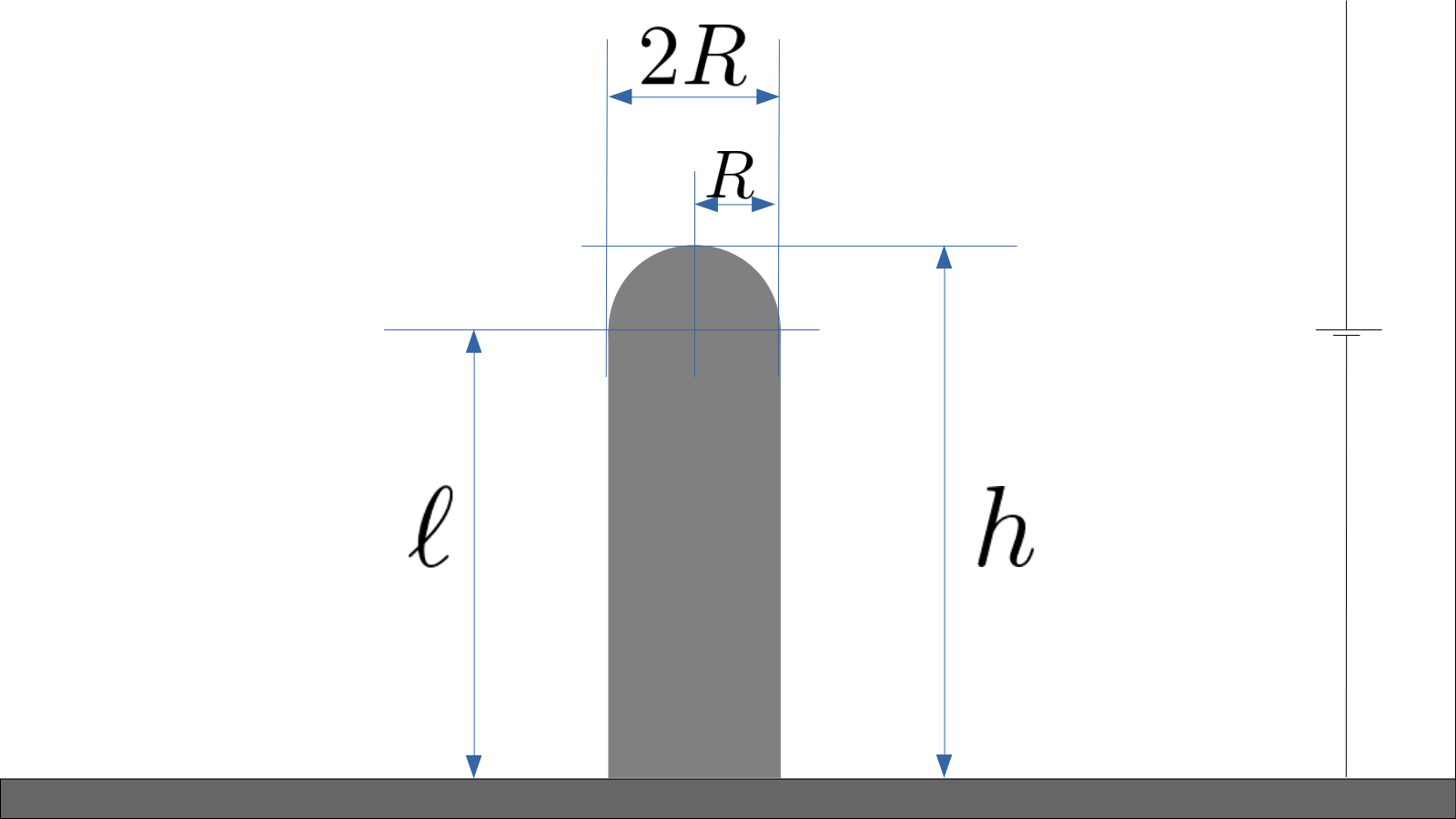}
\caption{(a) Hemiellipsoid over a plate model. (b) Hemisphere on a post (HCP) model.}
\end{figure}

\newpage
\chapter{Method}
In general, it is possible to expand the electrostatic potential by means of spherical harmonics. However, if a system has axial symmetry (such as, the HCP model), the potential can be written \cite{jackson} by means of Legendre polynomials, in the following way:
\begin{equation} \label{potential}
V(r, \theta) = -E_0 r \cos\theta + \sum_{l=0}^\infty A_l r^{-(l+1)} P_l(\cos\theta).
\end{equation}

In (\ref{potential}), $E_0$ is the magnitude of the external field, oriented towards the positive z-axis. In the context of field emission, the applied electrostatic field is often considered to be $F_M$ (called macroscopic field), oriented towards the negative z-axis (that is, $F_M = -E_0$) (this also implies $\mathbf F = -\mathbf E$).

The problem of interest is to find $V(r, \theta)$, where, the emitter coincides with an equipotential surface. In order to find the proper coefficients $A_l$, it is imposed the following boundary condition: $V = 0$ at the surface $S$. Based on this, an error function, to measure how far the solution deviates from the correct one, is defined as the following:
\begin{equation} \label{error_function}
\Sigma(A_l) \overset{\underset{\mathrm{def}}{}}{=} \int_S V(r, \theta)^2 dS.
\end{equation}

Notice that, if $V$ is smooth, $\Sigma(A_l) = 0$ if and only if boundary conditions are met. Otherwise, $\Sigma(A_l) > 0$.

Inserting (\ref{potential}) at (\ref{error_function}), then:
\begin{equation} \label{error_function_inserted}
\Sigma(A_l) = \int_S E_0^2 r^2 \cos^2\theta dS
- \sum_{i=0}^\infty 2E_0 A_i G_i
+ \sum_{i=0}^\infty \sum_{j=0}^\infty A_i A_j I_{ij},
\end{equation}
where the quantities $G_i$ and $I_{ij}$, which will be called $G$-integrals and the $I$-integrals from now on, are defined as:
\begin{equation} \label{IG_integrals}
G_i \overset{\underset{\mathrm{def}}{}}{=} \int_S r^{-i} P_i(\cos\theta)\cos\theta dS,\quad\quad
I_{ij} \overset{\underset{\mathrm{def}}{}}{=} \int_S r^{-i-j-2} P_i(\cos\theta) P_j(\cos\theta) dS
\end{equation}

As for a property visible from the definition, the $I$ integral is symmetric, that is, $I_{ij} = I_{ji}$.

The goal is to find all $A_l$ that minimizes $\Sigma(A_l)$. A necessary condition is:
\begin{equation} \label{error_condition}
\frac{\partial\Sigma}{\partial A_l} = 0, \quad\text{for all $l\in\{0, 1, 2, ...\}$.}
\end{equation}

Imposing (\ref{error_condition}) at (\ref{error_function_inserted}), then:
$$
\frac{\partial\Sigma}{\partial A_l} = 0
\quad\implies\quad
\sum_{k=0}^\infty I_{lk} A_k = E_0 G_l.
$$

The above equation can be written in matrix form:
\begin{equation} \label{linear_system}
\begin{bmatrix}
I_{00} & I_{01} & I_{02} & I_{03} & \cdots \\
I_{10} & I_{11} & I_{12} & I_{13} & \cdots \\
I_{20} & I_{21} & I_{22} & I_{23} & \cdots \\
I_{30} & I_{31} & I_{32} & I_{33} & \cdots \\
\vdots & \vdots & \vdots & \vdots & \ddots
\end{bmatrix}
\begin{bmatrix}
\tilde A_0 \\
\tilde A_1 \\
\tilde A_2 \\
\tilde A_3 \\
\vdots
\end{bmatrix}
=
\begin{bmatrix}
G_0 \\
G_1 \\
G_2 \\
G_3 \\
\vdots
\end{bmatrix}
\end{equation}

Where $\tilde A_l = A_l / E_0$. If the linear system can be solved, then the values of $A_l$ are obtained, and one can find the electrostatic potential (\ref{potential}). However, it is reasonable to assume truncating the system yields approximate solutions for the values of $A_l$. Truncating the $I$-Matrix into a $n\times n$ matrix, and the $G,A$-Vectors into a $n$-dimensional vector, then the now finite linear system can be solved for the $A_l$ values.

\begin{definition}
\textit{I-Matrix, G-Vector, A-Vector, normalized A-Vector} are the truncated matrix filled with the integrals $I_{ij}$, the truncated vector filled with the values $G_l$, $A_l$ and $\tilde A_l$, respectively.
\end{definition}

\section{The method: Truncating the system}
Even if the values $I_{ij}$ and $G_l$ doesn't depend on where the system is truncated (that is, on the value of $n$), however, $A_l$ does. Thus, if the system is truncated having a matrix of dimensions $(n+1)\times (n+1)$, the A-Values will be called as $A_{l}^{(n)}$. As an example, a good exercise is to solve for $n=0$ and $n=1$.

\subsubsection{Truncating at $n=0$}
On that, the matrix is of dimension $1\times 1$:
$$
I_{00} \tilde A_0^{(0)} = G_0 \quad\implies\quad
\tilde A_0^{(0)} = \frac{G_0}{I_{00}}
$$

\subsubsection{Truncating at $n=1$}
On that, the matrix is of dimension $2\times 2$:
$$
\begin{bmatrix}
I_{00} & I_{01} \\
I_{10} & I_{11} \\
\end{bmatrix}
\begin{bmatrix}
\tilde A_0^{(1)} \\
\tilde A_1^{(1)} \\
\end{bmatrix}
=
\begin{bmatrix}
G_0 \\
G_1
\end{bmatrix}
\quad\implies\quad
\begin{bmatrix}
\tilde A_0^{(1)} \\
\tilde A_1^{(1)} \\
\end{bmatrix}
=
\frac{1}{\det I_{2\times 2}}
\begin{bmatrix}
I_{11} & -I_{01} \\
-I_{10} & I_{00} \\
\end{bmatrix}
\begin{bmatrix}
G_0 \\
G_1
\end{bmatrix}
$$

Henceforth:
\begin{equation}
\tilde A_0^{(1)} = \frac{I_{00}G_0 - I_{01} G_1}{I_{00}I_{11} - I_{10}I_{01}},\quad\quad
\tilde A_1^{(1)} = \frac{I_{10}G_0 - I_{00} G_1}{I_{00}I_{11} - I_{10}I_{01}}
\end{equation}

An important comment is that, $A_0^{(0)} \neq A_0^{(1)}$, but, it is possible to write $A_0^{(1)}$ as a function of $A_0^{(0)}$. For instance, using the fact $I_{01} = I_{10}$, then:
$$
A_0^{(1)} = A_0^{(0)}\frac{I_{11} - I_{01}\frac{G_1}{G_0}}{I_{11} - \frac{I_{01}^2}{I_{11}}}
$$

The only difficulty in this method, is solving all $I$-integrals and $G$-integrals. That can be done analytically depending on the shape to be integrated, but also, by means of numerical integration procedures. In increasing the dimension of the linear system by truncating at a high value of $n$, it is expected to get closer to the real potential. Furthermore, this procedure is valid for any shape $S$, as long as it has axial symmetry.

\section{Local field enhancement factor (FEF) at the apex}
Once the $A_l$ values are known, one can use it and build the entire potential as in (\ref{potential}). Having the potential, the local FEF can be calculated.

The Field Enhancement Factor at the apex $\gamma_a$, is defined as the magnitude of the local electric field at the apex of the surface $E_a$, divided by the external field $E_0$.
$$
E_a =
-\frac{\partial V}{\partial r}\Bigr|_{\substack{\theta=0}} =
E_0 + \sum_{l=0}^\infty (l+1) \frac{A_l}{r^{l+2}}
$$

Therefore, the FEF:
\begin{equation} \label{FEF}
\gamma_a =
\frac{E_a}{E_0} =
1 + \sum_{l=0}^\infty \tilde A_l\frac{l+1}{r^{l+2}}
\end{equation}

In equation (\ref{FEF}) it was assumed two things: (1) the surface's tangent plane at the apex coincides with the $xy$-plane, hence why the field only has a z-direction. (2) The apex of the surface is located at the $z$-axis. If both conditions are not met, the full gradient (with radial, polar, and azimuthal terms), must be taken, as opposed to an ordinary derivative in $r$. Thus, the apex field $E_a$ depends on the radial variable $r$, which, in this case, should be how far the apex of the surface is from the origin. In summary, in both formulas above, the apex is located at cartesian coordinates $(0, 0, r)$. For an sphere of radius $R$, then, $r = R$. For a revolution ellipse with radius $R$ and height $h$, then, $r = h$. In case of a HCP surface, $r = \ell+R = h$.

\section{Physical interpretation}
The $A_l$ terms of the Legendre expansion are associated with the $2^l$-multipole moments. Thus, what is really happening, is a multipole expansion. The linear system (\ref{linear_system}) is choosing the multipole coefficients such that it better meets the boundary conditions. As for an example on the multipole $l=0$, the charge monopole can be calculated by means of Gauss Law taken in a sphere $P$ in which $S$ is completely inside the sphere. All one has to do is to derive (\ref{potential}) with respect to $r$, and then:
$$
\frac{Q}{\epsilon_0} =
\int_P \mathbf E\cdot dS =
-\int_0^\pi \int_0^{2\pi} \frac{\partial V}{\partial r} r^2 \sin\theta dr d\theta =
4\pi A_0
$$

Thus, the $A_0$ term is associated with the zeroth-multipole coefficient (a simple image charge), located at the center of the coordinate system. The term $A_1$ is associated with a dipole moment, also located at the center. $A_2$ with a quadrupole moment, and so on.

In fact, an axially symmetric multipole expansion (written below) can be compared with the electrostatic potential (\ref{potential}), and thus, a relation with $A_l$ and $Q_l$ is possible:
\begin{equation} \label{Ql_Al}
V(r, \theta) = \frac{1}{4\pi\epsilon_0}\sum_{l=0}^\infty\frac{Q_l}{r^{l+1}} P_l(\cos\theta)
\quad\implies\quad
Q_l = 4\pi\epsilon_0 A_l
\end{equation}

In which, for an axially symmetric shape:
\begin{equation} \label{multipole_moment}
Q_l = \int_S \sigma(\mathbf r) r^l P_l(\cos\theta) dS(\mathbf r)
\end{equation}

Where $\sigma$ is the charge density on the surface, which, is induced by the presence of the external field $E_0$.

Furthermore, because (\ref{linear_system}), the values $\tilde A_l$ are independent on the external field $E_0$. It was defined that $\tilde A_l = A_l / E_0$, and, the values of $A_l$ are the ones that actually appear in the electrostatic potential (\ref{potential}). Because $A_l$ related to the multipoles $Q_l$ by (\ref{Ql_Al}), then:
\begin{equation} \label{Ql_tilde_Al}
Q_l = 4\pi\epsilon_0 E_0 \tilde A_l
\end{equation}

Where $\tilde A_l$ depends only on the geometry of the system. In other words, the multipole moments scale linearly with $E_0$. If one doubles the external field, all multipole moments will double.

\newpage
\chapter{Application of the method}
\section{Hemisphere over plane}
As the method was just explained, one can test it in a simple shape, such as, a hemisphere with radius $R$ centered at the origin, on top of the conducting plane $z=0$, and to calculate the apex FEF. By symmetry, this is equivalent of considering a sphere, and no plane, as the potential in the whole space won't be changed. Therefore, instead of integrating the hemisphere and the plane (which is infinite in size), the integration will be made on a whole sphere, and no plane, because the area of a sphere is finite, as opposed to a plane. The solution will be equivalent to both problems, of course. In addition, about the HCP model, setting $h=R$, that is, $\nu = h/R = 1$, is equivalent to the sphere problem. Furthermore, a prolate spheroid with $\nu = h/R = 1$ is also equivalent to the sphere problem. Thus, solving the sphere also gives the results for a HCP $\nu=1$ and a spheroid of $\nu = 1$.

\subsection{Spherical G-Integrals}
To solve the problem, the $A_l$ values are required, and, to calculate those, the $G_l$ and $I_{ij}$ values are required. First, the $G$-integrals:
\begin{equation}
G_l = \int_S r^{-l} P_l(\cos\theta)\cos\theta dS =
\int_0^\pi \int_0^{2\pi} R^{-l} P_l(\cos\theta)\cos\theta R^2\sin\theta d\theta d\phi
\end{equation}

If $x=\cos\theta$, and using the fact that $x = P_1(x)$
\begin{equation}
G_l = 2\pi R^{-l+2}\int_{-1}^1 x P_l(x) dx =
2\pi R^{-l+2}\int_{-1}^1 P_1(x) P_l(x) dx
\end{equation}

By orthogonality of the Legendre polynomials, one can finally get the value of the $G$-integrals:
\begin{equation} \label{spherical_Gl_exact}
G_l = \frac{4}{3}\pi R \delta_{l,1}
\end{equation}

Where $\delta_{i,j} = 0$ if $i\neq j$, and $\delta_{i,j} = 1$ if $i=j$.

\subsection{Spherical I-Integrals}
Now, doing the same with the I-Integrals:
\begin{equation}
I_{ij} = \int_S r^{i-j-2} P_i(\cos\theta) P_j(\cos\theta) =
\int_0^\pi \int_0^{2\pi} R^{-i-j-2} P_i(\cos\theta) P_j(\cos\theta) R^2\sin\theta d\theta d\phi
\end{equation}

Again, doing substitution $x=\cos\theta$, and using orthogonality of Legendre polynomials:
\begin{equation} \label{spherical_Iij_exact}
I_{ij} = 2\pi R^{-i-j}\int_{-1}^1 P_i(x) P_j(x) dx = \frac{2\pi}{R^{i+j}}\frac{2}{2j+1}\delta_{i,j}
\end{equation}

\subsection{A-Vector}
The infinite $I$-matrix is diagonal, more, among the $G$-Vector, only $G_1$ is nonzero. Thus, all relevant values are:
\begin{equation} \label{sphere_1_values}
I_{11} = \frac{4\pi}{3 R^2},\quad\quad
G_1 = \frac{4}{3}\pi R
\end{equation}

Therefore, solving the system (\ref{linear_system}), one gets:
\begin{equation}
\sum_{n=0}^\infty I_{ln} \tilde A_n = G_n
\quad\implies\quad
I_{nn}\tilde A_n = G_n
\quad\implies\quad
I_{11} \tilde A_1 = G_1
\end{equation}

Therefore:
\begin{equation}
\tilde A_1 = \frac{G_1}{I_{11}} = R^3, \quad\quad
\tilde A_{n\neq 1} = 0
\end{equation}

\subsection{Solution}
Writing the expression for the potential at (\ref{potential}) considering $A_{n\neq 1} = 0$, then:
\begin{equation} \label{potential_spherical_case}
V(r, \theta) = -E_0 r\cos\theta + \frac{A_1}{r^2}\cos\theta
= E_0\left(r - \frac{R^3}{r^2}\right)\cos\theta
\end{equation}

Which coincides with the well-known analytical result for the sphere in an external field \cite{jackson}. The FEF can be calculated by expression (\ref{FEF}), that is:
\begin{equation} \label{FEF_spherical_case}
\gamma = 1 + \sum_{l=0}^\infty \tilde A_l\frac{l+1}{r^{l+2}}
= 1 + \tilde A_1\frac{1+1}{R^{1+2}}
= 1 + (1+1)\frac{R^3}{R^3}
= 1 + 1 + 1
= 3
\end{equation}

\newpage
\section{Half prolate spheroid over plane}
The method worked with a hemisphere over a plane: next step would be to consider a spheroid. A spheroid is by definition an ellipsoid of revolution. With respect the $z$-axis, the spheroid can be elongated or flattened. An elongated spheroid is called a prolate spheroid, while, a flattened spheroid is called an oblate spheroid.

The intersection of the spheroid across the plane $z=0$ forms a circle of radius $R$. The intersection of the spheroid across the plane $y=0$ (or $x=0$) forms an ellipse of semi-major axis $h$. The height of the spheroid is thus $h$. Given the shape of interest is in the case $h\ge R$, by definition, such shape is a prolate spheroid. The aspect ratio of a prolate spheroid is defined as $\nu = h/R$.

Again, just like it was done with the sphere, the half prolate spheroid over a conducting plane is equivalent to a full prolate spheroid, with no plane. The advantage of the later is clear from the fact that, a plane has infinite surface area, while, a spheroid has not, thus, all integrations will be finite. As before, one should calculate the $G$-vector, the $I$-matrix, then $A$-vector, and then, the local FEF.

\subsection{Spheroidal G-Integrals}
The $G$-integral:
\begin{equation} \label{general_Gl_sph}
G_l = \int_S r^{-l} P_l(\cos\theta)\cos\theta dS
\end{equation}

The $G$ integral was calculated for a sphere, and thus, to avoid confusion, a different notation is employed: from now on, $\dot G_l$ (with the dot on top), will be denoted to exclusively speak about the $G$ integral on a spheroid. In fact, every dotted quantity will be speaking about a spheroid.

To solve the integrals, two quantities are required: $r$ and $dS$, that is, the values of $r(\theta, \phi)$ which coincides with a prolate spheroid, and, the area element $J$, such that $dS = Jd\theta d\phi$. A detailed derivation of both can be found in Appendix A, where the final results are in equations (\ref{spheroidal_r}) and (\ref{spheroidal_area_element}), written below:

\begin{equation} \label{spheroidal_rJ_data}
r = \frac{R}{\sqrt{1 - \epsilon^2\cos^2\theta}},
\quad\quad
\epsilon^2 = 1 - \frac{R^2}{h^2},
\quad\quad
J = r^2\sin\theta\sqrt{1 + \frac{r^4}{R^4}\epsilon^4\cos^2\theta\sin^2\theta}
\end{equation}

One can identify $\epsilon$ as the first eccentricity of the ellipse which generated the spheroid. Therefore $0 \le\epsilon < 1$, and, if $\epsilon = 0$, then, the spheroid collapses into a sphere, as $R = h$.

Inserting equation (\ref{spheroidal_rJ_data}) into equation (\ref{general_Gl_sph}), one can get (\ref{spheroidal_Gl_thetaphi}) [located in Appendix B], written below:
\begin{equation} \label{spheroidal_Gl_thetaphi}
\dot G_l = \int_0^{2\pi}\int_0^\pi r^{-l} P_l(\cos\theta) \cos\theta\cdot r^2\sin\theta \cdot \sqrt{1 + \frac{r^4}{R^4}\epsilon^4\cos^2\theta\sin^2\theta} d\theta d\phi
\end{equation}

Therefore, integrating at $\phi$:
\begin{equation}
\dot G_l = 2\pi \int_0^\pi r^{-l+2} P_l(\cos\theta) \cos\theta\sin\theta\sqrt{1 + \frac{r^4}{R^4}\epsilon^4\cos^2\theta\sin^2\theta} d\theta
\end{equation}

Like before, substitution $x=\cos\theta$ is done, thus arriving at equation (\ref{spheroidal_G}) [located in Appendix B], written below:
\begin{equation}
\dot G_{2l+1} = \frac{2\pi}{R^{2l-1}}\int_{-1}^1 x P_{2l+1}(x)\left[1 - \epsilon^2 x^2\right]^{l-1/2}\cdot\sqrt{1 + \frac{\epsilon^4 x^2 (1-x^2)}{(1 - \epsilon^2 x^2)^2}} dx
\end{equation}

Where $G_{2l} = 0$. Only terms of the form $2l+1$ contributes with the $G$-Vector. All others are zero. Integral above can be solved exactly, and it is done in Appendix B, as shown in equation (\ref{exact_spheroidal_Gl}), written below for convenience:
\begin{equation}
\dot G_{2l+1} = \frac{2\pi}{R^{2l-1}}\sum_{n=0}^\infty \sum_{k=0}^{l} (-1)^n a_{2l+1,2k+1}\binom{l-\frac{1}{2}}{n}\epsilon^{2n} A_{2(k+n+1)}
\end{equation}

Where $A_{2n}$ (not to be confused with the coefficients $A_l$ in the Legendre expansion of the potential) is the $2n$-th moment of the square root part of the integrand, that is:
\begin{equation} \label{A2n_moment_integrals}
A_{2n} = \int_{-1}^1 x^{2n}\sqrt{1 + \frac{\epsilon^4 x^2 (1-x^2)}{(1 - \epsilon^2 x^2)^2}} dx
\end{equation}

The moments $A_{2n}$ can be solved exactly using hypergeometric functions. Also, $a_{l,k}$ are the coefficients of the Legendre polynomials, that is, $P_l(x) = \sum_k a_{l,k} x^k$. And $\binom{\cdot}{\cdot}$ is the generalized binomial coefficient as defined in (\ref{generalized_binomial}) [located at Appendix B]. The level of sophistication is unnecessary, as such moments can be approximated by Taylor expansion. The calculation is done in Appendix B, and the final result can be found in (\ref{approx_spheroidal_Gl}).

Finally, thus, the expression of the $G$ integrals for a prolate spheroid, are:
\begin{equation}
\dot G_{2l+1} \approx \frac{2\pi}{R^{2l-1}}\sum_{n=0}^\infty (-1)^n \binom{l-\frac{1}{2}}{n}\epsilon^{2n}
\sum_{k=0}^{l} a_{2l+1,2k+1}\left[\frac{1}{n+k+1+\frac{1}{2}} + \frac{\epsilon^4}{2}\frac{1}{n+k+1+\frac{3}{2}}\right]
\end{equation}

From it, inserting $l=0$, one can get $G_1$, by using $a_{1,1} = 1$ because $P_1(x) = x$.
\begin{equation}
\dot G_1 = 2\pi R\sum_{n=0}^\infty (-1)^n \binom{l-\frac{1}{2}}{n}\epsilon^{2n}
a_{1,1}\left[\frac{1}{n+1+1+\frac{1}{2}} + \frac{\epsilon^4}{2}\frac{1}{n+1+1+\frac{3}{2}}\right]
\end{equation}

Furthermore, one can further approximate $G_l$, for instance, if the prolate spheroid is not elongated much, that is, if the $\epsilon\approx 0$, then, higher powers of $\epsilon$ becomes closer to zero, allowing the possibility of truncation of the summation as a tool for approximation. For powers of $\epsilon^0$, that is, $n=0$ truncation, yields:
\begin{equation}
\dot G_{2l+1} \approx \frac{2\pi}{R^{2l-1}}\sum_{k=0}^{l} a_{2l+1,2k+1}\left(l-\frac{1}{2}\right)\frac{1}{k+1+\frac{1}{2}}
\end{equation}

Truncating at powers of $\epsilon^2$:
\begin{equation}
\dot G_{2l+1} \approx \frac{2\pi}{R^{2l-1}}\sum_{k=0}^{l} a_{2l+1,2k+1}
\left[
\frac{1}{k+1+\frac{1}{2}} - \epsilon^2\left(l - \frac{1}{2}\right)
\frac{1}{k+2+\frac{1}{2}}
\right]
\end{equation}

An example:
\begin{equation}
\dot G_1 = 2\pi R a_{1,1}
\left[
\frac{1}{0+1+\frac{1}{2}} - \epsilon^2\left(0 - \frac{1}{2}\right)
\frac{1}{0+2+\frac{1}{2}}
\right]
=
2\pi R\left[\frac{2}{3} + \frac{\epsilon^2}{5}\right]
\end{equation}

Which can be re-written:
\begin{equation} \label{spheroidal_G1_value}
\dot G_1 = \frac{4}{3}\pi R\left[1 + \frac{3}{10}\epsilon^2 + O(\epsilon^4)\right]
\end{equation}

Comparing with (\ref{sphere_1_values}), the $G_1$ found is exactly the sphere's $G_1$, corrected by an eccentricity. Such expression is suitable for small $\epsilon$.

\subsection{I-Integrals}
The same way that was done with the $G$-integrals, it will be done with the $I$-integrals. Inserting the $r(\theta)$ and $J$, and integrating on $\phi$, equation (\ref{spheroidal_Iij_theta}) [located at Appendix B] is obtained, written below:
\begin{equation}
\dot I_{ij} = 2\pi \int_0^\pi r^{-i-j} P_i(\cos\theta) P_j(\cos\theta)\sin\theta\sqrt{1 + \frac{r^4}{R^4}\epsilon^4\cos^2\theta\sin^2\theta} \cdot d\theta
\end{equation}

And then doing substitution $x=\cos\theta$, one arrives at expression (\ref{spheroidal_Iij_x}), written below:
\begin{equation}
\dot I_{ij} = \frac{2\pi}{R^{i+j}} \int_{-1}^1 P_i(x) P_j(x) \left[1 - \epsilon^2 x^2\right]^{\frac{i+j}{2}} \sqrt{1 + \frac{\epsilon^4 x^2(1-x^2)}{(1 - \epsilon^2 x^2)^2}} dx
\end{equation}

As with the $G$-integrals, there are parity considerations to make: $I_{ij} = 0$ iff $i+j$ is odd. And, of course, $I_{ij}\neq 0$ iff $i+j$ is even. The integral above has an analytical solution as shown by equation (\ref{exact_spheroidal_Iij}), shown below:
\begin{equation}
\dot I_{ij} = \frac{2\pi}{R^{i+j}}
\sum_{u=0}^i \sum_{v=0}^j \sum_{n=0}^{\frac{i+j}{2}} (-1)^n a_{i,u} a_{j,v} \binom{\frac{i+j}{2}}{n}\epsilon^{2n}
A_{u+v+2n}
\end{equation}

Using (\ref{spheroidal_A2n_approx}) for the purpose of approximations, one gets (\ref{approx_spheroidal_Iij}), shown below:
\begin{equation}
\dot I_{ij} = \frac{2\pi}{R^{i+j}}
\sum_{u=0}^i \sum_{v=0}^j \sum_{n=0}^{\frac{i+j}{2}} (-1)^n a_{i,u} a_{j,v} \binom{\frac{i+j}{2}}{n}\epsilon^{2n}
\left[
\frac{1}{n + \frac{u+v}{2} + \frac{1}{2}} +
\frac{\epsilon^4}{2}\frac{1}{n + \frac{i+v}{2} + \frac{3}{2}}
\right]
\end{equation}

With above expression, the value of $\dot I_{11}$:
\begin{equation}
\dot I_{11} = \frac{2\pi}{R^2}\sum_{n=0}^1 \binom{1}{n}\epsilon^{2n}
\left[
\frac{1}{n + \frac{1+1}{2} + \frac{1}{2}} +
\frac{\epsilon^4}{2}\frac{1}{n + \frac{1+1}{2} + \frac{3}{2}}
\right]
\end{equation}

Re-writing:
\begin{equation}
\dot I_{11} = \frac{2\pi}{R^2}
\left[
\frac{2}{3} + \frac{2}{5} \epsilon^2 + \frac{1}{5}\epsilon^4 + \frac{1}{7}\epsilon^6
\right]
\end{equation}

Which, neglecting some terms:
\begin{equation} \label{spheroidal_I11_value}
\dot I_{11} = \frac{4\pi}{3R^2}
\left[
1 + \frac{3}{5} \epsilon^2 + O(\epsilon^4)
\right]
\end{equation}

Which, again, as comparing with (\ref{sphere_1_values}), gives exactly the same for a sphere, except, a correction for the eccentricity.

\subsection{A-Vector}
Having expressions for $G_l$ and $I_{ij}$, the values of $A_l$ can be found by means of (\ref{linear_system}). Again, $G_l = 0$ iff $l$ is even, and $I_{ij} = 0$ iff $i+j$ is odd. Therefore, the linear system (\ref{linear_system}) becomes:
\begin{equation} \label{spheroidal_system}
\begin{bmatrix}
I_{00} &   0    & I_{02} &   0    & I_{04} &   0    & I_{06} &    0   & \cdots \\
0    & I_{11} &   0    & I_{13} &    0   & I_{15} &    0   & I_{17} & \cdots \\
I_{20} &   0    & I_{22} &   0    & I_{24} &   0    & I_{26} &    0   & \cdots \\
0    & I_{31} &   0    & I_{33} &    0   & I_{35} &    0   & I_{37} & \cdots \\
I_{40} &   0    & I_{42} &   0    & I_{44} &   0    & I_{46} &    0   & \cdots \\
0    & I_{51} &   0    & I_{53} &    0   & I_{55} &    0   & I_{57} & \cdots \\
I_{60} &   0    & I_{62} &   0    & I_{64} &   0    & I_{66} &    0   & \cdots \\
0    & I_{71} &   0    & I_{73} &    0   & I_{75} &    0   & I_{77} & \cdots \\
\vdots & \vdots & \vdots & \vdots & \vdots & \vdots & \vdots & \vdots & \ddots
\end{bmatrix}
\begin{bmatrix}
\tilde A_0 \\
\tilde A_1 \\
\tilde A_2 \\
\tilde A_3 \\
\tilde A_4 \\
\tilde A_5 \\
\tilde A_6 \\
\tilde A_7 \\
\vdots
\end{bmatrix}
=
\begin{bmatrix}
0 \\
G_1 \\
0 \\
G_3 \\
0 \\
G_5 \\
0 \\
G_7 \\
\vdots
\end{bmatrix}
\end{equation}

The main linear system above is separated into two different linear systems, completely equivalent to the first: a system for odd components, and another for even components, as written below:
\begin{equation} \label{spheroidal_linear_system_zero}
\begin{bmatrix}
I_{00} & I_{02} & I_{04} & I_{06} & \cdots \\
I_{20} & I_{22} & I_{24} & I_{26} & \cdots \\
I_{40} & I_{42} & I_{44} & I_{46} & \cdots \\
I_{60} & I_{62} & I_{64} & I_{66} & \cdots \\
\vdots & \vdots & \vdots & \vdots & \ddots
\end{bmatrix}
\begin{bmatrix}
\tilde A_0 \\
\tilde A_2 \\
\tilde A_4 \\
\tilde A_6 \\
\vdots
\end{bmatrix}
=
\begin{bmatrix}
0 \\
0 \\
0 \\
0 \\
\vdots
\end{bmatrix} \\
\end{equation}

\begin{equation} \label{spheroidal_linear_system_nonzero}
\begin{bmatrix}
I_{11} & I_{13} & I_{15} & I_{17} & \cdots \\
I_{31} & I_{33} & I_{35} & I_{37} & \cdots \\
I_{51} & I_{53} & I_{55} & I_{57} & \cdots \\
I_{71} & I_{73} & I_{75} & I_{77} & \cdots \\
\vdots & \vdots & \vdots & \vdots & \ddots
\end{bmatrix}
\begin{bmatrix}
\tilde A_1 \\
\tilde A_3 \\
\tilde A_5 \\
\tilde A_7 \\
\vdots
\end{bmatrix}
=
\begin{bmatrix}
G_1 \\
G_3 \\
G_5 \\
G_7 \\
\vdots
\end{bmatrix} \\
\end{equation}

Notice the system (\ref{spheroidal_linear_system_zero}) is a linear system of the form $IA = 0$. If $I$ is non invertible (there is, if $\det I = 0$), then there exists non-null solution $A$. Furthermore, if $A$ is a solution, then $cA$ is also a solution, for any scalar $c\in\mathbb{R}$. Thus, there exists infinite solutions, all belonging to the kernel of $I$. In such situation, even valued $A$ wouldn't be unique, violating the uniqueness theorem of the Laplace Equation. Therefore, such possibility is discarded, even if a formal proof is not offered about if matrix $I$ is indeed singular. Therefore, $A_{2l} = 0, l\in\{0, 1, 2, 3, 4, 5, ...\}$.

Becaused $A_{2l} = 0$, only $A_{2l+1}$ will be relevant, and they will be given by the system on (\ref{spheroidal_linear_system_nonzero}). This means, only odd valued multipoles will contribute to the potential field: dipole, octopole, etc. Even valued multipoles will give a zero contribution (that is, image charge, quadrupole moments, etc).

\subsection{Local FEF}
All that is missing is to find the values of the $A_l$. Already knowing $A_0 = 0$, and truncating (\ref{spheroidal_linear_system_nonzero}) for a one dimensional system. Therefore:
\begin{equation}
\tilde A_1^{(1)} = \frac{G_1}{I_{11}}
\end{equation}

Substituting (\ref{spheroidal_G1_value}) and (\ref{spheroidal_I11_value}) at above equation, one gets:
\begin{equation} \label{spheroidal_A1_value}
\tilde A_1^{(1)} = R^3\frac{1 + \frac{3}{10}\epsilon^2}{1 + \frac{3}{5}\epsilon^2}
\end{equation}

Inserting $A_1^{(1)}$ at (\ref{FEF}), and because $A_l^{(1)} = 0$ for $l > 1$, then:
\begin{equation} \label{spheroidal_gamma_a_1_value}
\gamma_a^{(1)} =
1 + \sum_{l=0}^\infty \tilde A_l^{(1)}\frac{l+1}{r^{l+2}}
= 1 + \tilde A_1^{(1)}\frac{1+1}{r^{1+2}}
= 1 + \tilde A_1^{(1)}\frac{2}{r^3}
\end{equation}

Knowing the value of $A_1^{(1)}$, only $r$ is required, which is the radial position at the apex of the shape, that is, $r=h$. Therefore:
\begin{equation} \label{FEF_prolate_spheroidal_case}
\gamma_a^{(1)} = 1 + 2\frac{R^3}{h^3}\frac{1 + \frac{3}{10}\epsilon^2}{1 + \frac{3}{5}\epsilon^2}
\end{equation}

One can see, above expression only depends on the aspect ratio $\nu$, because $\epsilon$ can be written in function of $\nu$.

It is important to realize at this point all approximations used: An infinite linear system was truncated to one dimension, such that, only dipole contributions to the local FEF are being computed. More, the eccentricity terms on numerator and denominator were truncated until $\epsilon^2$ order, and, the values of $A_{2n}$ where approximated by a second order Taylor expansion. All approximations used are suitable if $\epsilon\ll 1$, which is equivalent to $R\approx h$, which is equivalent to $\nu\approx 1$.

\subsection{Numerical calculations}
The method was used to numerically calculate $\gamma_a^{(n)}$ (see Fig. (3.1a)) for a prolate spheroid of height and radius $h$ and $R$, respectively. It was found the error from the analytical apex FEF $\gamma_a$ was able to be curved-fitted into an exponential function (see Fig (3.1b)):
$$
\frac{\gamma_a - \gamma_a^{(n)}}{\gamma_a} = A\exp\left(-b n\right), \quad b>0
$$

Here, $n$ is the truncation order, and $\gamma_a^{(n)}$ is the local FEF calculated at order $n$. Parameter $A$ should not be confused with the coefficients $A_l$. The $\gamma_a$ used here, is the exact theoretical value, which can be found in equation (26) of \cite{richard}, written below:
\begin{equation}
\nu = \frac{h}{R},\quad\quad
\xi = \sqrt{\nu^2 - 1},\quad\quad
\gamma_a = \frac{\xi^3}{\nu\ln\left(\nu + \xi\right) - \xi}
\end{equation}

At the table below, it is shown for various aspect ratios the fitting parameters $A$ and $b$. For instance, when $\nu=2$ then $A\approx 0.723$ and $b = 0.05086$. Using $1/b \approx 19.658$, it is concluded that, by increasing the order by roughly $20$, the error decreases by $1/e$. Assuming such behavior continues forward, one can calculate the order required to numerically calculate the FEF given some precision.

\begin{tabular}{c|ccc}
Shape & $A$ & $b$ & $1/b$ \\
\hline
$h=1.1$ & $0.332$ & $0.821$ & $1.218$ \\
$h=1.3$ & $0.499$ & $0.355$ & $2.8169$ \\
$h=1.5$ & $0.598$ & $0.1896$ & $5.2742$ \\
$h=1.7$ & $0.656$ & $0.1080$ & $9.2592$ \\
$h=2$ & $0.723$ & $0.05086$ & $19.652$ \\
$h=3$ & $0.851$ & $0.006964$ & $143.584$ \\
$h=5$ & $0.937$ & $0.0006493$ & $1539.93$ \\
\end{tabular}

\begin{figure}
\includegraphics[scale=0.32]{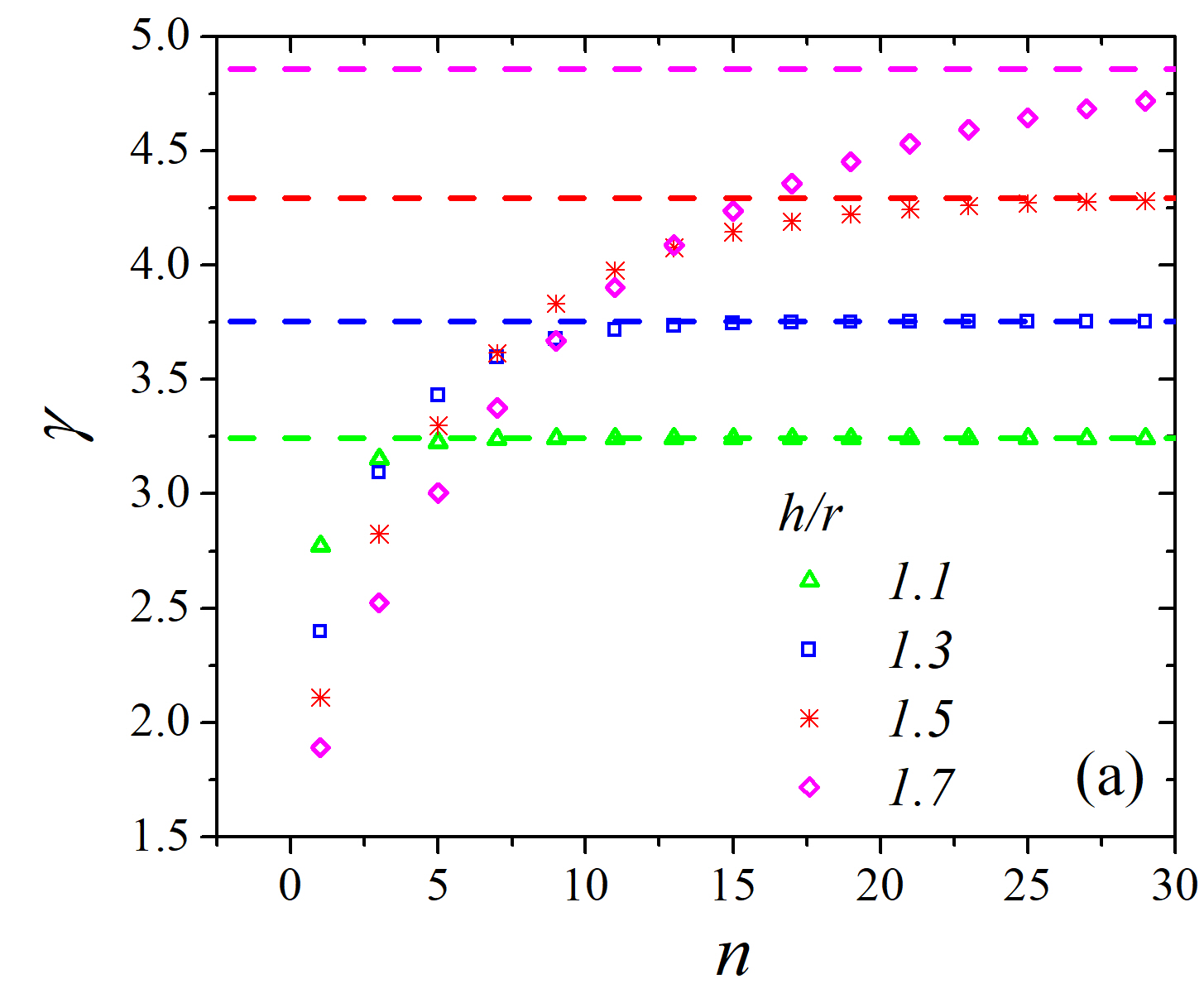}
\includegraphics[scale=0.32]{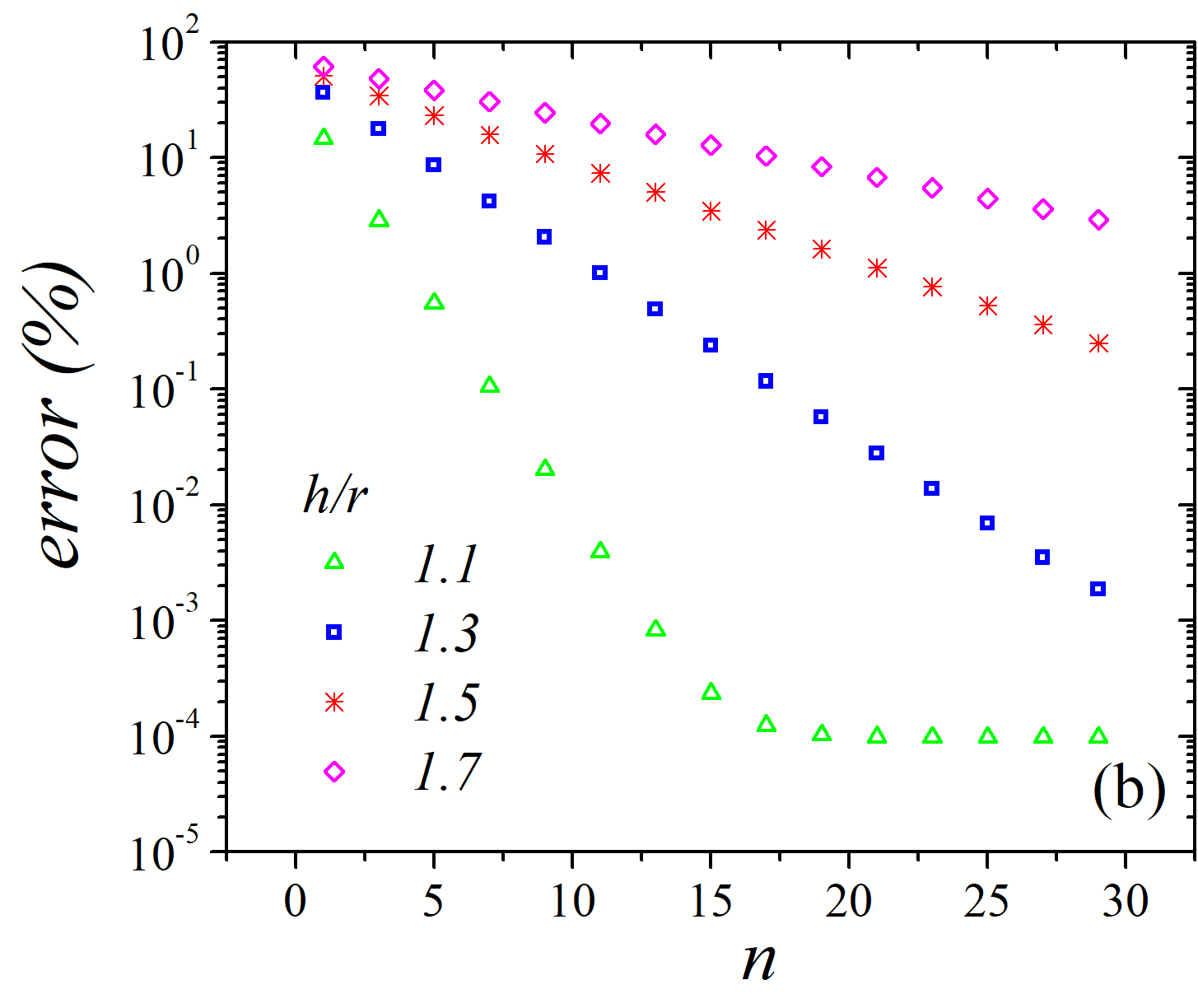}
\caption{(a) $\gamma_a^{(n)}$ vs truncation order $n$ for various aspect ratios. The horizontal dashed lines correspond to the analytical values of the apex FEF $\gamma_a$. (b) Error $\frac{\gamma_a - \gamma_a^{(n)}}{\gamma_a}\times 100(\%)$ vs $n$, for the same aspect ratios as shown in (a). The horizontal plateau observed for $\nu=1.1$ might be explained by the fact that the numerical integration of $A_{2n}$ (see Eq. (\ref{A2n_moment_integrals})) was made with $10^6$ trapezoidal elements.}
\end{figure}

\newpage
\subsection{Slow convergence of the FEF}
As Fig. (3.1) shows (and the values of $1/b$), that, a slow convergence of the FEF is observed as the aspect ratio $\nu = h/R$ increases. The reason for this, the hemi-ellipsoid concentrates its charge distribution on the top. Therefore, it can be approximated as a charge on the top, and its respective image charge (opposite in sign) below the plane. It was shown in (\ref{Ql_Al}) that the coefficients $\tilde A_l$ are proportional to the multipole coefficients $Q_l$ of the system.

For a single charge $q$ located in the $z$-axis, above the plane at a distance $a$, that is, $(0, 0, a)$, then, the axial multipole moments are known to be $Q_l = q a^l$. Because of the image charge below, the total multipole is:
\begin{equation}
Q_l = q a^l + (-q)\cdot (-a)^l = q a^l + (-1)^{l+1} q a^l
\end{equation}

In other words:
\begin{equation} \label{Ql_estimation_multipole}
\begin{split}
Q_{2l} &= 0 \\
Q_{2l+1} &= q a^{2l+1}
\end{split}
\end{equation}

It is clear that $Q_l$ grows with the power $l$. The potential for a multipole $Q_l$ is proportional to $1/r^{l+1}$, and the field as $1/r^{l+2}$. At the apex, $r=h$, and $a < h$ but $a\approx h$, therefore, for sake of illustration, using expression (\ref{FEF}), considering in a rough way that $Q_l\approx aq^l$ for all $l$, then:
\begin{equation} \label{FEF_multipole_approximation}
\gamma_a \approx 1 + \sum_{n=0}^\infty (l+1)\frac{qa^l}{h^{l+2}}
=1 + \frac{1}{h^2}\sum_{n=0}^\infty (l+1)q\left(\frac{a}{h}\right)^l
\end{equation}

Because $a/h\approx 1$ (and $a/h < 1$), the convergence is slow. Such result also predicts that, for some small fixed $n$, then, $\gamma_a^{(n)}\to 1$ as $h$ increases. In general, the larger $h$ is, the more orders are required to correctly approximate $\gamma_a$.

\newpage
\chapter{Hemisphere on a post model}
Having shown how the method works for shapes with known FEFs, now, the HCP model will be focused, analytically, using the former method. The HCP model consists of a conducting cylinder of radius $R$ and height $\ell$ on top of a conducting plane located at $z=0$. On top of the cylinder, there's a hemisphere of radius $R$. It would be required to integrate the hemisphere, the cylinder, and the plane (with a circular hole at the center), however, that won't be necessary: An equivalent problem is posed by eliminating the plane, and considering an mirror shape down below, as if the plane is a mirror. Thus, now there exists a cylinder of radius $R$, but with height $2\ell$, with $\ell$ of height on the positive z axis, and $\ell$ of height on the negative z-axis, and two hemispheres: one located on top of the cylinder, and the other at the bottom of the cylinder. These are easier to integrate, because their area is finite. This will define a HCP shape of height $h = \ell + R$, and height to radius ratio of $\nu = h/R$.

\section{Aspect ratio $\nu$}
The aspect ratio $\nu$ is defined as the ratio between the height and radius of the HCP shape. That is:
\begin{equation}
\nu = \frac{h}{R}
\end{equation}

% It will be shown in this section, that, any expression for the electrostatic potential (and thus the FEF) must depend on $\nu$, or any variable which is a function of $\nu$.
It will be motivated in this section, that, it makes sense to have the FEF depending only on $\nu$.

Assume two HCP shapes denoted by $(h_1, R_1)$ and $(h_2, R_2)$, such that, both shapes have an equal aspect ratio $\nu$. This implies:
\begin{equation}
\nu_1 = \frac{h_1}{R_1} = \frac{h_2}{R_2} = \nu_2
\end{equation}

The only way to meet such criteria, is having: $h_2 = \alpha h_1$, and $R_2 = \alpha R_1$, for some $\alpha\in\mathbb{R}$ where $\alpha > 0$. Thus, this represents an scale by $\alpha$. Therefore, a way to transform from shape $(h_1, R_1)$ to shape $(h_2, R_2)$ where $\nu(h_1, R_1) = \nu(h_2, R_2)$, is by scaling the entire coordinate system $(x, y, z) \rightarrow (\alpha x, \alpha y, \alpha z)$.

One can verify what happens to the Laplacian operator applied over a generic function $\phi$  when such scale transformation is applied. By defining $\mathbf r' = \alpha\mathbf r$, then:
\begin{equation}
\frac{\partial\phi}{\partial x'}
= \frac{\partial x}{\partial x'}\frac{\partial\phi}{\partial x}
= \frac{1}{\alpha}\frac{\partial\phi}{\partial x}
\end{equation}

And:
\begin{equation}
\frac{\partial^2 \phi}{\partial x'^2}
= \frac{\partial}{\partial x'}\frac{1}{\alpha}\frac{\partial\phi}{\partial x}
= \frac{1}{\alpha}\frac{\partial x}{\partial x'}\frac{\partial^2\phi}{\partial x^2}
= \frac{1}{\alpha^2}\frac{\partial\phi}{\partial x^2}
\end{equation}

The same can be done for $y, z$, and, therefore, $\nabla^2\phi = \alpha^2\nabla'\phi$. Thus, if Laplace equation is satisfied:
\begin{equation}
\nabla^2\phi = 0
\quad\iff\quad
\nabla'^2\phi = 0
\end{equation}

Therefore, if $\phi(x)$ is a solution of Laplace equation, so is $\phi(\alpha x)$. If $\psi(\mathbf r)$ be the analytical solution of the HCP shape $(h, R)$, that is, $\psi$ satisfies Laplace equation, and boundary conditions $(h, R)$, then, $\psi(\alpha\mathbf r)$ is a solution of Laplace equation, and satisfies boundary conditions $(\alpha h, \alpha R)$.

If $\psi$ depends only on $\nu$ and the position vector $\mathbf r$, as opposed to $h$ and $R$ separately, then it is clear that the scaling conditions are immediately satisfied. This highlights the importance of the aspect ratio $\nu$. It is expected that all quantities ($G$-integrals, $I$-integrals, $A$-values, and the apex FEF) will depend explicitly on the aspect ratio $\nu$.

\section{Integrals}
\subsection{Cylindrical I-Integrals}
The integral to solve:
\begin{equation}
I_{ij} = \int_S r^{-i-j-2} P_i(\cos\theta) P_j(\cos\theta) dS
\end{equation}

Where $S$ is a cylinder with height $2\ell$ and radius $R$. Recall that $r$ is a radial vector, as in spherical coordinate system. For a cylinder, from Appendix A, it satisfies equation:
\begin{equation}
r = \frac{R}{\sin\theta},\quad\quad J = r^2
\end{equation}

Where $J$ is the area element, that is, $dS = Jd\theta d\phi$. The limits of integration are from $\theta_0$ (the top of the cylinder) to $\pi - \theta_0$, the bottom of the cylinder, where:
\begin{equation} \label{cylindrical_ac_as}
a_c = \cos\theta_0 =
\frac{\ell}{r_0},\quad\quad
a_s = \sin\theta_0 =
\frac{R}{r_0},\quad\quad
r_0 = \sqrt{R^2 + \ell^2}
\end{equation}

A detailed derivation is in Appendix A. The integral is thus:
\begin{equation}
\ddot I_{ij} = \int_S r^{-i-j-2} P_i(\cos\theta) P_j(\cos\theta) dS
= \int_0^{2\pi}\int_{\theta_0}^{\pi - \theta_0} r^{-i-j-2} P_i(\cos\theta) P_j(\cos\theta) r^2 d\theta d\phi
\end{equation}

Thus:
\begin{equation}
\ddot I_{ij} = \frac{2\pi}{R^{i+j}} \int_{\theta_0}^{\pi - \theta_0} P_i(\cos\theta) P_j(\cos\theta) \left[\sin\theta\right]^{i+j} d\theta
\end{equation}

Or, changing $x = \cos\theta$:
\begin{equation} \label{cylindrical_Iij}
\ddot I_{ij} = \frac{2\pi}{R^{i+j}} \int_{-a_c}^{a_c} P_i(x) P_j(x) \left[\sqrt{1 - x^2}\right]^{i+j-1} dx
\end{equation}

First, the function $\sqrt{1 - x^2}$ is always even. Second, if $i+j$ is even, then $P_i P_j$ is even. Also, if $i+j$ is odd, then $P_i P_j$ is odd. Because the limits of the integral are symmetrical, then an odd function must be evaluated to zero. Therefore, $I_{ij} = 0$ if $i+j$ is odd.

Now, if one wishes to calculate specific values:
\begin{equation}
\begin{split}
\ddot I_{00} &= 2\pi\int_{\theta_0}^{\pi-\theta_0} d\theta = 2\pi\left(\pi - 2\theta_0\right) \\
\ddot I_{11} &= \frac{2\pi}{R^2}\int_{\theta_0}^{\pi-\theta_0} \cos^2\theta\sin^2\theta d\theta
\end{split}
\end{equation}

The integrals in $\theta$ where used for $I_{00}$ and $I_{11}$, because they are easier to solve than the integrals in the variable $x$.
\begin{equation}
\int\left[\sin\theta\cos\theta\right]^2d\theta
=\frac{1}{4}\int\sin^2\left(2\theta\right)d\theta
=\frac{1}{8}\int\left[1 - \cos\left(4\theta\right)\right] d\theta
\end{equation}

Thus:
\begin{equation}
\begin{split}
\int\sin^2\theta\cos^2\theta d\theta
&= \frac{1}{8}\theta - \frac{1}{32}\sin\left(4\theta\right) \\
&= \frac{1}{8}\theta - \frac{1}{16}\sin\left(2\theta\right)\cos\left(2\theta\right) \\
&= \frac{1}{8}\theta - \frac{1}{8}\sin\theta\cos\theta\left(\cos^2\theta - \sin^2\theta\right)
\end{split}
\end{equation}

Or, defining $r_0^2 = \ell^2 + R^2$, then:
\begin{equation}
\ddot I_{11} = \frac{2\pi}{R^2}\left[
\frac{\pi}{8}
-\frac{1}{4}\arccos\left(\frac{\ell}{r_0}\right)
+\frac{1}{4}\frac{R\ell}{r_0^2}\left(\frac{\ell^2}{r_0^2} - \frac{R^2}{r_0^2}\right)
\right]
\end{equation}

Or, re-writing both:
\begin{equation}\label{cylindrical_I0011_exact}
\begin{split}
\ddot I_{00} &= 2\pi\left[\pi - 2\arccos\left(\frac{\ell}{r_0}\right)\right] \\
\ddot I_{11} &= \frac{2\pi}{R^2}\left[\frac{\pi}{8}
-\frac{1}{4}\arccos\left(\frac{\ell}{r_0}\right)
+\frac{1}{4}\frac{R\ell}{\ell^2 + R^2}\frac{\ell^2 - R^2}{\ell^2 + R^2}\right]
\end{split}
\end{equation}

If $\ell = 0$, then $\ddot I_{00} = \ddot I_{11} = 0$, because, $\arccos(0) = \pi/2$. Expansion of $\arccos$ is:
\begin{equation}
\arccos(x) = \frac{\pi}{2} - x - \frac{x^3}{6} - \frac{3x^5}{40} + O(x^7)
\end{equation}

Thus, inserting the expansion into (\ref{cylindrical_I0011_exact}), it becomes, for $\ell\ll R$:
\begin{equation}\label{cylindrical_I0011_ellllR}
\begin{split}
I_{00} &= \frac{4\pi\ell}{r_0} \\
\ddot I_{11} &=
\frac{\pi}{2R^2}\left[
\left(\frac{\ell}{r_0}\right)
+\frac{R\ell}{\ell^2 + R^2}\frac{\ell^2 - R^2}{\ell^2 + R^2}\right]
\end{split}
\end{equation}

\subsection{Cylindrical G-Integrals}
The integral to solve:
\begin{equation}
G_l = \int_S r^{-l} P_l(\cos\theta)\cos\theta dS,\quad\quad
\end{equation}

The same recipe will be followed, substituting $r$ and writing $dS = r^2 d\theta d\phi$. Thus:
\begin{equation}
\ddot G_l = \int_0^{2\pi} \int_{\theta_0}^{\pi - \theta_0} r^{-l} P_l(\cos\theta) \cos\theta r^2 d\theta d\phi
= 2\pi \int_{\theta_0}^{\pi - \theta_0} \left(\frac{R}{\sin\theta}\right)^{-l+2} P_l(\cos\theta) \cos\theta d\theta
\end{equation}

Therefore, finalizing:
\begin{equation}
\ddot G_l = 2\pi R^{-l+2} \int_{\theta_0}^{\pi - \theta_0} P_l(\cos\theta) \left[\sin\theta\right]^{l-2} \cos\theta d\theta
\end{equation}

As before, applying substitution $x=\cos\theta$:
\begin{equation} \label{cylindrical_Gl}
\ddot G_l = 2\pi R^{-l+2} \int_{-a_c}^{a_c} x\left[\sqrt{1 - x^2}\right]^{l-3} P_l(x) dx
\end{equation}

Function $\sqrt{1-x^2}$ is always even. If $l$ is odd, then $x P_l$ is even. If $l$ is even, then $x P_l$ is odd. Because the integral limits are symmetric, if the entire integrand is odd, then the integral must evaluate to zero, which happens when $l$ is even. Thus, $\ddot G_{2l} = 0$.

One can now calculate the value of $G_1$.
\begin{equation}
\ddot G_1 = 2\pi R\int_{-a_c}^{a_c}\frac{x^2}{1 - x^2}dx = 2\pi R\left[\arctanh(x) - x\right]_{-a_c}^{a_c}
\end{equation}

Where substitution $x \leftarrow \tanh x$ was done to solve the integral. Therefore:
\begin{equation} \label{cylindrical_G1}
\ddot G_1 = 4\pi R\left[
\arctanh\left(\frac{\ell}{r_0}\right) - \frac{\ell}{r_0}
\right]
\end{equation}

If it is the case that $\ell\ll R$, then $a_c\ll 1$, meaning, the Taylor expand $\arctanh$ would be good approximation.
\begin{equation}
\arctanh(x) = x + \frac{x^3}{3} + \frac{x^5}{5} + O(x^7)
\end{equation}

Therefore, for $\ell\ll R$, expression (\ref{cylindrical_G1}) can be written:
\begin{equation} \label{cylindrical_G1_ellllR}
\ddot G_1 = \frac{4}{3}\pi R\left[
\left(\frac{\ell}{r_0}\right)^3
+\frac{3}{5}\left(\frac{\ell}{r_0}\right)^5
+\cdots
\right]
\end{equation}

\subsection{Hemispherical area element}
Expression (\ref{hemispherical_J}), as derived at Appendix A, show the area element of the hemisphere:
\begin{equation} \label{sus_hemispherical_J}
J = r^2\sin\theta\sqrt{1 + \frac{\ell^2}{r^2}\sin^2\theta\left(1 + \frac{\ell\cos\theta}{\sqrt{R^2 - \ell^2\sin^2\theta}}\right)^2}
\end{equation}

Where $r$ is shown in expression (\ref{suspended_hemispherical_r_signed_l}) [also in Appendix A], written below:
\begin{equation} \label{sus_hemispherical_r}
r = \ell\cos\theta + \sqrt{R^2 - \ell^2\sin^2\theta}
\end{equation}

In above equation, the value of $\ell$ is signed: $\ell > 0$ for the hemisphere above the z-axis, and $\ell < 0$ for the hemisphere below the z-axis.

The area element $J$ can be approximated as follows: if $\ell\gg R$ (also $h\gg R$), as in (\ref{hemispherical_J_ellggR}), then:
\begin{equation}
J = r^2\sin\theta\sqrt{1 + \frac{\ell^2}{r^2}\frac{\ell^2\cos^2\theta\sin^2\theta}{R^2 - \ell^2\sin^2\theta}},
\quad\quad\frac{\ell}{R}\gg 1
\end{equation}

On the other hand, if $\ell\ll R$, as in (\ref{hemispherical_J_ellllR}), then:
\begin{equation}
J = r^2\sin\theta\sqrt{1 + \frac{\ell^2}{r^2}\sin^2\theta}
\quad\quad\frac{\ell}{R}\ll 1
\end{equation}

\subsection{Suspended hemispherical G-Integrals}
Recalling $dS = J d\theta d\phi$ where $(\ref{hemispherical_J_ellggR})$ shows the value equation for $J$. Recall the integral to solve:
$$
\mathring G_l = \int_S r^{-l} P_l(\cos\theta) \cos\theta dS
$$

Inserting the exact $J$ in expression above, and doing parity considerations, just like it was done with the prolate spheroid, the integral can be simplified as in equation (\ref{suspended_hemispherical_Gl_exact}) [located in Appendix C], written below:
\begin{equation} \label{hemispherical_Gl}
\begin{split}
\mathring G_{2l} &= 0 \\
\mathring G_{2l+1} &= 4\pi \int_{\ell/r_0}^1 r^{-2l+1} P_{2l+1}(x) x J_A dx
\end{split}
\end{equation}

Where:
\begin{equation} \label{sus_r0_JA}
r_0^2 = \ell^2 + R^2, \quad\quad
J = r^2\sin\theta J_A
\end{equation}

Therefore, just like the prolate spheroid, a suspended hemisphere also has $G_{l} = 0$ if $l$ is even. About solving for the odd values of $l$, above integral is still too complicated for a direct substitution of $r$ and $J$. If $\ell = 0$, the calculations reduces to the case of the sphere, which was solved exactly using this method, and found that a single dipole describes the entire field. It is expected, that, higher deviations from the spherical shape will cause more and more contributions of multipole elements. With that in mind, it makes sense to consider $\ell\ll R$: in this condition, the dipole is still the dominant coefficient, and higher multipole contributions will probably be small.

With that in mind, one can Taylor expand the integral $r^n P_l(x) x J_A$ around the center $\ell = 0$ until second order, and simplify above integral to the expression written at (\ref{suspended_hemispherical_Gl_hllR_integral_solution}) [detailed calculations at Appendix C] case, written below:
\begin{equation}
\begin{split}
G_l &= 4\pi R^{-l+2} \int_{\ell/r_0}^1  P_l(x) x dx \\
&+ 4\pi (-l+2) R^{-l+1}\ell \int_{\ell/r_0}^1  P_l(x) x^2 dx \\
&+ (l^2 - 4l + 2)\frac{4\pi}{R^l}\ell^2 \int_{\ell/r_0}^1 x^3 P_l(x) dx \\
&+ \frac{4l\pi}{R^l}\ell^2 \int_{\ell/r_0}^1 x P_l(x) dx \\
\end{split}
\end{equation}

These integrals are easier to solve, requiring only the calculations of moments of Legendre polynomials, and this can be solved exactly. For instance, expression (\ref{Legendre_first_moment}) [located at Appendix C] is the result of the first moment of Legendre polynomials $P_l$. A description of how to calculate higher moments can be found in Appendix C. Nevertheless, the final analytic solutions for a general $G_l$ are big and complicated, and, they do not provide intuitive grounds to base the analysis. Furthermore, solutions of moments of $P_i P_j$, required for $I_{ij}$, are not known. Thus, even if one writes complete expressions for $G_l$, the $I_{ij}$ would still be lacking, rendering the $G_l$ expressions useless.

With all of that in mind, one can find specific values of $l$, which will be done for $G_0$ and $G_1$. It is already known from expression (\ref{hemispherical_Gl}) that $G_0 = 0$. Solving (\ref{hemispherical_Gl}) for $l=1$ yields the expressions (\ref{suspended_hemispherical_G1_ellllR}) [located at Appendix C], written below:
\begin{equation} \label{hemispherical_G1_ellllR}
\begin{split}
G_1 &= \frac{4}{3}\pi R\left[1 - \left(\frac{\ell}{r_0}\right)^3\right]
+ \pi\ell\left[1 - \left(\frac{\ell}{r_0}\right)^4\right] \\
&- \frac{4\pi}{5R} \ell^2\left[1 - \left(\frac{\ell}{r_0}\right)^5\right]
+ \frac{4\pi}{3R} \ell^2\left[1 - \left(\frac{\ell}{r_0}\right)^3\right] \\
\end{split}
\end{equation}

Just like one could anticipate, if $\ell = 0$, the value of $G_1$ becomes the calculated value for the sphere, that is $4/3 \pi R$, as can be seen in equation (\ref{spherical_Gl_exact}). The other terms, are the Taylor corrections for low values of $\ell$.

For the case where $\ell\gg R$, a calculation of $G_1$ can also be found in Appendix C.

\subsection{Suspended hemispherical I-Integrals}
Recall the I-Integral at (\ref{IG_integrals}), that is:
\begin{equation}
\mathring I_{ij} = \int_S r^{-i-j-2} P_i(\cos\theta) P_j(\cos\theta) dS
\end{equation}

As before, substitute area element $dS = Jd\theta d\phi$, and thus:
\begin{equation}
\mathring I_{ij} = \int_0^{2\pi} \int_{0}^{\theta_0} r^{-i-j-2} P_i(\cos\theta) P_j(\cos\theta) J d\theta d\phi
\end{equation}

The same procedure that was done for the $G_l$ integrals can be done to calculate $I_{ij}$. It is shown in expression at Appendix C, (\ref{suspended_hemispherical_Iij_exact}), exact expression for the $I_{ij}$ integrals, shown below:
\begin{equation} \label{suspended_hemispherical_Iij_exact2}
\mathring I_{ij} =
4\pi\int_{\ell/r_0}^1 r^{-i-j} P_i(x) P_j(x) J_A dx,
\quad\quad i+j\in\{0, 2, 4, 6, 8, \dots\}
\end{equation}

And $I_{ij} = 0$ otherwise. Again, doing a Taylor expand around $\ell = 0$, where detailed derivation is at Appendix C, final result can be found in equation (\ref{suspended_hemispherical_Iij_ellllR}), written below:
\begin{equation} \label{suspended_hemispherical_Iij_ellllR2}
\begin{split}
\mathring I_{ij} &=
\frac{4\pi}{R^{i+j}}
\int_{\ell/r_0}^1 P_i(x) P_j(x) dx \\
&-(i+j)\frac{4\pi}{R^{i+j+1}}\ell
\int_{\ell/r_0}^1 x P_i(x) P_j(x) dx \\
&+\left[(i+j)^2 - 2\right]\frac{4\pi}{R^{i+j+2}}\ell^2
\int_{\ell/r_0}^1 x^2 P_i(x) P_j(x) dx \\
&+\left[(i+j) + 2\right]\frac{4\pi}{R^{i+j+2}}\ell^2
\int_{\ell/r_0}^1 P_i(x) P_j(x) dx \\
\end{split}
\end{equation}

Given $\mathring I_{01} = \mathring I_{10} = 0$ because $1+0$ is odd, expressions (\ref{suspended_hemispherical_I00_ellllR}) and (\ref{suspended_hemispherical_I11_ellllR}) [also at Appendix C] show the values of $\mathring I_{00}$ and $\mathring I_{11}$, written below:
\begin{equation} \label{hemispherical_I0011_ellllR}
\begin{split}
\mathring I_{00} &= 4\pi \cdot\left[1 - \frac{\ell}{r_0}\right]
-\frac{8\pi}{3R^2}\ell^2\left[1 - \left(\frac{\ell}{r_0}\right)^3\right]
+\frac{8\pi}{R^2}\ell^2\left[1 - \frac{\ell}{r_0}\right]
\\
\mathring I_{11} &= \frac{4\pi}{3R^2}\left[1 - \left(\frac{\ell}{r_0}\right)^3\right]
-\frac{2\pi}{R^3}\ell\left[1 - \left(\frac{\ell}{r_0}\right)^4\right]
+\frac{8\pi}{5R^4}\ell^2\left[1 - \left(\frac{\ell}{r_0}\right)^5\right]
+\frac{16\pi}{3R^4}\ell^2\left[1 - \left(\frac{\ell}{r_0}\right)^3\right]
\end{split}
\end{equation}

It is relevant to notice yet agin again, how they generalize the values found for the sphere ($\ell=0$) at equation (\ref{spherical_Iij_exact}), namely:
\begin{equation}
I_{00} = 4\pi,\quad\quad
I_{11} = \frac{4\pi}{3R^2}
\end{equation}

\section{HCP Integrals}
With all of that, one can come up with the HCP Integrals. The HCP shape is composed of a hemisphere of radius $R$ on top of a cylinder of length $\ell$ in a plate. It was already discussed that, it was much better to, instead, compute the equivalent problem of a hemisphere of radius $R$ on top of a cylinder of radius $2\ell$ centered on the z-axis, and another hemisphere of radius $R$ on the bottom of the cylinder. That was the case when it was calculated: $\mathring I_{ij}$ and $\mathring G_l$ are the integrals calculated on both hemispheres, above and below the z-axis, while $\ddot I_{ij}$ and $\ddot G_l$ are the integrals for a cylinder of length $2\ell$, $\ell$ above the z-axis, and $\ell$ below the z-axis. Therefore, for a complete HCP shape:
\begin{equation} \label{HCP_IG_sum}
\begin{split}
G_l &= \ddot G_l + \mathring G_l \\
I_{ij} &= \ddot I_{ij} + \mathring I_{ij}
\end{split}
\end{equation}

That is true, precisely because, for arbitrary function $f$ is valid:
\begin{equation}
\int_{\text{HCP}} fdS
=\int_{\text{Suspended Hemispheres}} fdS
+\int_{2\ell\text{-length Cylinder}} fdS
\end{equation}

\section{HCP Multipoles}
Notice that, it was proved for the exact case, that, $\ddot G_{2l} = \mathring G_{2l} = 0$. In addition, $\ddot I_{ij} = \mathring I_{ij} = 0$ for $i+j$ even. Thus, just like the case with the prolate spheroid, the HCP linear system is as described in equation (\ref{spheroidal_system}). It was shown in equation (\ref{spheroidal_linear_system_zero}) that, for such case, $A_{2l} = 0$ for every $l$, and, by equation (\ref{spheroidal_linear_system_nonzero}) only odd $l$ at $A_l$ contributes to the potential (\ref{potential}).

Hereby, using similar arguments of the prolate spheroidal case, one can conclude that \textbf{only odd multipoles will contribute to the potential of a HCP shape}, that is, dipole, octopole, etc.

Recalling the multipole moments expression (\ref{multipole_moment}), if $Q_{2l} = 0$, then, the induced charge density over the surface $\sigma$ has to be of a certain way.
\begin{equation}
Q_{2l} = \int_S \sigma(\mathbf r) r^{2l} P_{2l}(\cos\theta) dS(\mathbf r)
\end{equation}

Because the shape is axially symmetric, integration can be done in domains $\phi\in[0, 2\pi]$ and $\theta\in [0, \pi]$, where, the $\phi$ integral will yield $2\pi$, thus:
\begin{equation}
Q_{2l} = 2\pi\int_{0}^\pi \sigma(\theta) r(\theta)^{2l+2} P_{2l}(\cos\theta)\sin\theta J_A(\theta) d\theta
\end{equation}

Substituting $x=\cos\theta$, then:
\begin{equation}
0 = Q_{2l} = \int_{-1}^1 \sigma(x) r(x)^{2l+2} P_{2l}(x) J_A(x) dx
\end{equation}

It was shown for a HCP shape (both cylinder and suspended hemisphere), that $r(x)$, $J_A(x)$ were even functions. Furthermore, $P_{2l}$ is even. Notice then, that, if $\sigma(x)$ is odd, then the entire integrand is odd, then all $Q_{2l}$ is automatically zero. One way for that to happen, is, $\sigma(x)$ being an odd function.

\begin{theorem}
If the induced charge density $\sigma(x)$ is an odd function, then $Q_{2l} = 0$.
\end{theorem}

Evidently, the goal is to prove the converse, that is, if $Q_{2l} = 0$, then $\sigma(x)$ is odd. That can be done, if one has an expression for $Q$:
\begin{equation}
Q_{l} = \int_{-1}^1 \sigma(x) r(x)^{l+2} P_{l}(x) J_A(x) dx
\end{equation}

Notice that, expression above shows that $Q_l$ are proportional to the coefficients of the Legendre expansion of the function $\sigma(x) r(x)^{l+2} J_A(x)$, that is:
\begin{equation}
r(x)^{l+2} J_A(x) \sigma(x) = \sum_{n=0}^\infty\frac{2Q_l}{2l+1} P_l(x)
\end{equation}

Therefore, if all $Q_{2l} = 0$, it means that only odd Legendre polynomials contribute to the summation, meaning $r(x)^{l+2} J_A(x) \sigma(x)$ must be an odd function. Because it is already known that $r(x)$ and $J_A(x)$ are even, then, $\sigma(x)$ must be odd. It was thus, proven the following theorem:

\begin{theorem} \label{HCP_odd_sigma_x}
$Q_{2l} = 0$ for all $l\in\mathbb{N}\cup\{0\}$, if and only if, $\sigma(x) + \sigma(-x) = 0$.
\end{theorem}

% Of course, by the substitution of variables, $\sigma(x)$ can be interpreted as a line distribution of charge from $-1$ to $+1$, along the z-axis. We collapsed a surface charge density $\sigma(\mathbf r)$ into a line charge density $\sigma(x)$ when we used the fact the shape was axially symmetric, and the integral of $\phi$ is $2\pi$. Therefore, equivalently, one can view the system as a line on the $z$-axis from $-1$ to $1$, with a charge density of $\sigma(x)$, producing a potential identical to the HCP shape. The fact only odd multipoles contribute implies that \textbf{the induced equivalent axial line charge density $\sigma(x)$ must be an odd function}.

Recall also, that, the HCP shape was considered to be one resembling a capsule, because of the convenience of excluding the plane from the integrals. Therefore, $\sigma(x), x>0$ is the distribution along the region of the hemisphere on a post, while $\sigma(x), x < 0$ is the distribution along the region on the plane (or, equivalently, on the mirror hemisphere on a post). This connects the charge distributions on the plane, and on the hemisphere and the post. More formally, the charge distribution over the plane can be calculated by means of the local electric field solution $\mathbf E(\mathbf r)$, normal to the plane.

More, if $l$ is even, then the $l$-multipole of the plane is equal in magnitude, but opposite in sign of the $l$-multipole of the hemisphere on a post (and thus, both of them sum to zero). And, if $l$ is odd, they are equal in magnitude and in sign.

\subsection{Line of charge model}
One can seek to explore even more the result of zero even multipoles, and apply it to other models. Let a line of charge along the z-axis, going from $-a < z < a$ for some $a > 0$. Such a line has a distribution of linear charge density $\lambda(z)$. The axial multipoles of such a system can be calculated:
\begin{equation} \label{multipole_moment_line_of_charge}
M_l = \int_{-a}^a \lambda(z) z^l dz
\end{equation}

If one is trying to approximate the potential of a HCP shape by using a line of charge, it is clear that the even multipoles must be zero. Also, from the expression above, if $\lambda(z)$ is an odd function, the multipole condition is automatically satisfied. That happens precisely because $z^l$ is even if $l$ is even. Thus, the following is valid:

\begin{theorem} \label{linear_odd_lambda_z}
If $\lambda(z)$ is the linear charge density of a line from $[-a, a]$ is an odd function, then the multipoles $M_{2l} = 0$.
\end{theorem}

Again, one would like to prove the converse, however, the lack of orthogonality of the monomials $z^l$ prevents a proof just like the surface charge density $\sigma(x)$ from a HCP shape.

%However, the converse is indeed true, provided $\lambda(z)$ meets certain mathematical criteria. For example, if $\lambda(z)$ can be written as an uniformely convergent power series, then, the converse is true.

Because of this, one is encouraged to choose an odd function when modeling with a line of charge. For instance, the linear charge density in work \cite{india} was chosen to be an odd function.

% Power series condition.
% Differentiability condition?
% https://en.wikipedia.org/wiki/Stone%E2%80%93Weierstrass_theorem
% https://math.stackexchange.com/questions/165825/if-int-01-fxxn-dx-0-for-every-n-then-f-0
% https://math.stackexchange.com/questions/2045595/if-int-01-fxxn-dx-frac1n2-for-every-nonnegative-n-and-f-i

% To show a large of charge [a, b] can be modeled by equivalent [-1, 1]
% To show if there's a distribution lambda from it that models HCP, then, it must be odd.
% To compare actual multipoles from lambda and sigma (from inequalities, hopefully).

\section{Interacting HCPs}
\subsection{Two interacting HCPs}
With two HCP shapes, separated by a distance $c$, it was proven there's no image charge (because it is the zeroth multipole, and zero is even), thus, the next multipole contribution is extremely likely to be of a dipole. If $c$ is large enough as compared to the dimensions of a HCP shape, the dominant interaction will be the dipole. With the knowledge of the dipole moment $Q_1$ and the FEF of an isolated $\gamma_a$, it is possible to calculate how much the FEF will fall.

\begin{figure} \label{interacting_hcp}
\includegraphics[scale=0.20]{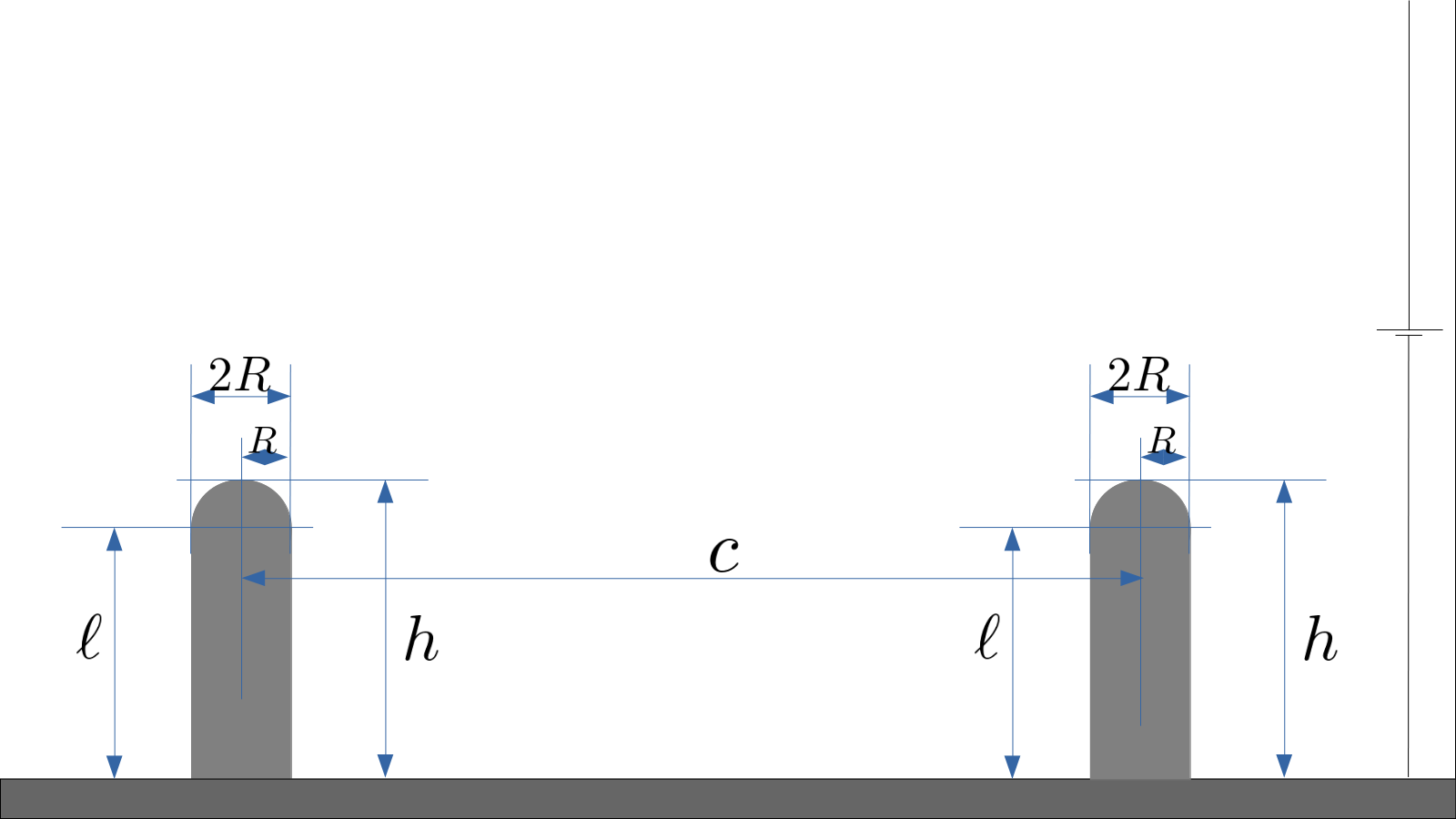}
\caption{Two interacting HCP shapes}
\end{figure}

The electric field of a dipole:
\begin{equation}
\mathbf E(\mathbf r) = \frac{3(\mathbf p\cdot\mathbf{\hat r})\mathbf{\hat r} - \mathbf p}{4\pi\epsilon_0 r^3}
\end{equation}

Therefore, considering the multipoles exactly at the $z=0$ at their respective positions, which, without loss of generality will be considered to be $(0, 0, 0)$ and $(c, 0, 0)$, the FEF is reduced to:
\begin{equation}
\gamma_a' = \gamma_a + \frac{\mathbf E(c, 0, h)\cdot\mathbf{\hat z}}{E_0}
\end{equation}

Where $E_0$ is the external field. Here, it was explicitly considered only the contribution in the $\mathbf{\hat z}$ direction, because, the electric field at the apex must parallel to the normal vector of the surface on the apex, which, for a HCP, happens to be $\mathbf{\hat z}$.

The distance from the dipole to the tip of the HCP is, $r = \sqrt{c^2 + h^2}$. Furthermore, it was shown earlier that all $A_l$ (and thus, all multipoles $Q_l$) scales linearly with the external field $E_0$. Therefore, the new $\gamma_a'$ because the presence of the other HCP shape, becomes:
\begin{equation}
\gamma_a' = \gamma_a + \frac{Q_1}{E_0}\frac{3\cos(\mathbf{\hat p}, \mathbf{\hat r})\mathbf{\hat r} - \mathbf{\hat p}}{4\pi\epsilon_0\sqrt{c^2 + h^2}^3}\cdot\mathbf{\hat z}
\end{equation}

Because $\mathbf p = Q_1\mathbf{\hat z}$, then:
\begin{equation}
\gamma_a' = \gamma_a + \frac{Q_1}{E_0}\frac{3\cos^2(\mathbf{\hat p}, \mathbf{\hat r}) - 1}{4\pi\epsilon_0\sqrt{c^2 + h^2}^3}
\end{equation}

Where:
\begin{equation}
\cos(\mathbf{\hat r}, \mathbf{\hat z}) = \frac{h}{r} = \frac{h}{\sqrt{h^2 + c^2}},\quad\quad \\
\sin(\mathbf{\hat r}, \mathbf{\hat z}) = \frac{c}{r} = \frac{c}{\sqrt{h^2 + c^2}} \\
\end{equation}

Thus:
\begin{equation}
\gamma_a' = \gamma_a + \frac{Q_1}{E_0}\frac{1}{4\pi\epsilon_0(c^2 + h^2)^{3/2}}
\left[
\frac{3h^2}{h^2 + c^2} - 1
\right]
\end{equation}

Or, equivalently, using equation (\ref{Ql_Al}) connecting $Q_1$ and $A_1$, instead, one can write in terms of $\tilde A_1$:
\begin{equation}
\gamma_a' = \gamma_a + \frac{\tilde A_1}{(c^2 + h^2)^{3/2}}
\left[
\frac{3h^2}{h^2 + c^2} - 1
\right]
\end{equation}

Which can also be written:
\begin{equation} \label{HCP_delta_third_power_law}
\delta =
\frac{\gamma_a' - \gamma_a}{\gamma_a}
=-\frac{1}{\gamma_a}
\frac{\tilde A_1}{(c^2 + h^2)^{3/2}}
\left[
1 -
\frac{3h^2}{h^2 + c^2}
\right]
\end{equation}

Therefore, $\delta$ decreases by the third power of $c$, until a certain value $c_k$, and it begins to increase again. Such value can be calculated from expression above:
\begin{equation}
\frac{3h^2}{h^2 + c_k^2} - 1 = 0
\quad\iff\quad
c_k = h\sqrt{2}
\end{equation}

That is:
\begin{equation}
\begin{split}
c > h\sqrt{2} &\quad\implies\quad\delta < 0 \\
c = h\sqrt 2 &\quad\implies\quad\delta = 0 \\
c < h\sqrt 2 &\quad\implies\quad\delta > 0
\end{split}
\end{equation}

That is, if $h$ is comparable to the distance between the posts $d$, then the variation in FEF begins to decrease, until zero is reached. It is important to notice this treatment considered the interaction of them to be identical to two independent dipoles.

It is not the case that at $h\sqrt{2}$ the field starts decreasing, because, at such short distances, octopole and higher multipole might contribute. However, what was proven is: if electrostatic interactions are negligible (if the two systems can be treated independently), then, the dipole moment contributes by decreasing the FEF until $h\sqrt{2}$. In fact, it is not even clear if the independent approximation is valid for $d\approx h$. Thus, what was done must be regarded as a good approximation only for large distances $c$.

% In fact, it was shown \cite{Assis2018} numerical evidence that $-\delta \sim c^{-3}$. From figure 3 of same paper \cite{Assis2018}, one can see that the third power law begins to fail around $\ln(c/h) \approx 1$, which is not far from the range of $c\approx h\sqrt{2}$, since $\ln\sqrt{2} \approx 0.35$.

As for the pre-factor $K$ in terms of the normalized distance $s = c/h$, one can get:
\begin{equation} \label{pre_factor_1}
\delta = -\frac{1}{\gamma_a}\frac{\tilde A_1}{h^3(1 + (c/h)^2)^{3/2}}
\approx -\frac{1}{\gamma_a}\frac{\tilde A_1}{h^3 s^3}
\quad\mbox{if}\quad\frac{c}{h}\gg 1
\end{equation}

From (\ref{FEF}), if $\gamma_a$ depends only on the aspect ratio term by term, then:
\begin{equation} \label{normalized_moment}
\frac{\tilde A_l}{h^{l+2}} = f_l(\nu)
\quad\implies\quad
\tilde A_l = h^{l+2} f_l(\nu)
\end{equation}

Where $f_l$ is some function. This implies $\tilde A_l$ (and hence all multipole moments $Q_l$, because of (\ref{Ql_tilde_Al})) do not depend on the aspect ratio, rather, they need to be normalized by $h^{l+2}$. Considering the special case $l=1$, the dipole term: $\tilde A_1 = h^3 f_1(\nu)$. Making the substitution into (\ref{pre_factor_1}), then:
\begin{equation} \label{pre_factor}
\delta \approx -\frac{1}{\gamma_a s^3}f_1(\nu)
\quad\implies\quad
K \approx \frac{\tilde A_1}{h^3}\frac{1}{\gamma_a} = \frac{f_1(\nu)}{\gamma_a}
\end{equation}

Therefore, if $\delta$ is put in term of the normalized distance $s = c/h$, then the pre-factor $K$ depends only on the aspect ratio, $K = K(\nu)$. Furthermore, $K\neq 1$, in contrast with some references \cite{Bonard2001,Jo2003,Harris2015} where the pre-factor was absent in the curve-fitting functions (that is, equivalently, $K=1$ was assumed).

\subsection{One-dimensional regular array of HCPs}
Let an infinite number of HCP shapes located at positions $(nc, 0, 0)$, where $n\in\mathbb{Z}$. That is, the HCPs are distributed in a single axis. The aim is to calculate the fractional change in the apex FEF, $\delta$. From (\ref{HCP_delta_third_power_law}), the result will be:
\begin{equation}
-\delta =
\frac{\gamma_a' - \gamma_a}{\gamma_a}
=\sum_{n\in\mathbb{Z^*}}\frac{1}{\gamma_a}
\frac{\tilde A_1}{(n^2 c^2 + h^2)^{3/2}}
\left[
1 -
\frac{3h^2}{h^2 + n^2 c^2}
\right]
\end{equation}

Where it was denoted $\mathbb{Z}^* = \mathbb{Z}-\{0\}$. A reasonable first step would be to simplify the expression in the summand. Notice that, as $n$ grows larger:
\begin{equation}
\frac{3h^2}{h^2 + n^2 c^2}\to 0,\quad\quad
\frac{\tilde A_1}{(n^2 c^2 + h^2)^{3/2}}\to\frac{\tilde A_1}{n^3 c^3}
\end{equation}

Therefore, the sum becomes:
\begin{equation}
-\delta
=2\sum_{n=1}^\infty\frac{1}{\gamma_a} \frac{\tilde A_1}{n^3 c^3}
=\frac{2\tilde A_1}{\gamma_a c^3}\sum_{n=1}^\infty\frac{1}{n^3}
=\frac{2\tilde A_1}{\gamma_a c^3}\zeta(3)
\end{equation}

Where $\zeta(n)$ is the Zeta-Riemann function, and $\zeta(3) \approx 1.2020569$ is known as Apéry's constant. Also, $E_0\gamma_a$ is the electric field on the top of the HCP shape. The dependence continues to be $-\delta = c^3$. In other words:
\begin{equation} \label{HCP_delta_1D_array}
-\delta=\frac{2\tilde A_1}{\gamma_a c^3}\zeta(3)
=\frac{2 Q_1\zeta(3)}{4\pi\epsilon_0 E_0 \gamma_a c^3}
\end{equation}

In order to investigate what happens in the borders, a possible model could be a semi-infinit e array: Let an infinite number of HCP shapes located at positions $(nc, 0, 0)$ where $n\in\mathbb{N}\cup\{0\}$. In that case:
\begin{equation}
\begin{split}
-\delta =
-\frac{\gamma_a' - \gamma_a}{\gamma_a}
&=\sum_{n\in\mathbb{N}}\frac{1}{\gamma_a}
\frac{\tilde A_1}{(n^2 c^2 + h^2)^{3/2}}
\left[
1 -
\frac{3h^2}{h^2 + n^2 c^2}
\right] \\
&\approx\sum_{n\in\mathbb{N}}\frac{1}{\gamma_a}
\frac{\tilde A_1}{(n^2 c^2 + h^2)^{3/2}} \\
&\approx\sum_{n\in\mathbb{N}}\frac{1}{\gamma_a}
\frac{\tilde A_1}{n^3 c^3} \\
&=\frac{\tilde A_1}{\gamma_a c^3}\sum_{n=1}^\infty\frac{1}{n^3} \\
&=\frac{\tilde A_1}{\gamma_a c^3}\zeta(3)
\end{split}
\end{equation}

Therefore, the fractional reduction in apex FEF, is:
\begin{equation}
-\delta =
-\frac{\gamma_a' - \gamma_a}{\gamma_a}
\approx\frac{\tilde A_1}{\gamma_a c^3}\zeta(3)
\end{equation}

This result shows that, $\delta$ in an infinite array is twice as more as the $\delta$ in a semi-infinite array, meaning, apex fields on the emitters of the borders are larger than the emitters located on the the middle of the linear array. This is consistent with \cite{Harris2016}, where it was reported that the emitters at the borders degrades faster.

\section{HCP FEF for $\ell\ll R$}
Hereby, for $\ell\ll R$, for $G_1$, equations (\ref{cylindrical_G1_ellllR}) and (\ref{hemispherical_G1_ellllR}) can be summed, written below:
\begin{equation}
\begin{split}
G_1 &= \frac{4}{3}\pi R\left[1 - \left(\frac{\ell}{r_0}\right)^3\right]
+ \pi\ell\left[1 - \left(\frac{\ell}{r_0}\right)^4\right] \quad\text{(hemisphere)} \\
&- \frac{4\pi}{5R} \ell^2\left[1 - \left(\frac{\ell}{r_0}\right)^5\right]
+ \frac{4\pi}{3R} \ell^2\left[1 - \left(\frac{\ell}{r_0}\right)^3\right]\quad\text{(hemisphere)} \\
&+\frac{4}{3}\pi R\left(\frac{\ell}{r_0}\right)^3\quad\text{(cylinder)}
\end{split}
\end{equation}

Thefore, one term of the suspended hemispherical $G$ and cylindrical $G$ integrals cancel out.
\begin{equation} \label{HCP_G1_ellllR}
\begin{split}
G_1 &= \frac{4}{3}\pi R
+ \pi\ell\left[1 - \left(\frac{\ell}{r_0}\right)^4\right] \\
&- \frac{4\pi}{5R} \ell^2\left[1 - \left(\frac{\ell}{r_0}\right)^5\right]
+ \frac{4\pi}{3R} \ell^2\left[1 - \left(\frac{\ell}{r_0}\right)^3\right] \\
\end{split}
\end{equation}

At first order of $\ell$, the correction of $G_1$ from the sphere case is merely $\pi\ell$.

The relevant value now, would be $I_{11}$. Thus, one can sum (\ref{cylindrical_I0011_ellllR}) with (\ref{hemispherical_I0011_ellllR}), then:
\begin{equation} \label{HCP_I11_ellllR}
\begin{split}
I_{11} &= \frac{4\pi}{3R^2}\left[1 - \left(\frac{\ell}{r_0}\right)^3\right]
-\frac{2\pi}{R^3}\ell\left[1 - \left(\frac{\ell}{r_0}\right)^4\right] \\
&+\frac{8\pi}{5R^4}\ell^2\left[1 - \left(\frac{\ell}{r_0}\right)^5\right]
+\frac{16\pi}{3R^4}\ell^2\left[1 - \left(\frac{\ell}{r_0}\right)^3\right] \\
&+\frac{\pi}{2R^2}\left(\frac{\ell}{r_0}\right)
+\frac{\pi}{2R^2}\frac{R\ell}{\ell^2 + R^2}\frac{\ell^2 - R^2}{\ell^2 + R^2}
\end{split}
\end{equation}

Truncating the linear system (\ref{spheroidal_linear_system_nonzero}) at $1\times 1$, then:
\begin{equation}
A_1 = \frac{G_1}{I_{11}}
\end{equation}

Where, all it is required to do, is to insert (\ref{HCP_I11_ellllR}) and (\ref{HCP_G1_ellllR}) at equation above. In order to simplify a little, terms of $\ell^2$ can be ignored. Therefore:
\begin{equation} \label{HCP_A1_value}
A_1^{(1)} \approx \frac{\frac{4}{3}\pi R + \pi\ell}{\frac{4\pi}{3R^2} - \frac{2\pi}{R^3}\ell +  \frac{\pi}{2R^2}\frac{\ell}{r_0}}
=R^2\frac{\frac{4}{3} R + \ell}{\frac{4}{3} - \frac{2\ell}{R} + \frac{1}{2}\frac{\ell}{r_0}}
\end{equation}

Using (\ref{FEF}), the FEF for becomes:
\begin{equation}
\gamma_a^{(1)} = 1 + \frac{2A_1^{(1)}}{(R + \ell)^3}
=1 + \frac{2R^2}{(R+\ell)^3}\cdot
\frac{\frac{4}{3} R + \ell}{\frac{4}{3} - \frac{2\ell}{R} + \frac{1}{2}\frac{\ell}{r_0}}
,\quad\quad\ell\ll R
\end{equation}

One can see, it depends on the ratios $h/R = \nu$, or $\ell/R = (h-R)/R = \nu-1$, therefore, it does depend explicitly only on the aspect ratio.

According to the formula above, as $h$ increases, $\gamma_a\to 1$ up to first order. The reason for this is exactly the same as explained in the case of the ellipsoid: the charge in general is concentrated on the hemisphere on top of the post, thus, the multipoles scales $Q_l\approx qa^l$ for some $a < h$, but $a\approx h$, thus the FEF converges slowly, as shown in equation (\ref{FEF_multipole_approximation}) which presents a rough estimation, written below:
\begin{equation}
\gamma_a \approx 1 + \frac{1}{h^2}\sum_{n=0}^\infty (l+1)q\left(\frac{a}{h}\right)^l
\end{equation}

\newpage
\section{FEF for HCP Model}
For a HCP shape of radius $R$ and total height $h = R + \ell$, where $R>0$ and $\ell > 0$,  the first step is to calculate the integrals (\ref{IG_integrals}), written below:
\begin{equation}
G_i = \int_S r^{-i} P_i(\cos\theta)\cos\theta dS,\quad\quad
I_{ij} = \int_S r^{-i-j-2} P_i(\cos\theta) P_j(\cos\theta) dS
\end{equation}

Above surface integrals simplifies to ordinary integrals as stated in equations (\ref{cylindrical_Gl}), (\ref{cylindrical_Iij}), (\ref{hemispherical_Gl}), (\ref{suspended_hemispherical_Iij_exact2}), (\ref{HCP_IG_sum}) with aid of (\ref{sus_hemispherical_J}), (\ref{sus_r0_JA}) (\ref{sus_hemispherical_r}) where integration limits can be found from (\ref{sus_r0_JA}) and (\ref{cylindrical_ac_as}). All required equations are summarized below:

\begin{equation}
\begin{split}
I_{ij} = \mathring I_{ij} + \ddot I_{ij} \\
G_l = \mathring G_l + \ddot G_l \\
\end{split}
\end{equation}

\begin{equation}
\ddot G_l = 2\pi R^{-l+2} \int_{-\ell/r_0}^{\ell/r_0} x\left[\sqrt{1 - x^2}\right]^{l-3} P_l(x) dx
\end{equation}

\begin{equation}
\ddot I_{ij} = \frac{2\pi}{R^{i+j}} \int_{-\ell/r_0}^{\ell/r_0} P_i(x) P_j(x) \left[\sqrt{1 - x^2}\right]^{i+j-1} dx
\end{equation}

\begin{equation}
\mathring G_{2l+1} = 4\pi \int_{\ell/r_0}^1 r^{-2l+1} P_{2l+1}(x) x J_A dx
\end{equation}

\begin{equation}
\mathring I_{ij} =
4\pi\int_{\ell/r_0}^1 r^{-i-j} P_i(x) P_j(x) J_A(x) dx, \quad i+j\text{ even}
\end{equation}

\begin{equation} \label{JA_HCP}
J_A(x) = \sqrt{1 + \frac{\ell^2}{r^2}(1-x^2)\left(1 + \frac{\ell x}{\sqrt{R^2 - \ell^2(1-x^2)}}\right)^2}
\end{equation}

\begin{equation} \label{r_HCP}
r(x) = \ell x + \sqrt{R^2 - \ell^2(1-x^2)}
\end{equation}

\begin{equation}
r_0 = \sqrt{R^2 + \ell^2}
\end{equation}

Where it was proved:
\begin{equation}
\begin{split}
\ddot G_{2l} &= \mathring G_{2l} = G_{2l} = 0, \quad\forall l\in\mathbb{N}\cup\{0\} \\
\ddot I_{ij} &= \mathring I_{ij} = I_{ij} = 0, \quad\forall i,j\text{ s.t. } i+j\in\{1, 3, 5, 7, 9, \dots\}
\end{split}
\end{equation}

Once $I_{ij}$ and $G_l$ are all known, solve the linear system (\ref{spheroidal_linear_system_nonzero}), written below:
\begin{equation}
\begin{bmatrix}
I_{11} & I_{13} & I_{15} & I_{17} & \cdots \\
I_{31} & I_{33} & I_{35} & I_{37} & \cdots \\
I_{51} & I_{53} & I_{55} & I_{57} & \cdots \\
I_{71} & I_{73} & I_{75} & I_{77} & \cdots \\
\vdots & \vdots & \vdots & \vdots & \ddots
\end{bmatrix}
\begin{bmatrix}
\tilde A_1 \\
\tilde A_3 \\
\tilde A_5 \\
\tilde A_7 \\
\vdots
\end{bmatrix}
=
\begin{bmatrix}
G_1 \\
G_3 \\
G_5 \\
G_7 \\
\vdots
\end{bmatrix} \\
\end{equation}

Once the values $\tilde A_l$ are solved, one can find the potential in the overall space by inserting the values at (\ref{potential}), written below:
\begin{equation}
V(r, \theta) = -E_0 r \cos\theta + E_0\sum_{l=0}^\infty\tilde A_l r^{-(l+1)} P_l(\cos\theta).
\end{equation}

Where, the FEF is expressed in equation (\ref{FEF}), where the apex is located at $r=h$, written below:
\begin{equation}
\gamma_a =
1 + \sum_{l=0}^\infty \tilde A_l\frac{l+1}{h^{l+2}}
\end{equation}

\section{Numerical Results}
Calculations have been done using the method, and truncating the linear system (\ref{linear_system}) until $n$, for aspect ratios $\nu = h/R = 1.5$ and $\nu=2$. For comparison, the apex FEF was obtained numerically [Thiago Albuquerque de Assis, private communication] by using the minimal domain size (MDS) method \cite{thiago}.

For $\nu=1.5$ and $\nu=2$, the apex FEF values were found to be $\gamma_{MDS} \approx 3.62527$ and $\gamma_{MDS} \approx 4.20577$, respectively. It was found that $\gamma_a^{(n)}$ from the method does converge to $\gamma_{MDS}$, as shown in the plots on Fig (4.1).

% v=1.5, FEF=3.62527
% v=2.0, FEF=4.20577

\begin{figure}
\includegraphics[scale=0.32]{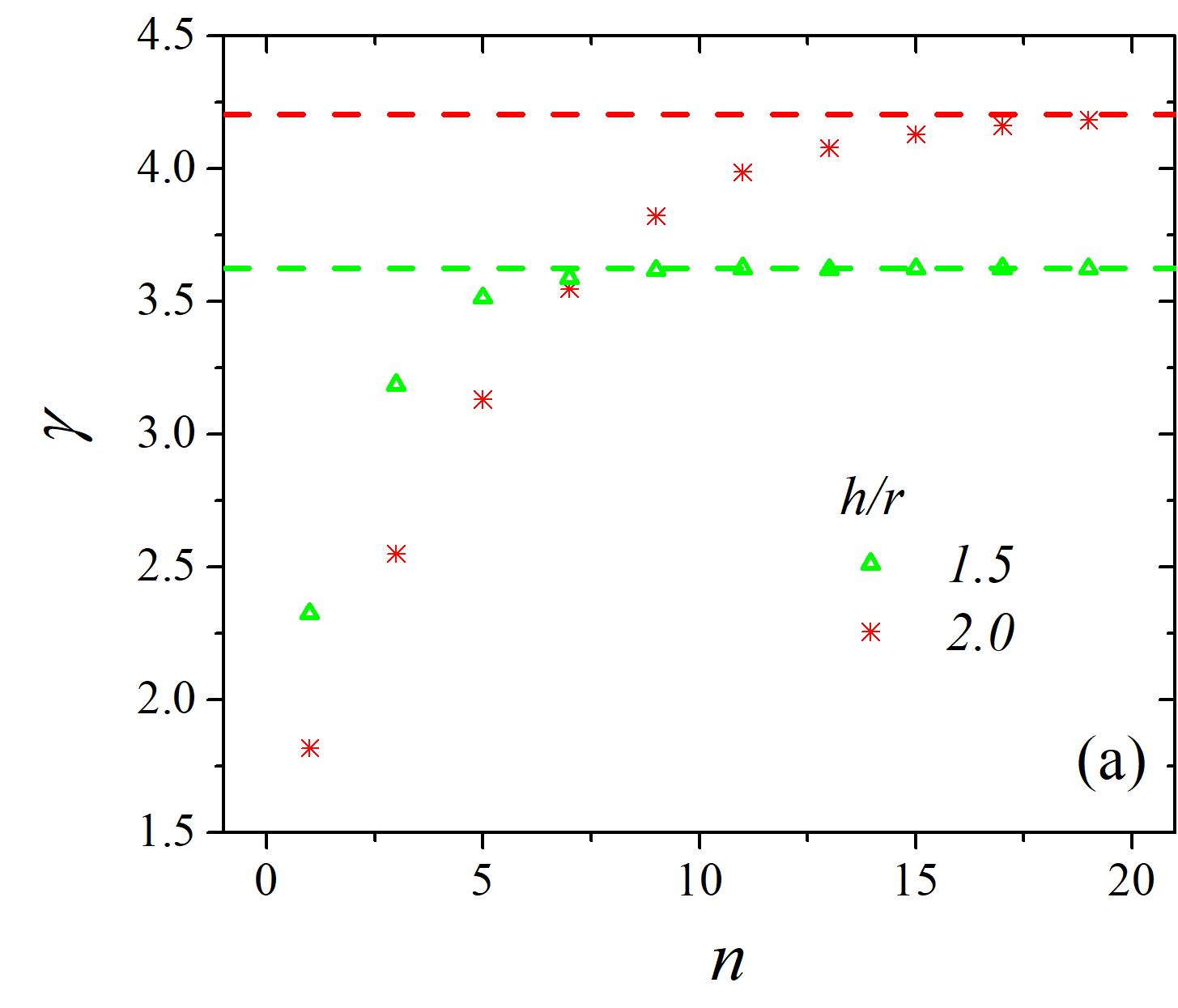}
\includegraphics[scale=0.32]{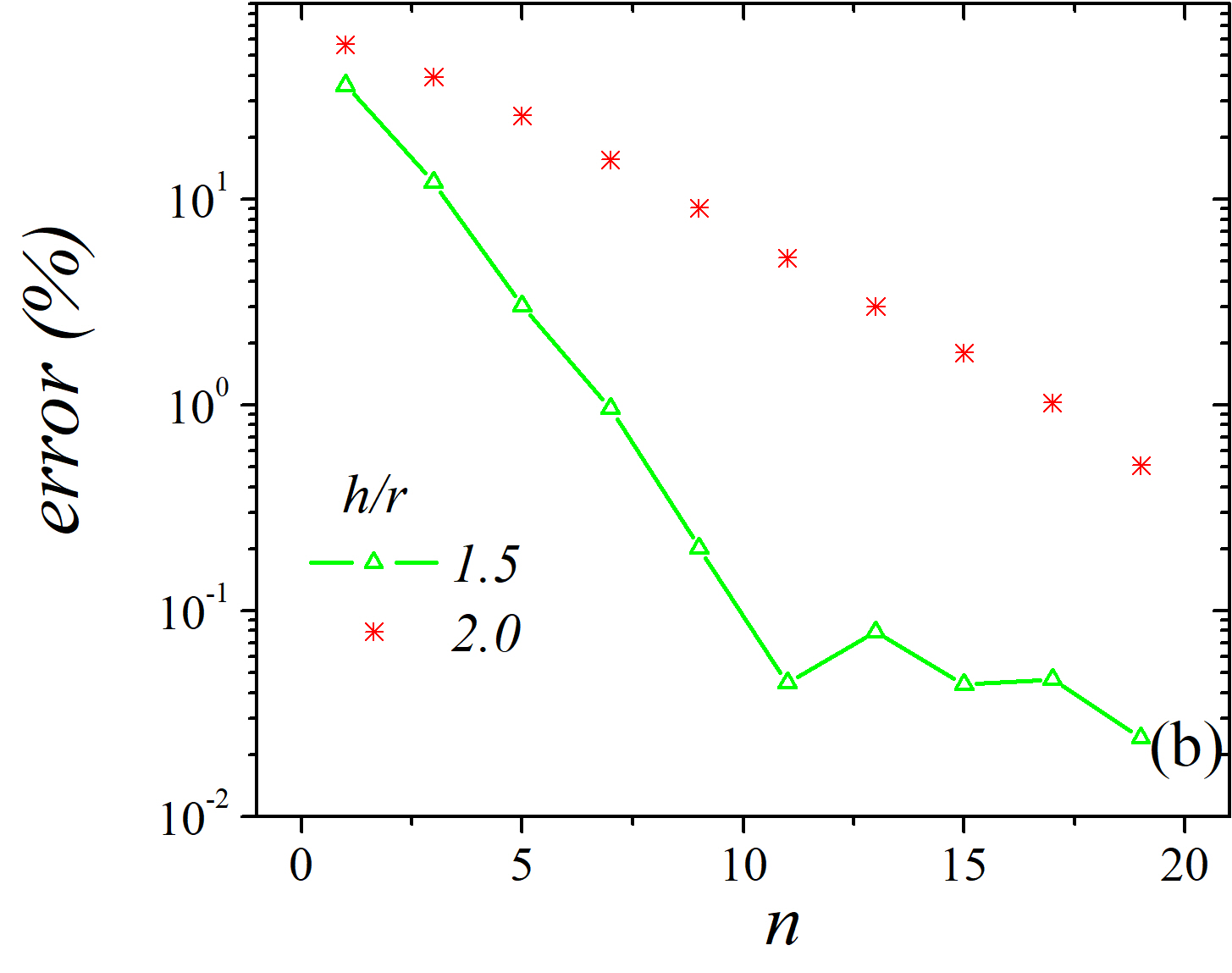}
\caption{(a) The FEF $\gamma_a^{(n)}$ is plotted against the truncation $n$. The dashed line correspond to $\gamma_{MDS}$. (b) Error $\frac{\gamma_a^{(n)} - \gamma_{MDS}}{\gamma_{MDS}}\times 100(\%)$ is plotted against $n$.}
\end{figure}

\newpage
\chapter{General shapes}
The fact that even multopole moments were zero for the hemisphere, hemi-ellipsoid on a plate, and the HCP shape, is no coincidence. It turns out, all shapes which are axially symmetric, and have mirror symmetry in the xy-plane, will have all even multipoles zero. In this chapter, this will be proved, and the consequences explored.

\section{Multipole coefficients}
\begin{theorem} \label{theorem_even_multipoles}
Let an emitter shape $S$ under a external electrostatic field $E_0$. Let $S$ be described in spherical coordinates by $r(\theta, \phi)$.
If a shape $S$ has axial symmetry, that is $r(\theta, \phi) = r(\theta)$, independent of $\phi$, and, if $S$ has mirror symmetry in the xy-plane, that is $r(\theta) = r(\pi - \theta)$, then, under an uniform applied field $E_0$, the combined system will have all even multipoles zero.
\end{theorem}

\begin{proof}
To prove the stated theorem, it is suficient to sshow that $r(x)$ is even, $J_A(x)$ is even, and show the respective $I$ and $G$ integrals are zero. Doing conversion $x = \cos\theta$, then, $r(x) = r(-x)$, because $\cos(\pi - \theta) = -\cos\theta = -x$. Therefore, $r(x)$ is an even function.

A detailed derivation of $J$ can be found in the Appendix A. Defining $J = J_A r^2\sin\theta$, then:
\begin{equation}
J_A(\theta) = \sqrt{1 + \left(\frac{dr}{d\theta}\right)^2}
\end{equation}

Again, doing $x = \cos\theta$:
\begin{equation}
J_A(x) = \sqrt{1 + \left(\frac{dr}{dx}\frac{dx}{d\theta}\right)^2}
= \sqrt{1 + \left(-\frac{dr}{dx}\sqrt{1-x^2}\right)^2}
\end{equation}

That is:
\begin{equation}
J_A(x) = \sqrt{1 + (1-x^2)\left(\frac{dr}{dx}\right)^2}
\end{equation}

If $dr/dx$ is either odd or even, then $J_A(x)$ is even function. Because $r(x)$ is even, then $r'(x)$ has definite parity, namely:
\begin{equation}
r(x) = r(-x)
\quad\implies\quad
r'(x) = -r'(-x)
\end{equation}

Thus, if $r(x)$ is even, then $r'(x)$ is odd, then $r'(x)^2$ is even, then $J_A(x)$ is even.

About the $IG$ integrals, doing substitution $x=\cos\theta$, and recall $\theta\in[0, \pi]$ and $\phi\in[0, 2\pi]$. Because $S$ has axial symmetry, nothing does depend on $\phi$, thus the $\phi$ integral evaluates to $2\pi$. One is left with:
\begin{equation}
\begin{split}
G_l &= \int_S r^{-l} P_l(\cos\theta)\cos\theta dS = 2\pi\int_{-1}^1 r^{-l+2}(x) P_l(x) x J_A(x) dx\\
I_{ij} &= \int_S r^{-i-j-2} P_i(\cos\theta) P_j(\cos\theta) dS = 2\pi\int_{-1}^1 r^{i+j}(x) P_i(x) P_j(x) J_A(x) dx
\end{split}
\end{equation}

Because $r(x)$ is even, then $r^n(x)$ is even, for any $n\in\mathbb{N}$. Therefore, because the symmetrical intervals $[-1, 1]$, and because $P_n(x)$ has definite parity, the entire integrand on both integrals have definite parity. If integrand is odd, integral is zero. Then $G_{2l} = 0$. And, if $i+j$ is odd, then $I_{ij} = 0$. The linear system (\ref{linear_system}) does reduces to the same as equation (\ref{spheroidal_linear_system_zero}) and (\ref{spheroidal_linear_system_nonzero}), and therefore, $A_{2l} = 0$. Because $A_l$ is proportional with the axial multipoles as in equation (\ref{Ql_Al}), then, $Q_{2l} = 0$, finishing the proof.
\end{proof}

\begin{theorem} \label{theorem_axial_potential}
Let an emitter shape $S$ under an external electrostatic field $E_0$. Let $S$ be described in spherical coordinates by $r(\theta, \phi)$.
If a shape $S$ has axial symmetry, that is $r(\theta, \phi) = r(\theta)$, independent of $\phi$, then the local potential is axially symmetric, that is, $V(r, \theta, \phi) = V(r, \theta)$.
\end{theorem}

\begin{proof}
In general, the potential can be written:
\begin{equation} \label{harmonic_potential}
V(r, \theta, \phi) = E_0 r\cos\theta + \sum_{l=0}^\infty\sum_{m=-l}^l B_{lm}r^{-(l+1)} Y_{lm}(\theta, \phi)
\end{equation}

At the surface, $V=0$. Because the shape is axially symmetric, then:
\begin{equation}
V(r(\theta), \theta, \phi) = 0 = E_0 r\cos\theta + \sum_{l=0}^\infty\sum_{m=-l}^l B_{lm}r(\theta)^{-(l+1)} Y_{lm}(\theta, \phi)
\end{equation}

Because $Y_{lm}(\theta, \phi) = N_{lm} P_{lm}(\cos\theta) e^{im\phi}$, then:
\begin{equation} \label{hehe1}
-E_0 r\cos\theta = \sum_{l=0}^\infty r(\theta)^{-(l+1)} \sum_{m=-l}^l N_{lm} B_{lm} P_{lm}(\cos\theta) e^{im\phi},\quad\forall\theta\in [0, \pi]
\end{equation}

Therefore, if we integrate both sides on $\phi$:
\begin{equation}
-2\pi E_0 r\cos\theta = \sum_{l=0}^\infty r(\theta)^{-(l+1)} \sum_{m=-l}^l N_{lm} B_{lm} P_{lm}(\cos\theta) \int_0^{2\pi}e^{im\phi}d\phi
\end{equation}

Using $\int_0^{2\pi} e^{im\phi} d\phi = 2\pi\delta_{m0}$, then:
\begin{equation} \label{hehe2}
-E_0 r\cos\theta = \sum_{l=0}^\infty r(\theta)^{-(l+1)} N_{l0} B_{l0} P_{l0}(\cos\theta), \quad\forall\theta\in [0, \pi]
\end{equation}

Comparing (\ref{hehe1}) with (\ref{hehe2}), we conclude $N_{lm} = 0$ if $m\neq 0$, because $r(\theta)$ is an arbitrary function. This means, potential (\ref{harmonic_potential}) reduces to (\ref{potential}), meaning, $V(r, \theta, \phi) = V(r, \theta)$, thus, the entire system is axially symmetric. Such result also justifies the usage of (\ref{potential}).
\end{proof}

\begin{corollary} \label{theorem_axial_sigma}
Let an emitter shape $S$ under an external electrostatic field $E_0$. Let $S$ be described in spherical coordinates by $r(\theta, \phi)$.
If a shape $S$ has axial symmetry, that is $r(\theta, \phi) = r(\theta)$, independent of $\phi$, then the surface charge density $\sigma(\theta, \phi)$ distribution has axial symmetry $\sigma(\theta, \phi) = \sigma(\theta)$.
\end{corollary}

\begin{proof}
If $V$ is axially symmetric, $\mathbf E = -\nabla V$ also is, and $\sigma = \epsilon_0 \mathbf E \cdot\mathbf n$, where $\mathbf n$ is the normal unit vector of S. Therefore, $\sigma$ must be axially symmetric.

\end{proof}

\begin{theorem} \label{theorem_odd_sigma}
Let an emitter shape $S$ under an external electrostatic field $E_0$. Let $S$ be described in spherical coordinates by $r(\theta, \phi)$.
If a shape $S$ has axial symmetry, that is $r(\theta, \phi) = r(\theta)$, independent of $\phi$, and, if $S$ has mirror symmetry in the xy-plane, that is $r(\theta) = r(\pi - \theta)$, then the surface charge density $\sigma(\theta, \phi)$ distribution both has axial symmetry $\sigma(\theta, \phi) = \sigma(\theta)$, and, opposite mirror symmetry: $\sigma(\theta) = -\sigma(\pi - \theta)$.
\end{theorem}

\begin{proof}
By theorem (\ref{theorem_axial_sigma}), because $S$ is axially symmetric, then $\sigma$ is also axially symmetric. This means, the multipole moments $Q_l$ can be integrated directly on $\phi$ yielding $2\pi$:

\begin{equation}
Q_l = \int_S \sigma(\theta) r^{l} P_l(\cos\theta) dS =
2\pi\int_{-1}^1 \sigma(x) r(x)^{l+2} P_{l}(x) J_A(x) dx
\end{equation}

In which substitution $x=\cos\theta$ was done. Notice that, expression above shows that $Q_l$ are proportional to the coefficients of the Legendre expansion of the function $\sigma(x) r(x)^{l+2} J_A(x)$, that is:
\begin{equation}
r(x)^{l+2} J_A(x) \sigma(x) = \sum_{n=0}^\infty\frac{2Q_l}{2l+1} P_l(x)
\end{equation}

By previous theorem, $Q_{2l} = 0$, it means that only odd Legendre polynomials contribute to the summation, meaning $r(x)^{l+2} J_A(x) \sigma(x)$ must be an odd function. Because it is already known that $r(x)$ and $J_A(x)$ are even (see proof of previous theorem), then, $\sigma(x)$ is odd, that is $\sigma(x) = -\sigma(-x)$. Because $\sigma(x)$ is odd, and $-x = -\cos\theta = \cos(\pi - \theta)$, then, $\sigma(\pi - \theta) = -\sigma(\theta)$.
\end{proof}

\begin{theorem} \label{theorem_line_of_charge}
The potential $V$ of an axially-symmetric mirror-symmetric shape $S$ may be approximated with arbitrary precision by a line of charge in the symmetry-axis along interval $[-a, a], a>0$, with linear charge density $\lambda(z)$.
\end{theorem}

\begin{proof}
The multipole moments of an axially symmetric surface, with a surface charge density $\sigma(\theta)$, can be calculated as:
\begin{equation}
Q_l = \int_S \sigma(\theta) r^{l} P_l(\cos\theta) dS =
2\pi\int_{-1}^1 \sigma(x) r(x)^{l+2} P_{l}(x) J_A(x) dx
\end{equation}

The multipole moments of a line of charge from $z\in [-a, a]$ with linear charge density $\lambda(z)$ can be calculated as:
\begin{equation}
M_l = \int_{-a}^a \lambda(z) z^l dz
= \int_{-1}^1 \lambda(ax) (ax)^l adx
= 2a^{l+1}\int_{0}^1 \lambda(ax) x^l dx
\end{equation}

Choosing $\lambda$ such that $M_l = Q_l$ for all $l$, means the response from the line of charge is identical to the response of the surface $S$.

Let $\lambda(ax) = \Lambda(x)$. Under a polynomial approximation:
\begin{equation}
\Lambda(z) = \sum_n a_n x^n
\quad\implies\quad
m_l = \int_0^1 \Lambda(x) x^l dx = \sum_{k=0}^{n-1}\frac{1}{k+l+1} a_k
\end{equation}

In other words:
\begin{equation}
\begin{bmatrix}
m_0 \\
m_1 \\
m_2 \\
\vdots \\
m_n \\
\end{bmatrix}
=
\begin{bmatrix}
1 & 1/2 & 1/3 & 1/4 & \cdots & 1/n \\
1/2 & 1/3 & 1/4 & 1/5 & \cdots & 1/(n+1) \\
1/3 & 1/4 & 1/5 & 1/6 & \cdots & 1/(n+2) \\
\vdots & \vdots & \vdots & \vdots & \ddots & \vdots\\
\frac{1}{n} & \frac{1}{n+1} & \frac{1}{n+2} & \frac{1}{n+3} & \cdots & \frac{1}{2n-1}\\
\end{bmatrix}
\begin{bmatrix}
a_0 \\
a_1 \\
a_2 \\
\vdots \\
a_n \\
\end{bmatrix}
\end{equation}

Such matrix, is known as the Hilbert Matrix, and it is known to be an ill-posed problem. Nevertheless, solving for larger and larger $n$ should yield more accurate $\Lambda$. Therefore, in theory, one can find $\Lambda$ with arbitrary precision.

% It is also worth to point out that, such a system is equivalent to the problem of finding a polynomial approximation for $\Lambda$ using least squares, in the case $\Lambda$ is known.
\end{proof}

\section{Two interacting shapes}
Let two shapes $S_1$ and $S_2$, axially symmetric, mirror symmetric on the xy-plane. By Theorem (\ref{theorem_even_multipoles}), the monopole charge is zero, because monopole is the zeroth multipole, and zero is even. Therefore, the next contributing multipole factor, is the dipole. If these two shapes are in a distance big enough as compared with their own lengths, then, the leading multipole contribution is the dipole. By theorem (\ref{theorem_odd_sigma}), it was proven that $\sigma$ must have xy-symmetry. However, it is possible to choose $\sigma$ such that dipole is zero (thus, the next contributing term, would be the octopole). So far, nothing prevents a special shape from having zero (or close to zero) dipole moments, and high octopole moments.

Consider $S_1$ and $S_2$ emitter shapes with height $h'$ and $h''$, axial dipole moments $Q_1'$ and $Q_1''$, both nonzero, and FEFs $\gamma_a'$ and $\gamma_a''$. Therefore, if both shapes interact independently of each other, then, they will feel each other's electric dipole. It will be assumed shape $S_1$ is centered at position $(0, 0, 0)$ and shape $S_2$ is at position $(c, 0, 0)$.
\begin{equation}
\mathbf E(\mathbf r) = \frac{3(\mathbf p\cdot\mathbf{\hat r})\mathbf{\hat r} - \mathbf p}{4\pi\epsilon_0 r^3},
\quad\quad
\gamma_a^{(\text{new})} = \gamma_a + \frac{\mathbf E(d, 0, h)\cdot\mathbf{\hat z}}{E_0}
\end{equation}

Thus:
\begin{equation}
\begin{split}
\gamma_a'^{(\text{new})} &= \gamma_a' + \frac{Q_1''}{E_0}\frac{1}{4\pi\epsilon_0(c^2 + h'^2)^{3/2}}
\left[
\frac{3h'^2}{h'^2 + c^2} - 1
\right]
\\
\gamma_a''^{(\text{new})} &= \gamma_a'' + \frac{Q_1'}{E_0}\frac{1}{4\pi\epsilon_0(c^2 + h''^2)^{3/2}}
\left[
\frac{3h''^2}{h''^2 + c^2} - 1
\right]
\end{split}
\end{equation}

In other words:
\begin{equation} \label{delta_general_shapes}
\begin{split}
\delta' &=
\frac{\gamma_a'^{(\text{new})} - \gamma_a'}{\gamma_a'}
=-\frac{1}{\gamma_a'}
\frac{\tilde A_1''}{(c^2 + h'^2)^{3/2}}
\left[
1 -
\frac{3h'^2}{h'^2 + c^2}
\right]
\\
\delta'' &=
\frac{\gamma_a''^{(\text{new})} - \gamma_a''}{\gamma_a''}
=-\frac{1}{\gamma_a''}
\frac{\tilde A_1'}{(c^2 + h''^2)^{3/2}}
\left[
1 -
\frac{3h''^2}{h''^2 + c^2}
\right]
\end{split}
\end{equation}

Therefore, the fractional change in the apex FEF, $\delta$, falls as $-\delta \sim K c^{-3}$ for the general shapes. More, the pre-factor $K$ depends only on the geometry of the emitter shape, since $\gamma_a$ and $A_1$ depends only on geometry.

It is worth pointing out that, if $\gamma_a$ depends only on the aspect ratio term by term, then $A_1$ (in fact, all $A_l$) will not depend on the aspect ratio, rather, will depend on $h$ and $R$ in general. This comes from Eq. (\ref{FEF}), in which, by hypothesis, $A_1/h^3$ depends only on $\nu$.
\begin{equation}
\frac{\tilde A_1}{h^3} = f(\nu)
\quad\implies\quad
\tilde A_1 = h^3 f(\nu)
\end{equation}

This can be seen, for instance, in the spheroidal case in Eq. (\ref{spheroidal_A1_value}), and even in the HCP case in Eq. (\ref{HCP_A1_value}). Also, it can be seen from (\ref{Ql_estimation_multipole}), where $a$ is hypothesized to be much more sensitive on $h$ than the base radius $R$ of a given shape. This also shows that the multipole moments $Q_l$ are not expected to depend only on the aspect ratio $\nu$.

It was assumed that electrostatic interactions are negligible (independent shapes), that is, the interaction of both shapes is such, that, the surface charge distribution of them is not disturbed. A natural question that rises, is if such approximation is indeed valid, and when the charge distributions begin to change. That will be done in the next section.

\section{Potential Theory}
The only way to check electrostatic interactions, is by fully solving Laplace equation. That will be done here. The system to be solved:
\begin{equation}
\nabla^2 V = 0,\quad\quad
\begin{matrix}
& V(\mathbf r') = 0, &\forall\mathbf r'\in S,\\
& -\nabla V(\mathbf r') = \mathbf E_0, &\text{ if }|\mathbf r'|\to\infty
\end{matrix}
\end{equation}

This one is slightly harder to be solved analytically by the methods of potential theorem. However, notice above problem is equivalent to:
\begin{equation} \label{Laplace_problem}
\nabla^2 \tilde V = 0,\quad\quad
\tilde V(\mathbf r') = E_0 r\cos\theta, \quad\forall\mathbf r'\in S,\\
\end{equation}

Where solution is $V(r, \theta) = -E_0 r\cos\theta + \tilde V(r, \theta)$. This indeed solves boundary conditions because:
\begin{equation}
\left. V(r, \theta)\right|_S = -E_0 r\cos\theta + \left.\tilde V(r, \theta)\right|_{S} =
-E_0 r(\theta) \cos\theta + E_0 r(\theta) \cos\theta = 0
\end{equation}

The problem (\ref{Laplace_problem}) can be solved by a Fredholm integral equation of the second kind \cite{fredholm}, written below:
\begin{equation} \label{fredholm_integral_equation}
E_0 r\cos\theta = \frac{1}{2}h(\mathbf r) + \int_S h(\mathbf r')\frac{\partial\Phi}{\partial\nu'}\left(\mathbf r - \mathbf r'\right) dS(\mathbf r')
\end{equation}

Where the solution is given by:
\begin{equation} \label{double_layer_solution}
\psi(\mathbf r) = -\int_S h(\mathbf r')\frac{\partial\Phi}{\partial\nu'}\left(\mathbf r - \mathbf r'\right) dS(\mathbf r'),
\quad\quad
\Phi(\mathbf r) = \frac{1}{4\pi}\frac{1}{|\mathbf r|} = \frac{1}{4\pi r}
\end{equation}

If one solves (\ref{fredholm_integral_equation}) for $h(\mathbf r)$ (this function should not be confused with the height of a emitter, a mere number), one can plug $h(\mathbf r)$ in (\ref{double_layer_solution}), and the solution $\psi$ is known. The single and double layer kernels for three dimensional systems are:
\begin{equation} \label{newtonian_kernel}
\Phi(\mathbf r) = \frac{1}{4\pi}\frac{1}{|\mathbf r|} = \frac{1}{4\pi r},\quad\quad
\nabla\Phi = -\frac{1}{4\pi}\frac{\mathbf r}{|\mathbf r|^3}
\end{equation}

Integral Equation (\ref{fredholm_integral_equation}) can be solved by means of Liouville–Neumann series, which basically consists in noticing $h$ appears both outside and both inside the integral. Thus, isolating $h$:
\begin{equation}
\begin{split}
h(\mathbf r) &=
2E_0 r\cos\theta
-2\int_S h(\mathbf r')\frac{\partial\Phi}{\partial\nu'}\left(\mathbf r - \mathbf r'\right) dS(\mathbf r') \\
&=
2E_0 r\cos\theta
-2\int_S h(\mathbf r')\frac{\mathbf r - \mathbf r'}{|\mathbf r - \mathbf r'|^3}\cdot\mathbf n' dS(\mathbf r') \\
\end{split}
\end{equation}

We substitute above equation, in the $h$ inside the integral:
\begin{equation}
\begin{split}
h(\mathbf r) &=
2E_0 r\cos\theta \\
&-2\int_S
\left[
2E_0 r'\cos\theta'-2\int_S h(\mathbf r'')\frac{\partial\Phi}{\partial\nu''}\left(\mathbf r' - \mathbf r''\right) dS(\mathbf r'')
\right]
\frac{\partial\Phi}{\partial\nu'}\left(\mathbf r - \mathbf r'\right) dS(\mathbf r')
\end{split}
\end{equation}

Or:
\begin{equation}
\begin{split}
h(\mathbf r) &= 2E_0 r\cos\theta - 2\int_S 2E_0 r'\cos\theta'
\frac{\partial\Phi}{\partial\nu'}\left(\mathbf r - \mathbf r'\right) dS(\mathbf r') \\
&+ 4\int_S \int_{S'} h(\mathbf r'')
\frac{\partial\Phi}{\partial\nu''}\left(\mathbf r' - \mathbf r''\right)
\frac{\partial\Phi}{\partial\nu'}\left(\mathbf r - \mathbf r'\right) dS(\mathbf r'') dS(\mathbf r')
\end{split}
\end{equation}

Or:
\begin{equation}
\begin{split}
h(\mathbf r) &= 2E_0 r\cos\theta - 4\int_S E_0 r'\cos\theta'
\frac{\mathbf r - \mathbf r'}{|\mathbf r - \mathbf r'|^3}\cdot\mathbf n' dS(\mathbf r') \\
&+ 4\int_S \int_{S'} h(\mathbf r'')
\left[\frac{\mathbf r' - \mathbf r''}{|\mathbf r' - \mathbf r''|^3}\cdot\mathbf n''\right]
\left[\frac{\mathbf r - \mathbf r'}{|\mathbf r - \mathbf r'|^3}\cdot\mathbf n'\right]
dS'(\mathbf r'') dS(\mathbf r')
\end{split}
\end{equation}

Plugging $h$ again, we get:
\begin{equation}
\begin{split}
h(\mathbf r) &= 2E_0 r\cos\theta - \int_S 4E_0 r'\cos\theta'
\frac{\mathbf r - \mathbf r'}{|\mathbf r - \mathbf r'|^3}\cdot\mathbf n' dS(\mathbf r') \\
&+ 4\int_S \int_{S'} E_0 r\cos\theta
\left[\frac{\mathbf r' - \mathbf r''}{|\mathbf r' - \mathbf r''|^3}\cdot\mathbf n''\right]
\left[\frac{\mathbf r - \mathbf r'}{|\mathbf r - \mathbf r'|^3}\cdot\mathbf n'\right]
dS'(\mathbf r'') dS(\mathbf r') \\
&- 4\int_S \int_{S'} \int_{S''}  h(\mathbf r'')
\left[\frac{\mathbf r'' - \mathbf r'''}{|\mathbf r'' - \mathbf r'''|^3}\cdot\mathbf n'''\right]
\left[\frac{\mathbf r' - \mathbf r''}{|\mathbf r' - \mathbf r''|^3}\cdot\mathbf n''\right]
\left[\frac{\mathbf r - \mathbf r'}{|\mathbf r - \mathbf r'|^3}\cdot\mathbf n'\right]
dS'(\mathbf r''') dS(\mathbf r'') dS(\mathbf r')
\end{split}
\end{equation}

Plugging $h$ again, and doing so iteratively, one arrives at Liouville–Neumann series, and a complete solution for $h$.
\begin{equation} \label{liouville_neumann_series}
h(\mathbf r) =
E_0 r\cos\theta +
\sum_{n=1}^\infty (-2)^n
\int_{S^{(1)}}
\int_{S^{(2)}}
\cdots
\int_{S^{(n)}}
\left[
\quad\cdots\quad
\right]
\prod_{k=1}^n dS(\mathbf r^{(k)})
\end{equation}

It is important to say, no attempt have been made to prove that above series converges for some shape.

\subsection{Interacting electrostatic systems}
Now, it is easy to find the contribution of interacting systems, by simply considering $S = S_1\cup S_2$, where $S_1$ is a shape, and $S_2$ is another shape. Let $h_1$ be the resolvent kernel of the isolated shape $S_1$, and $h_2$ for the isolated shape $S_2$. Their combined resolvent kernel of $S$ will be, as written by a Liouville-Neumann series:
\begin{equation}
\begin{split}
h(\mathbf r) &= h_1(\mathbf r) + h_2(\mathbf r) \\
&+2\int_{S_1}\int_{S_2}
E_0 r'\cos\theta'
\left[\frac{\mathbf r' - \mathbf r''}{|\mathbf r' - \mathbf r''|^3}\cdot\mathbf n''\right]
\left[\frac{\mathbf r - \mathbf r'}{|\mathbf r - \mathbf r'|^3}\cdot\mathbf n'\right]
dS_1(\mathbf r'') dS_2(\mathbf r') \\
&+2\int_{S_2}\int_{S_1}
E_0 r'\cos\theta'
\left[\frac{\mathbf r' - \mathbf r''}{|\mathbf r' - \mathbf r''|^3}\cdot\mathbf n''\right]
\left[\frac{\mathbf r - \mathbf r'}{|\mathbf r - \mathbf r'|^3}\cdot\mathbf n'\right]
dS_2(\mathbf r'') dS_1(\mathbf r') \\
&+\cdots
\end{split}
\end{equation}

Notice that, if both shapes are separated by a distance $c$, then, about the first interacting integral $W_{12}$, calculated at $\mathbf r\in S_1$, because, ultimately, objective is the apex FEF.
\begin{equation}
\begin{split}
|W_{12}| &= \left|2\int_{S_1}\int_{S_2}
E_0 r'\cos\theta'
\left[\frac{\mathbf r' - \mathbf r''}{|\mathbf r' - \mathbf r''|^3}\cdot\mathbf n''\right]
\left[\frac{\mathbf r - \mathbf r'}{|\mathbf r - \mathbf r'|^3}\cdot\mathbf n'\right]
dS_1(\mathbf r'') dS_2(\mathbf r')\right| \\
&\le
2\int_{S_1}\int_{S_2}
E_0 z'
\left|\frac{\mathbf r' - \mathbf r''}{|\mathbf r' - \mathbf r''|^3}\cdot\mathbf n''\right|
\left|\frac{\mathbf r - \mathbf r'}{|\mathbf r - \mathbf r'|^3}\cdot\mathbf n'\right|
dS_1(\mathbf r'') dS_2(\mathbf r') \\
&\le
2\int_{S_1}\int_{S_2}
E_0 z'
\left|\frac{\mathbf r' - \mathbf r''}{|\mathbf r' - \mathbf r''|^3}\right|
\left|\frac{\mathbf r - \mathbf r'}{|\mathbf r - \mathbf r'|^3}\right|
dS_1(\mathbf r'') dS_2(\mathbf r') \\
&\le
2\int_{S_1}\int_{S_2}
E_0 z'
\sup\left\{
\left|\frac{1}{|\mathbf r' - \mathbf r''|^2}\right|, \forall\mathbf r'\in S_1, \forall\mathbf r''\in S_2\right\}^2
dS_1(\mathbf r'') dS_2(\mathbf r') \\
&=
2 E_0 z_2
\int_{S_1}\int_{S_2}
\left[
\frac{1}{s^2 + (c-t)^2}
\right]^2
dS_1(\mathbf r'') dS_2(\mathbf r') \\
&=
2 E_0 z_2
\frac{1}{(s^2 + (c-t)^2)^2}
\int_{S_1}\int_{S_2}
dS_1(\mathbf r'') dS_2(\mathbf r') \\
&=
\frac{2E_0 z_2 A_1 A_2}{(s^2 + (c-t)^2)^2}
\end{split}
\end{equation}

In this case, $A_1$ and $A_2$ (not to be confused with coefficients $A_l$ of the Legendre expansion of the potential) are the surface areas of the shapes $S_1$ and $S_2$ respectively. In addition, $z_2$ is the maximum $z$ distance of shape $S_2$. Also, it was considered that:
\begin{equation}
\frac{1}{s^2 + (c-t)^2} = \sup\left\{
\left|\frac{1}{|\mathbf r' - \mathbf r''|^2}\right|, \forall\mathbf r'\in S_1, \forall\mathbf r''\in S_2\right\}
\end{equation}

The same thing can be done for $W_{21}$:
\begin{equation}
\begin{split}
|W_{21}| &= \left|2\int_{S_2}\int_{S_1}
E_0 r'\cos\theta'
\left[\frac{\mathbf r' - \mathbf r''}{|\mathbf r' - \mathbf r''|^3}\cdot\mathbf n''\right]
\left[\frac{\mathbf r - \mathbf r'}{|\mathbf r - \mathbf r'|^3}\cdot\mathbf n'\right]
dS_2(\mathbf r'') dS_1(\mathbf r')\right| \\
&\le
2\int_{S_2}\int_{S_1}
E_0 z'
\left|\frac{\mathbf r' - \mathbf r''}{|\mathbf r' - \mathbf r''|^3}\cdot\mathbf n''\right|
\left|\frac{\mathbf r - \mathbf r'}{|\mathbf r - \mathbf r'|^3}\cdot\mathbf n'\right|
dS_2(\mathbf r'') dS_1(\mathbf r') \\
&\le
2\int_{S_2}\int_{S_1}
E_0 z'
\sup\left\{\left|\frac{1}{|\mathbf r' - \mathbf r''|^2}\right|\right\}
\left|\frac{\mathbf r - \mathbf r'}{|\mathbf r - \mathbf r'|^3}\cdot\mathbf n'\right|
dS_2(\mathbf r'') dS_1(\mathbf r') \\
&=
2 E_0 z_1
\left[
\frac{1}{s^2 + (c-t)^2}
\right]
\int_{S_2}\int_{S_1}
\left|\frac{\mathbf r - \mathbf r'}{|\mathbf r - \mathbf r'|^3}\cdot\mathbf n'\right|
dS_2(\mathbf r'') dS_1(\mathbf r') \\
&=
2E_0 z_1
\frac{1}{s^2 + (c-t)^2}
\int_{S_2}
|k| dS_2(\mathbf r') \\
&=
\frac{2|k|E_0 z_1 A_2}{s^2 + (c-t)^2}
\end{split}
\end{equation}

A useful result is Gauss' lemma, which states:
\begin{equation}
k' = \int_{S}
\left[\frac{\mathbf r - \mathbf r'}{|\mathbf r - \mathbf r'|^3}\cdot\mathbf n'\right]
dS(\mathbf r')
=
\int_S \frac{\partial\Phi}{\partial n}(\mathbf r - \mathbf r')dS(\mathbf r')
\end{equation}

Where, if $S = \partial V$, then:
\begin{equation}
\left\{
\begin{array}{ll}
k' = 0, &\quad\text{if } \mathbf r\in V \\
k' = -1/2, &\quad\text{if } \mathbf r\in\partial V \\
k' = -1, &\quad\text{if } \mathbf r\in V^c \\
\end{array}
\right.
\end{equation}

However, because the integrand is an absolute value, the sign of $(\mathbf r - \mathbf r')\cdot\mathbf n'$ is relevant. If there's no change in sign in the entire integration domain, then the integral will yield $|k| = |k'|$, as in Eq (5.3.18). If there's change in sign, then, nothing can be said.

If $\mathbf r$ is taken immediately above the protrusion, then $\mathbf r\in V^c$. However, there is a change in sign, as $(\mathbf r - \mathbf r')\cdot\mathbf n' > 0$ if $\mathbf r'$ is taken at the apex, and $(\mathbf r - \mathbf r')\cdot\mathbf n' < 0$ if $\mathbf r'$ is taken at the tip of the lower hemisphere. Because of that, one cannot use Gauss' lemma in this case, and show that $k = 0$. In this case, $|k|$ is non-zero.

If $\mathbf r$ is taken immediately below the protrusion (that is, $\mathbf r\in V$), or if $\mathbf r$ is exactly at the apex (that is, $\mathbf r\in\partial V$), intuitively there's no sign change, and Gauss' lemma can be used, yielding $|k|\neq 0$ as well.

Therefore, first order electrostatic interaction term is:
\begin{equation} \label{electrostatic_interacting_integrals}
|W| \le |W_{12}| + |W_{21}|
%\le \frac{2E_0 A_1 A_2}{(s^2 + (c-t)^2)^2}(z_1 + z_2)
\le\frac{2E_0 z_2 A_1 A_2}{(s^2 + (c-t)^2)^2}
+\frac{2|k|E_0 z_1 A_2}{s^2 + (c-t)^2}
\sim\frac{1}{c^4} + \frac{|k|}{c^2}
\end{equation}

If $k=0$, the resolvent falls at least as $c^{-4}$. In this case, because electric dipole approximation is of $\delta\sim c^{-3}$, and electrostatic interactions become important with $c^{-3}$, thus, expressions (5.2.3) are expected to hold.

But, if $k\neq 0$, it was shown the resolvent falls at least as $c^{-2}$, therefore, nothing can be told about if independent dipole-dipole is a good approximation, precisely because all expressions obtained for $W$ are upper bounds.

\newpage
\chapter{Conclusion}
The electrostatic potential was assumed to be of the form of Eq. (\ref{potential}), which is true for a surface $S$ with axial symmetry (see theorem (\ref{theorem_axial_potential})). Then, an error function $\Sigma$ was defined in order to minimize the errors from boundary conditions. A necessary condition was found at Eq. (\ref{linear_system}), on which, the parameters $\tilde A_l$, $G_l$ and $I_{ij}$ were related, and only depended on the boundary $S$. Several analytical conclusions could be drawn from this set up.

The coefficients $A_l$ related to the multipole moments $Q_l$ of the system by means of equation (\ref{Ql_tilde_Al}), in which, was shown that for a general axially symmetric shape, all multipole moments scale linearly with the applied external electrostatic field $E_0$. That is, if one doubles the field, all multipole moments will double, as expected.

At the case of a hemisphere on a plate, the system (\ref{linear_system}) was solved exactly, giving the known potential (\ref{potential_spherical_case}) for a sphere, which a known apex FEF $\gamma_a = 3$ (\ref{FEF_spherical_case}).

Then, it was applied to the hemi-ellipsoid model on a plate. It was shown by equation (\ref{spheroidal_linear_system_zero}) that all even multipole contributions of the system are zero. It was possible to obtain an expression for the FEF for $h\ll R$ at equation (\ref{FEF_prolate_spheroidal_case}). Numerical calculations have been done, by truncating the linear system and attempting to approximate the values of $\tilde A_l$. Unfortunately, it was found that convergence is slow.

The method was then applied to the hemisphere on a cylindrical post (HCP) model. As in the case of the prolate spheroid, it was shown by the same way, that all even multipole contributions (image charge, quadrupole, etc) of a HCP shape over applied field are zero, thus, only odd multipoles contribute (dipole, octopole, etc). This had several consequences: it gave information of how the induced surface charge density must behave, in particular, it was shown in theorem (\ref{HCP_odd_sigma_x}) that $\sigma(x) + \sigma(-x) = 0$, that is, the charge density must be an odd function. Both multipole and $\sigma$ restrictions was generalized to shapes with axial symmetry, and mirror symmetry in the xy-plane, as given by theorems (\ref{theorem_even_multipoles}) and (\ref{theorem_odd_sigma}).

Another consequence is that other models attempting to approximate a HCP shape must comply the conditions over the multipoles, which imposes restrictions. For example, for a line of charge in the interval $z \in [-a, a]$ with linear charge density $\lambda(z)$, it was shown in theorem (\ref{linear_odd_lambda_z}) some encouragement to choose an odd function $\lambda(z)$ as well. By means of theorem (\ref{theorem_odd_sigma}), this holds also when modeling general shapes (as enunciated on the theorem) with a line of charge, as was done in \cite{india}.

Furthermore, theorem (\ref{theorem_line_of_charge}) was proven for a general emitter shape, showing that any axially symmetric shape can be approximated by the line of charge model, with arbitrary precision. A relation between the surface charge density $\sigma$ and the linear charge density $\lambda$ is shown in the theorem. This means, if one can find $\lambda$, one can also find $\sigma$, and vice versa.

In a HCP shape, because image charges are zero, the next non-zero contribution is of a dipole. It was therefore considered two HCP shapes, interacting via primitive dipole, and it was shown that the FEF decays by third power law with respect to the distances of the hemispheres (\ref{HCP_delta_third_power_law}), in agreement of numerical simulations \cite{Assis2018} and analytical models \cite{DallAgnol2018}. The same thing happens with one dimensional arrays as calculated at (\ref{HCP_delta_1D_array}). Such result was also generalized for two general axial symmetric and mirror symmetric shapes (not necessarily identical), where the fractional change in the apex FEF of both shapes was calculated at equation (\ref{delta_general_shapes}). Such equation shows that, $\delta \sim Kc^{-3}$, where $c$ is the distance between the emitters. It was also shown the pre-factor $K$ does depend on the geometry of the emitter shapes, confirming the tendency in recent analytical and numerical results \cite{Assis2018,DallAgnol2018}. The functional form of $K$ was shown in Eq. (\ref{pre_factor}).

It was shown by means of potential theory that the next contribution in the double layer resolvent kernel is bounded by equation (\ref{electrostatic_interacting_integrals}). This result shows the resolvent falls in at least $c^{-2}$, and therefore, nothing can be said about if the independent dipole-dipole approximation is a valid approximation, or the range where it is valid. However, other layers could have been chosen (say, simple layer, or, some other), and the analysis would be different. Furthermore, no proof was given that the resolvent converges, thus, such result should be used with caution.

As for perspectives from future work, the method could be extended for a general emitter (not necessarily axial symmetric). The potential would be written as in equation (\ref{harmonic_potential}). It is possible that some sort of $IG$ integrals could be found to minimize the potential on the surface. The entire treatment would be more complicated, but, it is possible that some interesting theorems can be proved. This could be investigated further. In addition, investigation about the role of potential theory in the context of field emission could be more explored.

% Future perspectives -- spherical harmonics expansion. IG integrals. Multipole theorem.
% Pehraps using BEM to solve FEF by means of hilbert series.

\newpage
\appendix
\chapter{Shapes in spherical coordinates}
\section{Cylinder in spherical coordinates}
The equation of a cylinder is $x^2 + y^2 = R^2$, and we know that, the radial coordinate is $r^2 = x^2 + y^2 + z^2 = R^2 + z^2$. However, because $z = r\cos\theta$, we have: $r^2 = R^2 + r^2\cos^2\theta$. We can isolate $r$, and get: $r^2(1-\cos^2\theta) = R^2$. Therefore, for the cylinder:
$$
r = \frac{R}{\sin\theta}
$$

That could be done quite obviously from another way: considering a triangle, where hypothenus is a vector from origin to a point in the cylinder, in such a way that $\theta$ is the azimuthal angle from the spherical coordinate system. From it, we can extract:
$$
\cos\theta = \frac{z}{r},\quad\quad
\sin\theta = \frac{R}{r},\quad\quad
\tan\theta = \frac{R}{z}
$$

We have in spherical coordinates:
$$
x = r\cos\phi\sin\theta = \frac{R}{\sin\theta}\cos\phi\sin\theta = R\cos\phi
$$

The same can be done for $y$ and $z$. And thus, we'll have a vector:
$$
\mathbf x = (r\cos\phi\sin\theta, r\sin\phi\sin\theta, r\cos\theta) = \left(R\cos\phi, R\sin\phi, \frac{R}{\tan\theta}\right).
$$

With it, we can figure out the normal vector of the cylinder:
$$
\frac{\partial\mathbf x}{\partial\phi}\times\frac{\partial\mathbf x}{\partial\theta} =
\begin{vmatrix}
\mathbf{\hat x} & \mathbf{\hat y} & \mathbf{\hat z} \\
-R\sin\phi & R\cos\phi & 0 \\
0 & 0 & \frac{-R}{\sin^2\theta}
\end{vmatrix}
=
\begin{bmatrix}
-R^2 \frac{\cos\phi}{\sin^2\theta} \\
-R^2 \frac{\sin\phi}{\sin^2\theta} \\
0
\end{bmatrix}
=
-\frac{R^2}{\sin^2\theta}
\begin{bmatrix}
\cos\phi \\
\sin\phi \\
0
\end{bmatrix}
$$

With the normal vector, we can finally figure out the area element: $dS = Jd\phi d\theta$. That is:
$$
J = \left\lVert\frac{\partial\mathbf x}{\partial\phi}\times\frac{\partial\mathbf x}{\partial\theta}\right\rVert
= \frac{R^2}{\sin^2\theta} = r^2
$$

Now we have all data necessary to integrate. We now seek the limits of integration. At the top of the cylinder, $r^2 = \ell^2 + R^2$. Henceforth:
$$
\cos\theta_0 = \frac{\ell}{r} = \frac{\ell}{\sqrt{\ell^2 + R^2}} = \frac{1}{\sqrt{1 + \left(\frac{R}{\ell}\right)^2}} = a_c
$$
$$
\sin\theta_0 = \frac{R}{r} = \frac{R}{\sqrt{\ell^2 + R^2}} = \frac{1}{\sqrt{1 + \left(\frac{\ell}{R}\right)^2}} = a_s
$$

And, lastly:
$$
\tan\theta_0 = \frac{R}{\ell}, \quad\quad
\frac{1}{\tan\theta_0} = \frac{\ell}{R}
$$

\section{General symmetrical shape in spherical coordinates}
For a general element of area, we need the equation of the shape itself, that is, $r = r(\theta, \phi)$. Because the shape is axially-symmetric, $r$ doesn't depend on $\phi$, and therefore, $r = r(\theta)$.

We'll now calculate the element of area. For that, we need to calculate several derivatives. Let us begin:
\begin{align*}
\frac{\partial x}{\partial\theta} &= \frac{\partial r}{\partial\theta}\cos\phi\sin\theta + r\cos\phi\cos\theta &= \frac{x}{r}\frac{\partial r}{\partial\theta} + z\cos\phi \\
\frac{\partial y}{\partial\theta} &= \frac{\partial r}{\partial\theta}\sin\phi\sin\theta + r\sin\phi\cos\theta &= \frac{y}{r}\frac{\partial r}{\partial\theta} + z\sin\phi \\
\frac{\partial z}{\partial\theta} &= \frac{\partial r}{\partial\theta}\cos\theta - r\sin\theta &= \frac{z}{r}\frac{\partial r}{\partial\theta} - r\sin\theta \\
\end{align*}

And now, derivating with respect with $\phi$.
\begin{align*}
\frac{\partial x}{\partial\phi} &= -r\sin\phi\sin\theta &= -y \\
\frac{\partial y}{\partial\phi} &= r\cos\phi\sin\theta &= +x \\
\frac{\partial z}{\partial\phi} &= 0 &= 0 \\
\end{align*}

The normal vector:
$$
\frac{\partial\mathbf x}{\partial\phi}\times\frac{\partial\mathbf x}{\partial\theta} =
\begin{vmatrix}
\mathbf{\hat x} & \mathbf{\hat y} & \mathbf{\hat z} \\
-y & x & 0 \\
\frac{x}{r}\frac{\partial r}{\partial\theta} + z\cos\phi &
\frac{y}{r}\frac{\partial r}{\partial\theta} + z\sin\phi &
\frac{z}{r}\frac{\partial r}{\partial\theta} - r\sin\theta \\

\end{vmatrix}
$$

Thus:
\begin{equation}
\frac{\partial\mathbf x}{\partial\phi}\times\frac{\partial\mathbf x}{\partial\theta} =
\begin{bmatrix}
x\left(\frac{z}{r}\frac{\partial r}{\partial\theta} - r\sin\theta\right) \\
y\left(\frac{z}{r}\frac{\partial r}{\partial\theta} - r\sin\theta\right) \\
-y\left(\frac{y}{r}\frac{\partial r}{\partial\theta} + z\sin\phi\right)
-x\left(\frac{x}{r}\frac{\partial r}{\partial\theta} + z\cos\phi\right) \\
\end{bmatrix}
\end{equation}

And, therefore, finally, the element of area is:
\begin{equation}
\begin{split}
J^2 &= \left(x^2 + y^2\right)\left(\frac{z}{r}\frac{\partial r}{\partial\theta} - r\sin\theta\right)^2 +
\left[y\left(\frac{y}{r}\frac{\partial r}{\partial\theta} + z\sin\phi\right)
+x\left(\frac{x}{r}\frac{\partial r}{\partial\theta} + z\cos\phi\right)\right]^2 \\
&= \left(x^2 + y^2\right)\left(\frac{z}{r}\frac{\partial r}{\partial\theta} - r\sin\theta\right)^2 + y^2\left(\frac{y}{r}\frac{\partial r}{\partial\theta} + z\sin\phi\right)^2 \\
&+x^2\left(\frac{x}{r}\frac{\partial r}{\partial\theta} + z\cos\phi\right)^2
+ 2xy\left(\frac{y}{r}\frac{\partial r}{\partial\theta} + z\sin\phi\right)\left(\frac{x}{r}\frac{\partial r}{\partial\theta} + z\cos\phi\right)
\end{split}
\end{equation}

We can expand the terms even more, and then, group them in powers of $z$, and factor out what we can:
\begin{equation}
\begin{split}
J^2 &= z^2\left[(x^2 + y^2)\left(\frac{1}{r}\frac{\partial r}{\partial\theta}\right)^2 + \left(x\cos\phi + y\sin\phi\right)^2\right] \\
&+ z\left[-2(x^2 + y^2)\sin\theta\frac{\partial r}{\partial\theta} + \left(\frac{1}{r}\frac{\partial r}{\partial\theta}\right)
\left(\sin\phi\left(2y^3 + 2x^2y\right) + \cos\phi\left(2x^3 + 2xy^2\right)\right)\right] \\
&+ \left[(x^2 + y^2) r^2\sin^2\theta + \left(\frac{1}{r}\frac{\partial r}{\partial\theta}\right)^2(x^2 + y^2)^2\right]
\end{split}
\end{equation}

And then we can finally arrive at:
\begin{equation}
\begin{split}
J^2 &= z^2\left[(x^2 + y^2)\left(\frac{1}{r}\frac{\partial r}{\partial\theta}\right)^2 + \left(x\cos\phi + y\sin\phi\right)^2\right] \\
&+ 2z(x^2 + y^2)\left[-\left(\sin\theta\frac{\partial r}{\partial\theta}\right) + \left(\frac{1}{r}\frac{\partial r}{\partial\theta}\right)\left(x\cos\phi + y\sin\phi\right)\right] \\
&+ \left(x^2 + y^2\right)^2\left[1 + \left(\frac{1}{r}\frac{\partial r}{\partial\theta}\right)^2\right]
\end{split}
\end{equation}

Now, we'll substitute $x,y,z$ by its values: $x = r\cos\phi\sin\theta, y = r\sin\phi\sin\theta, z=r\cos\theta$. And for that, some quantities might be useful:
\begin{equation}
\begin{split}
x^2 + y^2 &= r^2\sin^2\theta\left(\cos^2\theta + \sin^2\theta\right) = r^2\sin^2\theta \\
x\cos\phi + y\sin\phi &= r\cos^2\phi\sin\theta + r\sin^2\phi\sin\theta = r\sin\theta \\
(x\cos\phi + y\sin\phi)^2 &= (r\sin\theta)^2 = r^2 \sin^2\theta
\end{split}
\end{equation}

Now, we substitute, and, doing some more calculations, we arrive at this expression:
\begin{equation}
J^2 = r^4\sin^2\theta\left[1 + \left(\frac{1}{r}\frac{\partial r}{\partial\theta}\right)^2\right]
\quad\implies\quad
J = r^2\sin\theta\sqrt{1 + \left(\frac{1}{r}\frac{\partial r}{\partial\theta}\right)^2}
\end{equation}

Or, because $r$ doesn't depend on $\phi$, that is, depends exclusively on $\theta$, then:
\begin{equation} \label{general_area_element}
J = r^2\sin\theta\sqrt{1 + \left(\frac{1}{r}\frac{dr}{d\theta}\right)^2}
\end{equation}

\section{Prolate spheroid in spherical coordinates}
For a prolate spheroid (elongated ellipsoid of revolution), with radius $R$ and height $h$, obeys the quadric equation:
\begin{equation}
\frac{x^2}{R^2} + \frac{y^2}{R^2} + \frac{z^2}{h^2} = 1
\end{equation}

Because $r^2 = x^2 + y^2 + z^2$, it becomes:
\begin{equation}
\begin{split}
x^2 + y^2 + \frac{R^2}{h^2}z^2 = R^2 \\
x^2 + y^2 + z^2 + \left(\frac{R^2}{h^2}-1\right) z^2 = R^2 \\
r^2 + \left(\frac{R^2}{h^2}-1\right) z^2 = R^2 \\
r^2\left[1 + \left(\frac{R^2}{h^2}-1\right) \cos^2\theta\right] = R^2 \\
\end{split}
\end{equation}

Where it was used $z = r\cos\theta$. The equation of the prolate spheroid becomes:
\begin{equation} \label{spheroidal_r}
r = \frac{R}{\sqrt{1 - \epsilon^2\cos^2\theta}},
\quad\quad
\epsilon^2 = 1 - \frac{R^2}{h^2}
\end{equation}

In here, $\epsilon$ is the eccentricity of the revolution ellipse. Acceptable values of $\epsilon$ are: $0\le\epsilon < 1$, where $\epsilon = 0$ iff $R = h$, that is, the problem of the sphere which was solved. The derivative with respect to theta, becomes:
\begin{equation}
\frac{dr}{d\theta} = -\frac{R\epsilon^2\cos\theta\sin\theta}{\left[1 - \epsilon^2\cos\theta\right]^{3/2}}
\end{equation}

Thus:
\begin{equation} \label{spheroidal_drdtheta}
\frac{dr}{d\theta} = -\frac{r^3}{R^2}\epsilon^2\cos\theta\sin\theta
\end{equation}

Using (\ref{general_area_element}), we find for the spheroid:
\begin{equation} \label{spheroidal_area_element}
J = r^2\sin\theta\sqrt{1 + \frac{r^4}{R^4}\epsilon^4\cos^2\theta\sin^2\theta}
\end{equation}

\section{Suspended hemisphere in spherical coordinates}
The suspended hemisphere is much more complicated than the cylinder. The hemisphere is suspended over the cylinder, thus its center is located at $(0, 0, \pm\ell)$. The equation of both hemispheres are: $x^2 + y^2 + (z\mp\ell)^2 = R^2$. Again, we recall: $x^2 + y^2 + z^2 = r^2$, in spherical coordinates. Then:
\begin{align*}
x^2 + y^2 + (z\mp\ell)^2  &= R^2   \\
x^2 + y^2 +z ^2 \mp 2z\ell + \ell^2  &= R^2   \\
r^2 \mp 2z\ell + \ell^2  &= R^2   \\
\end{align*}

Therefore, using $z=r\cos\theta$, we get: $r^2 = R^2 -\ell^2 \pm 2r\ell\cos\theta$. The same result can be reached by a triangle, except this time, unlike the cylinder case, our triangle is no longer a right triangle. Yet, by cosine law, we can get at the same result.

Such a triangle becomes a right triangle in the limits of integration: $\phi\in[0, 2\pi]$ and $\theta\in[0, \theta_0]$. At $\theta=0$ we have $r = \ell + R = h$. At $\theta = \theta_0$, we have $r^2 = \ell^2 + R^2$. And, the $\theta_0$ can be calculated using such right triangle:
\begin{equation} \label{trig_hemi_limit}
\sin\theta_0 = \frac{R}{r_0}, \quad\quad
\cos\theta_0 = \frac{\ell}{r_0},\quad\quad
\tan\theta_0 = \frac{R}{\ell}
\end{equation}

The equation we got, $r^2 = R^2 -\ell^2 \pm 2r\ell\cos\theta$, is a quadratic equation, and can be solved analytically:
\begin{equation}
r = \pm\ell\cos\theta \pm \sqrt{R^2 - \ell^2\sin^2\theta}
\end{equation}

Where, the first $\pm$ is due to the location of the sphere (either upwards, or downwards), while the second $\pm$ is due to the $\pm\sqrt\Delta$ of the quadratic equation. Therefore, these are four equations, where only two of them are correct. To find out, we look at extreme values of $\theta$, namely, $0$ and $2\pi$.
\begin{equation}
\begin{split}
\theta = 0\quad &\implies\quad r = \pm\ell\pm R \\
\theta = \pi\quad &\implies\quad r = \mp\ell\pm R \\
\end{split}
\end{equation}

We know that, in both cases, the correct expression should be $r = \ell + R$. Therefore, we pick $++$. That leads us to:
\begin{equation} \label{suspended_hemispherical_r_unsigned_l}
\begin{split}
r = \ell\cos\theta + \sqrt{R^2 - \ell^2\sin^2\theta},\quad &\theta\le\frac{\pi}{2},\quad\ell\ge 0, \quad\text{(Upper Hemisphere)} \\
r = -\ell\cos\theta + \sqrt{R^2 - \ell^2\sin^2\theta},\quad &\theta\ge\frac{\pi}{2},\quad\ell\ge 0,\quad\text{(Lower Hemisphere)}\\
\end{split}
\end{equation}

Another way to interpret such result, is to have one unique formula:
\begin{equation} \label{suspended_hemispherical_r_signed_l}
r = \ell\cos\theta + \sqrt{R^2 - \ell^2\sin^2\theta},\quad\theta\in[0, \pi],\quad
\begin{tabular}{cc}
$\ell\ge 0,$ & $\text{(Upper Hemisphere)}$ \\
$\ell\le 0,$ & $\text{(Lower Hemisphere)}$
\end{tabular}
\end{equation}

\subsection{Element of area}
We already have equation (\ref{general_area_element}). Now, all we are lacking to do, is to evaluate the derivative. We could choose either to deal with unsigned $\ell$ as in equation (\ref{suspended_hemispherical_r_unsigned_l}, or with signed $\ell$, as in equation (\ref{suspended_hemispherical_r_signed_l}). Notice that, dealing with unsigned $\ell$, two expressions would be required for the derivative, while with signed $\ell$ expression, only one equation would be required. Choosing the signed expression, and calculating the derivative, one can get:
\begin{equation}
\frac{\partial r}{\partial\theta} = -\ell\sin\theta\left[1 + \frac{\ell\cos\theta}{\sqrt{R^2 - \ell^2\sin^2\theta}}\right]
\end{equation}

All we have to do now, is to insert this expression into the $J$ we had, and, then, finally:
\begin{equation} \label{hemispherical_J}
J = r^2\sin\theta\sqrt{1 + \frac{\ell^2}{r^2}\sin^2\theta\left(1 + \frac{\ell\cos\theta}{\sqrt{R^2 - \ell^2\sin^2\theta}}\right)^2}
\end{equation}

Don't forget that, $r$ still depends on $\theta$ by (\ref{suspended_hemispherical_r_signed_l}), so, there is still one substitution left. But, we're going to leave it at that.

\subsection{Approximation: $\ell\gg R$}
As we have noticed, the expression for the suspended hemispherical area element is quite complicated, and we seek to simplify doing approximations of these kind. For that, we seek in understanding how our variables, $\ell, \sin\theta, \cos\theta\, r$ behaves. For that, we seek our attention to (\ref{suspended_hemispherical_r_signed_l}). Notice that:
\begin{equation}
r = \underbrace{\ell\cos\theta}_{\text{Grows as} O(\ell)} \quad+\quad \underbrace{\sqrt{R^2 - \ell^2\sin^2\theta}}_{\text{Grows as } O(1)}
\end{equation}

Furthermore, with aid of (\ref{trig_hemi_limit}), we can come up with a more general relationship, which is valid for all $\ell$ (not only $\ell\gg R$), which is:
\begin{equation} \label{trig_hemi}
0 \le \sin\theta \le \frac{R}{r_0}, \quad\quad
\frac{\ell}{r_0} \le \cos\theta \le 1
\end{equation}

This helps us identify that, under approximation $\ell\gg R$, sine will keep itself very close to zero all the time (because $r_0\approx h$), and cosine will keep itself very close to one all the time.

This can be identified geometrically, from the triangle. We can do the same thing with the square root in expression (\ref{hemispherical_J}), in which we can expand, in order to determine the assymptoptics.

\begin{equation} \label{assymptoptics_Jggl}
J = r^2\sin\theta\sqrt{
\underbrace{1}_{J\in O(\ell^2)}
+\underbrace{\frac{\ell^2}{r^2}\sin^2\theta}_{J\in O(1)}
+\underbrace{\frac{\ell^2}{r^2}\frac{\ell^2 \cos^2\theta\sin^2\theta}{R^2 - \ell^2\sin^2\theta}}_{J\in O(\ell^2)}
+\underbrace{\frac{\ell^2}{r^2}\frac{2\ell\cos\theta\sin^2\theta}{\sqrt{R^2 - \ell^2\sin^2\theta}}}_{J\in O(\ell)}
}
\end{equation}

Now, we pick the highest asymptoptics, that is, $O(\ell^2)$, and neglect all others. We, thus, have:
\begin{equation} \label{hemispherical_J_ellggR}
J = r^2\sin\theta\sqrt{1 + \frac{\ell^2}{r^2}\frac{\ell^2 \cos^2\theta \sin^2\theta}{R^2 - h^2\sin^2\theta}}, \quad\ell\gg R
\end{equation}

Now we have a much simpler area element $J$: There's no longer a square root inside a square root, and a few other benefits.

\subsection{Approximation: $\ell\ll R$}
We can't use the same asymptoptics as equation $(\ref{assymptoptics_Jggl})$, because it was assumed $\ell\approx r$, which is clearly not true anymore under this approximation regime. Now, we have $r\approx R$. In addition, from (\ref{trig_hemi}), we conclude that cosine and sine will vary freely as $\theta$ changes from $0$ to $\theta_0$, precisely because $\theta_0$ is an angle close to $\pi/2$.

We that in mind, we can re-write equation (\ref{assymptoptics_Jggl}) for our approximation case:
\begin{equation} \label{assymptoptics_Jlll}
J = r^2\sin\theta\sqrt{
\underbrace{1}_{J\in O(1)}
+\underbrace{\frac{\ell^2}{r^2}\sin^2\theta}_{J\in O(\ell)}
+\underbrace{\frac{\ell^2}{r^2}\frac{\ell^2 \cos^2\theta\sin^2\theta}{R^2 - \ell^2\sin^2\theta}}_{J\in O(\ell^2)}
+\underbrace{\frac{\ell^2}{r^2}\frac{2\ell\cos\theta\sin^2\theta}{\sqrt{R^2 - \ell^2\sin^2\theta}}}_{J\in O(\ell^{3/2})}
}
\end{equation}

Considering only the $O(1)$ term, we recover $J = r^2\sin\theta$, the same for an sphere, especially if one considers $r\approx R$. Henceforth, we'll also pick the term $O(\ell)$. Therefore:
\begin{equation} \label{hemispherical_J_ellllR}
J = r^2\sin\theta\sqrt{1 + \frac{\ell^2}{r^2}\sin\theta},\quad\ell\ll R
\end{equation}

\newpage
\chapter{Spheroidal Integrals}
\section{G-Integrals}
Recalling (\ref{IG_integrals}):
$$
G_l = \int_S r^{-l} P_l(\cos\theta) \cos\theta dS
$$

Having the area element in (\ref{spheroidal_area_element}), together with (\ref{spheroidal_r}), the G-integral can be written as
\begin{equation} \label{spheroidal_Gl_thetaphi}
\dot G_l = \int_0^{2\pi}\int_0^\pi r^{-l} P_l(\cos\theta) \cos\theta\cdot r^2\sin\theta \cdot \sqrt{1 + \frac{r^4}{R^4}\epsilon^4\cos^2\theta\sin^2\theta} d\theta d\phi
\end{equation}

Therefore, integrating at $\phi$:
\begin{equation} \label{spheroidal_Gl_theta}
\dot G_l = 2\pi \int_0^\pi r^{-l+2} P_l(\cos\theta) \cos\theta\sin\theta\sqrt{1 + \frac{r^4}{R^4}\epsilon^4\cos^2\theta\sin^2\theta} d\theta
\end{equation}

The substitution $x = \cos\theta$ yields: $dx = -\sin\theta$, where $\theta\in[0,\pi]\implies x\in[-1, 1]$. Therefore:
\begin{equation}
\dot G_l = 2\pi \int_{-1}^1 r^{-l+2} P_l(x) x\sqrt{1 + \frac{r^4}{R^4}\epsilon^4 x^2 (1-x^2)} dx
\end{equation}

It is important to notice something:
\begin{equation} \label{spheroidal_Gl_parity}
\dot G_l = 2\pi \int_{-1}^1 \underbrace{r^{-l+2}}_{\text{Even}}
\underbrace{P_l(x) x}_{\text{Even if } l\in\{1, 3, 5, \dots\}}
\underbrace{\sqrt{1 + \frac{r^4}{R^4}\epsilon^4 x^2 (1-x^2)}}_{\text{Even}}
 dx
\end{equation}

Here we used the fact, that, $r(x)$ and its derivative are even:
\begin{equation} \label{spheroidal_rdrdtheta_parity}
r(x) = r(-x),\quad\quad
\left(\frac{1}{r}\frac{dr}{d\theta}\right)^2(x) = \left(\frac{1}{r}\frac{dr}{d\theta}\right)^2(-x)
\end{equation}
as can be computed directly from (\ref{spheroidal_r}) and (\ref{spheroidal_drdtheta}). Because the integral domain is symmetric (from -1 to 1), an odd integrand will yield zero, which will happen with $l\in\{0, 2, 4, 6, ...\}$. When integrand is odd, we guarantee nonzero result, due to symmetric domain (unless the entire integrand is identically zero, which is not the case).
\begin{equation}
\begin{split}
\dot G_l = 0, \quad&\text{if}\quad l\in\{0, 2, 4, 6, 8, \dots\} \\
\dot G_l \neq 0, \quad&\text{if}\quad l\in\{1, 3, 5, 7, 9, \dots\}
\end{split}
\end{equation}

Using (\ref{spheroidal_r}), we get:
\begin{equation}
\dot G_l = \frac{2\pi}{R^{-l+2}}\int_{-1}^1\frac{xP_l(x)}{\left[1 - \epsilon^2 x^2\right]^{\frac{-l+2}{2}}}\sqrt{1 + \frac{\epsilon^4 x^2 (1-x^2)}{1 - \epsilon^2 x^2}} dx
\end{equation}

It only makes sense to calculate $G_{2l+1}$, for $l\in\{0, 1, 2, \dots\}$. With that in mind, consider:
\begin{equation}
\begin{split}
\dot G_{2l+1} = 2\pi\int_{-1}^1 r^{-(2l+1)+2} \cdot x P_{2l+1}(x)\cdot\sqrt{1 + \frac{r^4}{R^4}\epsilon^4 x^2 (1-x^2)} dx \\
\dot G_{2l+1} = 2\pi\int_{-1}^1 r^{-2l+1} \cdot x P_{2l+1}(x)\cdot\sqrt{1 + \frac{r^4}{R^4}\epsilon^4 x^2 (1-x^2)} dx
\end{split}
\end{equation}

Again using (\ref{spheroidal_r}), we get:
\begin{equation}
\dot G_{2l+1} = 2\pi R^{-2l+1}\int_{-1}^1 \frac{x P_{2l+1}(x)}{\left[1 - \epsilon^2 x^2\right]^{-l+1/2}}\cdot\sqrt{1 + \frac{\epsilon^4 x^2 (1-x^2)}{(1 - \epsilon^2 x^2)^2}} dx \\
\end{equation}

Which can also be written as:
\begin{equation} \label{spheroidal_G}
\dot G_{2l+1} = \frac{2\pi}{R^{2l-1}}\int_{-1}^1 x P_{2l+1}(x)\left[1 - \epsilon^2 x^2\right]^{l-1/2}\cdot\sqrt{1 + \frac{\epsilon^4 x^2 (1-x^2)}{(1 - \epsilon^2 x^2)^2}} dx
\end{equation}

As the case with the Legendre polynomials, we can write the power in terms of a binomial expansion, except this one will have infinite terms. The more $\epsilon$ approaches one, the more terms will be needed to approximate it.
\begin{equation} \label{expansion_roverR}
\left[1 - \epsilon^2 x^2\right]^{l-1/2} = \sum_{n=0}^\infty\binom{l-\frac{1}{2}}{n}\left(-\epsilon^2 x^2\right)^n
\end{equation}

Where we define the generalized binomial coefficients as:
\begin{equation} \label{generalized_binomial}
\binom{\alpha}{k} := \frac{\alpha (\alpha-1) (\alpha-2) \cdots (\alpha-k+1)}{k!}.
\end{equation}

Rewriting (\ref{spheroidal_G}), we get:
\begin{equation}
\dot G_{2l+1} = \frac{2\pi}{R^{2l-1}}\int_{-1}^1 x P_{2l+1}(x)\cdot\sum_{k=0}^\infty\binom{l-\frac{1}{2}}{n}(-1)^n\epsilon^{2n}x^{2n}\cdot\sqrt{1 + \frac{\epsilon^4 x^2 (1-x^2)}{(1 - \epsilon^2 x^2)^2}} dx
\end{equation}

Making the proper substitutions, and writing the Legendre polynomials as:
\begin{equation} \label{Legendre_polynomials}
P_l(x) = \sum_{k=0}^l a_{lk} x^k
\end{equation}

Then:
$$
P_{2l+1}(x) =
\sum_{k=0}^{2l+1} a_{2l+1,k} x^k =
\sum_{k=0}^{l} a_{2l+1,2k+1} x^{2k+1}
$$

Therefore:
\begin{equation}
\begin{split}
\dot G_{2l+1} = \frac{2\pi}{R^{2l-1}}\int_{-1}^1 x \sum_{k=0}^{l} a_{2l+1,2k+1} x^{2k+1}\cdot\sum_{n=0}^\infty\binom{l-\frac{1}{2}}{n}(-1)^n\epsilon^{2n}x^{2n}\cdot\sqrt{1 + \frac{\epsilon^4 x^2 (1-x^2)}{(1 - \epsilon^2 x^2)^2}} dx \\
\dot G_{2l+1} = \frac{2\pi}{R^{2l-1}}\sum_{k=0}^{l} a_{2l+1,2k+1} \sum_{n=0}^\infty (-1)^n \binom{l-\frac{1}{2}}{n}\epsilon^{2n} \int_{-1}^1 x^{2k+2} x^{2n}\cdot\sqrt{1 + \frac{\epsilon^4 x^2 (1-x^2)}{(1 - \epsilon^2 x^2)^2}} dx \\
\dot G_{2l+1} = \frac{2\pi}{R^{2l-1}}\sum_{n=0}^\infty \sum_{k=0}^{l} (-1)^n a_{2l+1,2k+1}\binom{l-\frac{1}{2}}{n}\epsilon^{2n} \int_{-1}^1 x^{2(k+n+1)}\cdot\sqrt{1 + \frac{\epsilon^4 x^2 (1-x^2)}{(1 - \epsilon^2 x^2)^2}} dx \\
\end{split}
\end{equation}

Finally:
\begin{equation} \label{exact_spheroidal_Gl}
\dot G_{2l+1} = \frac{2\pi}{R^{2l-1}}\sum_{n=0}^\infty\sum_{k=0}^{l} (-1)^n a_{2l+1,2k+1}\binom{l-\frac{1}{2}}{n}\epsilon^{2n} A_{2(k+n+1)}
\end{equation}

Where $A_n$ is the nth moment of $f$ in the $[-1, 1]$, that is:
\begin{equation}
A_n = \int_{-1}^1 x^n f(x) dx
= \int_{-1}^1 x^n \sqrt{1 + \frac{\epsilon^4 x^2 (1-x^2)}{(1 - \epsilon^2 x^2)^2}} dx
\end{equation}

It is possible to write $A_n$ exactly in terms of hypergeometric functions, but, seeking simplicity, we'll content ourselves with the second order Taylor expansion, which we already computed. There are a few important properties: $A_{2n+1} = 0$ by parity. Furthermore, because $f(x) > 0,\forall x\in[-1, 1]$, one gets:
\begin{equation}
\begin{split}
x^2 < 1,&\quad\forall x\in [-1, 1] \\
x^{2n+2} < x^{2n},&\quad\forall x\in [-1, 1] \\
x^{2n+2} f(x) < x^{2n} f(x),&\quad\forall x\in [-1, 1] \\
A_{2n+2} < A_{2n},&\quad\forall x\in [-1, 1]
\end{split}
\end{equation}

Rewriting (\ref{exact_spheroidal_Gl}) as $\dot G_{2l+2} = \frac{2\pi}{R^{2l-1}}\sum_{n=0}^\infty a_n$, one gets:
\begin{equation}
\begin{split}
a_n = \sum_{k=0}^{l} (-1)^n a_{2l+1,2k+1}\binom{l-\frac{1}{2}}{n}\epsilon^{2n} A_{2(k+n+1)} \\
a_n = (-1)^n \binom{l-\frac{1}{2}}{n}\epsilon^{2n} \sum_{k=0}^{l} a_{2l+1,2k+1} A_{2(k+n+1)} \\
a_{n+1} = (-1)^{n+1} \binom{l-\frac{1}{2}}{n}\epsilon^{2n} \sum_{k=0}^{l} a_{2l+1,2k+1} A_{2(k+n+1)}
\end{split}
\end{equation}

Therefore:
\begin{equation}
\frac{a_{n+1}}{a_n} = \frac{
(-1)^{n+1} \binom{l-\frac{1}{2}}{n}\epsilon^{2n+2} \sum_{k=0}^{l} a_{2l+1,2k+1} A_{2(k+n+2)}
}{
(-1)^n \binom{l-\frac{1}{2}}{n}\epsilon^{2n} \sum_{k=0}^{l} a_{2l+1,2k+1} A_{2(k+n+1)} \\
}
\end{equation}

Thus:
\begin{equation}
\frac{a_{n+1}}{a_n} =
\left[
\frac{\frac{1}{2}-l}{n+1} +
\frac{n}{n+1}
\right]
\epsilon^2
\frac{
\sum_{k=0}^{l} a_{2l+1,2k+1} A_{2(k+n+2)}
}{
\sum_{k=0}^{l} a_{2l+1,2k+1} A_{2(k+n+1)} \\
}
\end{equation}

Thus, if $n + 1/2 > l$, then $a_{n+1} < a_n$. More, if both factors are close to one, then, we get $a_{n+1}\approx \epsilon^2 a_n$. This is particularly relevant, if one is interested using the summation for numerical computation.

What is left, is to Taylor expand $f(x)$.
\begin{equation}
f(x) = \sqrt{1 + g(x)},\quad\quad
g(x) = \frac{h(x)}{q(x)},\quad\quad
\begin{array}{ll}
h(x) = \epsilon^4 x^2 (1 - x^2) \\
q(x) = (1 - \epsilon^2 x^2)^2
\end{array}
\end{equation}

Then:
\begin{equation}
\begin{split}
f'(x) = \frac{g'(x)}{2f(x)},\quad\quad
\begin{array}{ll}
h'(x) = 2 \epsilon^4 x (1 - 2x^2) \\
q'(x) = -2\epsilon^2 x (1 - \epsilon^2 x^2)
\end{array} \\
f''(x) = \frac{g''(x)}{2f(x)} - \frac{g'(x) f'(x)}{2 f(x)^2},\quad\quad
\begin{array}{ll}
h''(x) = 2 \epsilon^4 (1 - 6x^2) \\
q''(x) = -2\epsilon^2 (1 - 2\epsilon^2 x)
\end{array} \\
g'(x) = \frac{h'(x) q(x) - h(x) q'(x)}{q(x)^2} \\
q(x)^2 g''(x) = h''(x) q(x) - h(x) q''(x) - 2h'(x) q'(x) + 2h(x)\frac{q'(x)^2}{q(x)} \\
\end{split}
\end{equation}

With that, we find:
$$
f(0) = 1,\quad\quad
f(\pm 1) = 1,\quad\quad
f'(0) = 0,\quad\quad
f''(0) = \epsilon^4
$$

Because we have $f(0) = f(\pm 1) = 1$ and $f'(0) = 0$ and $f''(0) > 0$, there must necessarily exist at least two points of maximum inside $[-1, 1]$, that is, there exists $x_m$ and $-x_m$, where $-1 < x_m < 1$, such that $f'(x_m) = 0$ and $f''(x_m) < 0$. Enforcing $f'(x) = 0$ yields a polynomial equation of 7th degree, in which all the roots can be found analytically. The value $\pm x_m$ is important, because, if we have $f(x_m)$, we'd have an upper bound for $f$ in the $[-1, 1]$ domain, enabling us to estimate the error we're committing by integrating a Taylor expansion, and, by allowing us to come up with better trying functions. It turns out, not much error is done by Taylor expanding, and, in general, the greater $\epsilon$, the greater the error. Therefore, we'll consider:
\begin{equation}
f(x) \approx 1 + \frac{\epsilon^4}{2}x^2
\end{equation}

Therefore, integrating, $A_n$ is found:
\begin{equation} \label{spheroidal_A2n_approx}
A_{2n}\approx\int_{-1}^1 x^{2n}\left[1 + \frac{\epsilon^4}{2}x^2\right] dx =
\frac{1}{n+\frac{1}{2}} + \frac{\epsilon^4}{2}\frac{1}{n+\frac{3}{2}}
\end{equation}

Therefore, while (\ref{exact_spheroidal_Gl}) is the exact expression, a reasonable approximation on top of it can be:
\begin{equation} \label{approx_spheroidal_Gl}
\dot G_{2l+1} \approx \frac{2\pi}{R^{2l-1}}\sum_{n=0}^\infty \sum_{k=0}^{l} (-1)^n a_{2l+1,2k+1}\binom{l-\frac{1}{2}}{n}\epsilon^{2n}
\left[\frac{1}{n+k+1+\frac{1}{2}} + \frac{\epsilon^4}{2}\frac{1}{n+k+1+\frac{3}{2}}\right]
\end{equation}

\section{I-Integrals}
Recalling (\ref{IG_integrals}):
$$
I_{ij} = \int_S r^{-(i+j+2)} P_i(\cos\theta) P_j(\cos\theta) dS
$$

Inserting (\ref{spheroidal_area_element}) at equation above:
$$
\dot I_{ij} = \int_0^\pi \int_0^{2\pi} r^{-(i+j+2)} P_i(\cos\theta) P_j(\cos\theta) \cdot r^2\sin\theta\sqrt{1 + \frac{r^4}{R^4}\epsilon^4\cos^2\theta\sin^2\theta} \cdot d\phi d\theta
$$

Integrating in $\phi$:
\begin{equation} \label{spheroidal_Iij_theta}
\dot I_{ij} = 2\pi \int_0^\pi r^{-i-j} P_i(\cos\theta) P_j(\cos\theta)\sin\theta\sqrt{1 + \frac{r^4}{R^4}\epsilon^4\cos^2\theta\sin^2\theta} \cdot d\theta
\end{equation}

Making the substitution $x=\cos\theta$:
$$
\dot I_{ij} = 2\pi \int_{-1}^1 r^{-i-j} P_i(x) P_j(x)\sqrt{1 + \frac{r^4}{R^4}\epsilon^4 x^2(1-x^2)} \cdot dx
$$

By the same parity arguments as shown in equation (\ref{spheroidal_Gl_parity}) and (\ref{spheroidal_rdrdtheta_parity}), we conclude: $\dot I_{ij} = 0$ iff $P_i P_j$ is odd iff $i+j$ is odd. And, thus, $\dot I_{ij}\neq 0$ iff $i+j$ is even. Therefore, one only needs $\dot I_{ij}$ iff $i+j = 2n$ for some integer $n$. Let $j = 2n-i$, and thus, $I_{i,2n-i}\neq 0$.

Substituing $x=\cos\theta$ into the radial equation (\ref{spheroidal_r}), one gets $r(x)$, and inserting at equation above, one gets:
$$
I_{ij} = 2\pi \int_{-1}^1 \left[\frac{R}{\sqrt{1 - \epsilon^2 x^2}}\right]^{-i-j} P_i(x) P_j(x)\sqrt{1 + \frac{\epsilon^4 x^2(1-x^2)}{(1 - \epsilon^2 x^2)^2}} dx
$$

Thus:
\begin{equation}
I_{ij} = \frac{2\pi}{R^{i+j}} \int_{-1}^1 P_i(x) P_j(x) \left[\sqrt{1 - \epsilon^2 x^2}\right]^{i+j} \sqrt{1 + \frac{\epsilon^4 x^2(1-x^2)}{(1 - \epsilon^2 x^2)^2}} dx
\end{equation}

Because $i+j$ must be even for $I_{ij}\neq 0$, then, $\frac{i+j}{2}$ is an integer. Meaning, the square root in the middle of above expression vanishes, as indicated by expression below:
\begin{equation} \label{spheroidal_Iij_x}
I_{ij} = \frac{2\pi}{R^{i+j}} \int_{-1}^1 P_i(x) P_j(x) \left[1 - \epsilon^2 x^2\right]^{\frac{i+j}{2}} \sqrt{1 + \frac{\epsilon^4 x^2(1-x^2)}{(1 - \epsilon^2 x^2)^2}} dx
\end{equation}

As it was done in the case of $G$-integrals, using (\ref{expansion_roverR}) and (\ref{Legendre_polynomials}), to find:
\begin{equation}
I_{ij} = \frac{2\pi}{R^{i+j}} \int_{-1}^1
\sum_{u=0}^i a_{u,i} x^u
\sum_{v=0}^j a_{v,j} x^v
\sum_{n=0}^{\frac{i+j}{2}} \binom{\frac{i+j}{2}}{n}(-\epsilon^2 x^2)^n
\sqrt{1 + \frac{\epsilon^4 x^2(1-x^2)}{(1 - \epsilon^2 x^2)^2}} dx,
\end{equation}

Unlike $G$-integrals, the expansion by the binomial theorem doesn't yield an infinite series, precisely because the exponent is an integer.
\begin{equation}
I_{ij} = \frac{2\pi}{R^{i+j}}
\sum_{u=0}^i \sum_{v=0}^j a_{u,i} a_{v,j}
\sum_{n=0}^{\frac{i+j}{2}} \binom{\frac{i+j}{2}}{n}(-\epsilon^2)^n
\int_{-1}^1 x^{u+v+2n}
\sqrt{1 + \frac{\epsilon^4 x^2(1-x^2)}{(1 - \epsilon^2 x^2)^2}} dx
\end{equation}

Therefore, for $i+j$ even, the exact expression is:
\begin{equation} \label{exact_spheroidal_Iij}
I_{ij} = \frac{2\pi}{R^{i+j}}
\sum_{u=0}^i \sum_{v=0}^j \sum_{n=0}^{\frac{i+j}{2}} (-1)^n a_{u,i} a_{v,j} \binom{\frac{i+j}{2}}{n}\epsilon^{2n}
A_{u+v+2n}
\end{equation}

We can approximate, using (\ref{spheroidal_A2n_approx}), for $i+j$ even, and get:
\begin{equation} \label{approx_spheroidal_Iij}
I_{ij} = \frac{2\pi}{R^{i+j}}
\sum_{u=0}^i \sum_{v=0}^j \sum_{n=0}^{\frac{i+j}{2}} (-1)^n a_{u,i} a_{v,j} \binom{\frac{i+j}{2}}{n}\epsilon^{2n}
\left[\frac{1}{n + \frac{u+v}{2} + \frac{1}{2}} +
\frac{\epsilon^4}{2}\frac{1}{n + \frac{u+v}{2} + \frac{3}{2}}
\right]
\end{equation}

For $i+j$ odd, we know $I_{ij} = 0$.

\newpage
\chapter{Suspended Hemispherical Integrals}
\section{Exact case}
From (\ref{suspended_hemispherical_r_signed_l}), recall:
$$
r = \ell\cos\theta + \sqrt{R^2 - \ell^2\sin^2\theta}
$$

From (\ref{hemispherical_J}), recall the element of area:
$$
J = r^2\sin\theta\sqrt{1 + \frac{\ell^2}{r^2}\sin^2\theta\left(1 + \frac{\ell\cos\theta}{\sqrt{R^2 - \ell^2\sin^2\theta}}\right)^2}
$$

Both equations are valid for signed $\ell$, depending if one is referring to the upper hemisphere or lower hemisphere. With substitution $x = \cos\theta$, just like it was done with the prolate spheroid, we'll get:
\begin{equation}
\begin{split}
r = \ell x + \sqrt{R^2 - \ell^2 (1 - x^2)} \\
r^2 = \ell^2 x^2 + R^2 - \ell^2 (1 - x^2) + 2\ell x \sqrt{R^2 - \ell^2 (1 - x^2)} \\
r^2 = (R^2 - \ell^2) + 2\ell^2 x^2 + 2\ell x \sqrt{R^2 - \ell^2 (1 - x^2)}
\end{split}
\end{equation}

Like the prolate spheroid, $r(x)$ is an even function, because, if $x>0$, then $\cos\theta > 0$, then we are dealing with the upper hemisphere, then $\ell>0$. For $x<0$ we get $\ell < 0$. Therefore, $r(x) = r(-x)$.

The same substitution can be done with $J$.

$$
J = r^2\sin\theta\sqrt{1 + \frac{\ell^2}{r^2}(1-x^2)\left(1 + \frac{\ell x}{\sqrt{R^2 - \ell^2 (1 - x^2)}}\right)^2}
$$

The $\sin\theta$ was left on purpose precisely because $dx = -\sin\theta d\theta$, thus, it will vanish when a change of variables is done in the integral. Thus, making the substitution into the $r^2$ inside the square root, we get:
\begin{equation}
J = r^2\sin\theta\sqrt{
1 +
\frac{\ell^2(1-x^2)}{
\left[\ell x + \sqrt{R^2 - \ell^2 (1 - x^2)}\right]^2}
\left(1 + \frac{\ell x}{\sqrt{R^2 - \ell^2 (1 - x^2)}}\right)^2}
\end{equation}

For the substitution of $r^2$, we motivate the following definition:
\begin{equation} \label{suspended_hemispherical_JA_exact}
J_A = \sqrt{
1 +
\frac{\ell^2(1-x^2)}{
\left[\ell x + \sqrt{R^2 - \ell^2 (1 - x^2)}\right]^2}
\left(1 + \frac{\ell x}{\sqrt{R^2 - \ell^2 (1 - x^2)}}\right)^2}
\end{equation}

Notice $J_A$ is also an even function, as $J_A(x) = J_A(-x)$, recalling that, if $x>0$ then $\ell > 0$, and, if $x<0$ then $\ell < 0$.

Then, one can write the element of area in the following way:
\begin{equation}
J = r^2 J_A\sin\theta =
J_A\sin\theta\left[
\ell x + \sqrt{R^2 - \ell^2 (1 - x^2)} \\
\right]^2
\end{equation}

Notice that, as $\ell\to 0$ then $J_A\to 1$, meaning, $J\to r^2\sin\theta$, the sphere case.

\subsection{G-Integrals}
Recalling (\ref{IG_integrals}):
$$
G_l = \int_S r^{-l} P_l(\cos\theta) \cos\theta dS
$$

Just like it was done previously with other shapes, it is done here. Unlike the spherical or spheroidal case, integration limits happens at the upper hemisphere is $\theta\in[0, \theta_0]$. For the upper hemisphere, we have:
\begin{equation}
G_l = \int_0^{\theta_0} \int_0^{2\pi} r^{-l} P_l(\cos\theta) \cos\theta r^2\sin\theta J_A d\theta d\phi
\end{equation}

Therefore, integrating in $\phi$, we get:
\begin{equation}
G_l = 2\pi \int_0^{\theta_0} r^{-l+2} P_l(\cos\theta) \cos\theta \sin\theta J_A d\theta
\end{equation}

One shouldn't forget there also exists the bottom hemisphere. The total integration limits are $\theta\in[0, \theta_0]\cup[\pi, \pi-\theta_0]$.
\begin{equation}
G_l = 2\pi \left(\int_0^{\theta_0} + \int_{\pi - \theta_0}^{\pi}\right) r^{-l+2} P_l(\cos\theta) \cos\theta \sin\theta J_A d\theta
\end{equation}

Now, substitution $x=\cos\theta$ can be done. Then, $dx = -\sin\theta d\theta$. Also, recall integration limits from (\ref{trig_hemi_limit}), thus we get:
\begin{equation} \label{suspended_hemispherical_Gl_separated}
G_l = 2\pi \int_{\ell/r_0}^1 r^{-l+2} P_l(x) x J_A dx
+ 2\pi \int_{-1}^{-\ell/r_0} r^{-l+2} P_l(x) x J_A dx
\end{equation}

Recalling $P_n(-x) = (-1)^n P_n(x)$, we can do a substitution $x = -y$ in the second integral of expression (\ref{suspended_hemispherical_Gl_separated}). Thus:
\begin{equation}
G_l = 2\pi \int_{\ell/r_0}^1 r^{-l+2} P_l(x) x J_A(x) dx
+ 2\pi \int_{\ell/r_0}^{1} (-1)^l r^{-l+2} P_l(y) (-y) J_A(-y) dy
\end{equation}

Therefore, using $J_A(x) = J_A(-x)$:
\begin{equation}
G_l = 2\pi \int_{\ell/r_0}^1 r^{-l+2} P_l(x) x J_A(x) dx
+ 2\pi(-1)^{l+1} \int_{\ell/r_0}^{1} r^{-l+2} P_l(x) x J_A(x) dx
\end{equation}

Thus:
\begin{equation} \label{suspended_hemispherical_Gl_exact}
\begin{split}
G_{2l} &= 0 \\
G_{2l+1} &= 4\pi \int_{\ell/r_0}^1 r^{-2l+1} P_{2l+1}(x) x J_A dx
\end{split}
\end{equation}

\subsection{I-Integrals}
Having found an integral expression for $G_l$, we seek attention to $I_{ij}$. Again, looking at (\ref{IG_integrals}):
$$
I_{ij} = \int_S r^{-i-j-2} P_i(\cos\theta) P_j(\cos\theta) dS
$$

Which becomes:
\begin{equation}
I_{ij} = 2\pi\left(\int_0^{\theta_0} + \int_{\pi-\theta}^\pi\right) r^{-i-j} P_i(\cos\theta) P_j(\cos\theta) \sin\theta J_A d\theta
\end{equation}

Substitution $x=\cos\theta$ yields:
\begin{equation}
I_{ij} =
2\pi\int_{\ell/r_0}^1 r^{-i-j} P_i(x) P_j(x) J_A dx
+2\pi\int_{-1}^{-\ell/r_0} r^{-i-j} P_i(x) P_j(x) J_A dx
\end{equation}

Making transformation $x \leftarrow -x$ in the second integral:
\begin{equation}
I_{ij} =
2\pi\int_{\ell/r_0}^1 r^{-i-j} P_i(x) P_j(x) J_A dx
+2\pi\int_{\ell/r_0}^{1} r^{-i-j} (-1)^i P_i(x) (-1)^j P_j(x) J_A dx
\end{equation}

Or:
\begin{equation}
I_{ij} =
2\pi\int_{\ell/r_0}^1 r^{-i-j} P_i(x) P_j(x) J_A dx
+2\pi(-1)^{i+j}\int_{\ell/r_0}^{1} r^{-i-j} P_i(x) P_j(x) J_A dx
\end{equation}

Which simplifies to:
\begin{equation} \label{suspended_hemispherical_Iij_exact}
I_{ij} =
4\pi\int_{\ell/r_0}^1 r^{-i-j} P_i(x) P_j(x) J_A dx,
\quad\quad i+j\in\{0, 2, 4, 6, 8, \dots\}
\end{equation}

And $I_{ij} = 0$ otherwise.

\section{Approximation $\ell\ll R$}
It would only be a matter of substituting $J_A$ from equation (\ref{suspended_hemispherical_JA_exact}) into (\ref{suspended_hemispherical_Gl_exact}) or (\ref{suspended_hemispherical_Iij_exact}) However, $J_A$ form is too complicated to yield a direct integration, thus, we seek to approximate it.

\subsection{Approximating $r^2$}
We need to find an approximation for $r^2$ and for $J$. One possibility, is to Taylor expand with respect to $\ell$, centered at $\ell=0$. Doing that for $r^2$, one can get:
\begin{equation}
\begin{split}
r^2 &= (R^2 - \ell^2) + 2\ell^2 x^2 + 2\ell x \sqrt{R^2 - \ell^2 (1 - x^2)} \\
\frac{dr^2}{d\ell} &= -2\ell + 4\ell x^2 + 2x\sqrt{R^2 - \ell^2(1 - x^2)} + \frac{1}{2}\frac{2x\ell\cdot (-2)\ell(1-x^2)}{\sqrt{R^2 - \ell^2 (1 - x^2)}} \\
\frac{dr^2}{d\ell} &= -2\ell + 4\ell x^2 + 2x\sqrt{R^2 - \ell^2(1 - x^2)} - \frac{2x\ell^2(1-x^2)}{\sqrt{R^2 - \ell^2 (1 - x^2)}} \\
\left.\frac{dr^2}{d\ell}\right|_{\ell=0} &= 2x\sqrt{R^2} = 2xR \\
\end{split}
\end{equation}

Now we find the second derivative:
\begin{equation}
\begin{split}
\frac{d^2 r^2}{d\ell^2} &= \frac{d}{d\ell}\left[ 2\ell + 4\ell x^2 + 2x\sqrt{R^2 - \ell^2(1 - x^2)} - \frac{2x\ell^2(1-x^2)}{\sqrt{R^2 - \ell^2 (1 - x^2)}} \right] \\
\frac{d^2 r^2}{d\ell^2} &= -2 + 4 x^2
+ \frac{1}{2}\frac{2x\cdot (-2)\ell(1 - x^2)}{\sqrt{R^2 - \ell^2(1 - x^2)}}
- \frac{4x\ell(1-x^2)}{\sqrt{R^2 - \ell^2 (1 - x^2)}}
- \frac{2x\ell^2(1-x^2)\cdot \frac{-1}{2} (-2)\ell(1 - x^2)}{\sqrt{R^2 - \ell^2 (1 - x^2)}^3}  \\
\frac{d^2 r^2}{d\ell^2} &= -2 + 4x^2
- \frac{2x\ell(1 - x^2)}{\sqrt{R^2 - \ell^2(1 - x^2)}}
- \frac{4x\ell(1-x^2)}{\sqrt{R^2 - \ell^2 (1 - x^2)}}
- \frac{2x\ell^3(1-x^2)^2}{\sqrt{R^2 - \ell^2 (1 - x^2)}^3}  \\
\frac{d^2 r^2}{d\ell^2} &= -2 + 4x^2
- \frac{6x\ell(1 - x^2)}{\sqrt{R^2 - \ell^2(1 - x^2)}}
- \frac{2x\ell^3(1-x^2)^2}{\sqrt{R^2 - \ell^2 (1 - x^2)}^3}  \\
\end{split}
\end{equation}

Therefore, we conclude:
\begin{equation}
\left. r^2(x)\right|_{\ell = 0} = R^2, \quad\quad
\left.\frac{dr^2}{d\ell}\right|_{\ell=0} = 2xR,\quad\quad
\left.\frac{d^2 r^2}{d\ell^2}\right|_{\ell=0} = -2 + 4x^2
\end{equation}

And thus, the Taylor approximation for $r(x)$ is:
\begin{equation}
r^2(x) = R^2 + 2xR\ell + \left(4x^2 - 2\right)\ell^2 + O(\ell^3)
\end{equation}

\subsection{Approximating $r^n$}
Recall:
\begin{equation}
r = \ell x + \sqrt{R^2 - \ell^2 (1 - x^2)}
\end{equation}

The first derivative and second derivative:
\begin{equation}
\begin{split}
\frac{dr^n}{d\ell} &= nr^{n-1}\frac{dr}{d\ell} \\
\frac{d^2 r^n}{d\ell^2} &= n(n-1)r^{n-2}\left(\frac{dr}{d\ell}\right)^2 + nr^{n-1}\frac{d^2 r}{d\ell^2} \\
\end{split}
\end{equation}

Thus, the following derivatives are required:
\begin{equation}
\begin{split}
\frac{dr}{d\ell} &= x - \frac{\ell(1-x^2)}{\sqrt{R^2 - \ell^2 (1-x^2)}} \\
\frac{d^2r}{d\ell^2} &= -\frac{1-x^2}{\sqrt{R^2 - \ell^2 (1-x^2)}} - \frac{\ell^2(1 - x^2)^2}{\left[R^2 - \ell^2(1-x^2)\right]^{3/2}}
\end{split}
\end{equation}

Therefore, evaluating them at $\ell=0$, yields:
\begin{equation}
\left.\frac{dr}{d\ell}\right|_{\ell=0} = x,\quad\quad
\left.\frac{d^2 r}{d\ell^2}\right|_{\ell=0} = -\frac{1-x^2}{R}
\end{equation}

Thus:
\begin{equation}
\left.\frac{dr^n}{d\ell}\right|_{\ell=0} = nxR^{n-1},\quad\quad
\left.\frac{d^2 r^n}{d\ell^2}\right|_{\ell=0} = n(n-1)R^{n-2}x^2 - nR^{n-1}\frac{1-x^2}{R}
\end{equation}

Re-writing:
\begin{equation} \label{suspended_hemispherical_rn_Taylor_second}
\begin{split}
\left.\frac{dr^n}{d\ell}\right|_{\ell=0} &= nxR^{n-1} \\
\left.\frac{d^2 r^n}{d\ell^2}\right|_{\ell=0} &= nR^{n-2}\left(nx^2 - 1\right)
\end{split}
\end{equation}

\subsection{Approximating $J_A$}
If a generic function $f: \mathbb{R}\to\mathbb{R}$ is written as: $f(x) = \sqrt{1 + g(x)}$, then:
$$
f'(x) = \frac{1}{2}\frac{g'(x)}{\sqrt{1 + g(x)}} = \frac{1}{2}\frac{g'(x)}{f(x)}
$$

Looking at $J_A$ from (\ref{suspended_hemispherical_JA_exact}), we notice $J_A$ as a function of $\ell$ has similar behavior as the above function. We define:
\begin{equation}
\begin{split}
c(\ell) &= \sqrt{R^2 - \ell^2(1 - x^2)} \\
a(\ell) &= \frac{\ell^2(1-x^2)}{\left[\ell x + \sqrt{R^2 - \ell^2(1 - x^2)}\right]^2} = \frac{\ell^2(1 - x^2)}{\left[\ell x +  c(\ell)\right]^2} \\
b(\ell) &= \frac{\ell x}{\sqrt{R^2 - \ell^2 (1 - x^2)}} = \frac{\ell x}{c(\ell)} \\
J_A &= \sqrt{1 + a(\ell)\left(1 + b(\ell)\right)^2}
\end{split}
\end{equation}

Therefore, the first derivative:
\begin{equation}
\begin{split}
\frac{d J_A}{d \ell} &= \frac{1}{2J_A}\left[
\frac{da}{d\ell}\left(1 + b(\ell)\right)^2
+ 2a(\ell) \frac{db}{d\ell} \left(1 + b(\ell)\right)
\right] \\
\frac{d J_A}{d \ell} &= \frac{1}{2J_A}\left(1 + b(\ell)\right)\left[
\frac{da}{d\ell}\left(1 + b(\ell)\right)
+ 2a(\ell) \frac{db}{d\ell}
\right]
\end{split}
\end{equation}

The second derivative:
\begin{equation}
\begin{split}
\frac{d^2 J_A}{d \ell^2} = -&\frac{1}{2J_A^2}\frac{dJ_A}{d\ell}
\left(1 + b(\ell)\right)\left[
\frac{da}{d\ell}\left(1 + b(\ell)\right)
+ 2a(\ell) \frac{db}{d\ell}
\right] \\
+& \frac{1}{2J_A}\frac{db}{d\ell}\left[
\frac{da}{d\ell}\left(1 + b(\ell)\right)
+ 2a(\ell) \frac{db}{d\ell}
\right] \\
+& \frac{1}{2J_A}\left(1 + b(\ell)\right)\left[
\frac{d^2 a}{d\ell^2}\left(1 + b(\ell)\right)
+ \frac{da}{d\ell}\frac{db}{d\ell}
+ 2\frac{da}{d\ell} \frac{db}{d\ell}
+ 2a(\ell) \frac{d^2b}{d\ell^2}
\right] \\
\end{split}
\end{equation}

Now, we find the derivatives of $a$, $b$, and $c$, starting with $c$:
\begin{equation}
\begin{split}
\frac{dc}{d\ell} &= \frac{1}{2}\frac{(-2)\ell(1 - x^2)}{\sqrt{R^2 - \ell^2(1-x^2)}} = -\frac{\ell(1-x^2)}{c(\ell)} \\
\frac{d^2 c}{d\ell^2} &= \frac{d}{d\ell}\left[-\frac{\ell(1-x^2)}{c(\ell)}\right] = -\frac{1 - x^2}{c} + \frac{\ell(1-x^2)}{c^2}\frac{dc}{d\ell} \\
\frac{d^2 c}{d\ell^2} &= -\frac{1 - x^2}{c} + \frac{\ell(1-x^2)}{c^2}\frac{\ell(1 - x^2)}{c} = -\frac{1 - x^2}{c} + \frac{\ell^2(1 - x^2)^2}{c^3} \\
\end{split}
\end{equation}

Then, derivative of $a$:
\begin{equation}
\begin{split}
\frac{da}{d\ell} &= \frac{2\ell(1-x^2)}{\left[\ell x + c(\ell)\right]^2}
-\frac{2\ell^2(1-x^2)}{\left[\ell x + c(\ell)\right]^3}\left(x - \frac{\ell(1 - x^2)}{c(\ell)}\right) \\
\frac{da}{d\ell} &= \frac{2\ell(1-x^2)}{\left[\ell x + c(\ell)\right]^2}
-\frac{2\ell^2(1-x^2)}{\left[\ell x + c(\ell)\right]^3}\left(x - \frac{\ell(1 - x^2)}{c(\ell)}\right) \\
\end{split}
\end{equation}

Then, derivative of $b$:
\begin{equation}
\begin{split}
\frac{db}{d\ell} &= \frac{d}{d\ell}\left[\frac{\ell x}{c}\right] = \frac{x}{c} - \frac{\ell x}{c^2}\cdot\frac{\ell(1 - x^2)}{c} \\
\frac{db}{d\ell} &= \frac{x}{c} - \frac{\ell^2 x(1-x^2)}{c^3},\\
\frac{d^2b}{d\ell^2} &= -\frac{x}{c^2}\frac{dc}{d\ell} + \frac{2\ell x(1-x^2)}{c^3} + \frac{4\ell^2 x(1-x^2)}{c^4}\frac{dc}{d\ell} \\
\end{split}
\end{equation}

Considering that $c(0) = R$, then:
\begin{equation}
\begin{split}
\left.\frac{dc}{d\ell}\right|_{\ell=0} = 0,\quad\quad
&\left.\frac{d^2 c}{d\ell^2}\right|_{\ell=0} = -\frac{1-x^2}{c} = -\frac{1-x^2}{R}, \\
\left.\frac{da}{d\ell}\right|_{\ell=0} = 0,\quad\quad
&\left.\frac{d^2 a}{d\ell^2}\right|_{\ell=0} = \frac{2(1-x^2)}{c^2} = \frac{2(1 - x^2)}{R^2}, \\
\left.\frac{db}{d\ell}\right|_{\ell=0} = \frac{x}{c} = \frac{x}{R},\quad\quad
&\left.\frac{d^2 b}{d\ell^2}\right|_{\ell=0} = 0 \\
\end{split}
\end{equation}

Therefore, evaluating $a(0) = b(0) = 0$, one can get:
\begin{equation} \label{suspended_hemispherical_JA_Taylor}
\left.\frac{dJ_A}{d\ell}\right|_{\ell=0} = 0,\quad\quad
\left.\frac{d^2 J_A}{d\ell^2}\right|_{\ell = 0} =
\left.\frac{d^2 a}{d\ell^2}\right|_{\ell = 0} = \frac{2(1-x^2)}{R^2}
\end{equation}

Thus, the Taylor expansion for $J_A$ is:
\begin{equation}
J_A = 1 + \frac{2(1 - x^2)}{R^2}\ell^2 + O(\ell^3)
\end{equation}

\subsection{Approximating $J$}
Recall we have defined: $J = r^2\sin\theta J_A$. Therefore:
\begin{equation}
\frac{dJ}{d\ell} = \frac{dr^2}{d\ell}\sin\theta J_A + r^2\sin\theta\frac{dJ_A}{d\ell} \\
\end{equation}

Therefore:
\begin{equation}
\begin{split}
\left.\frac{dJ}{d\ell}\right|_{\ell=0} &= 2xR\sin\theta J_A + R^2\sin\theta \cdot 0 \\
\left.\frac{dJ}{d\ell}\right|_{\ell=0} &= 2xR\sin\theta J_A
\end{split}
\end{equation}

Doing the same for second derivative:
\begin{equation}
\frac{dJ}{d\ell} =
\frac{d^2 r^2}{d\ell^2}\sin\theta J_A +
\frac{dr^2}{d\ell}\frac{d J_A}{d\ell}\sin\theta +
\frac{dr^2}{d\ell}\frac{dJ_A}{d\ell}\sin\theta +
r^2\sin\theta\frac{d^2 J_A}{d\ell^2}
\end{equation}

Therefore:
\begin{equation}
\begin{split}
\left.\frac{d^2 J}{d\ell^2 }\right|_{\ell=0} &= 2J_A\sin\theta + 2xR\cdot 0\cdot\sin\theta + 2xR\cdot 0\cdot\sin\theta + R^2\sin\theta\cdot\frac{2(1-x^2)}{R^2} \\
\left.\frac{d^2 J}{d\ell^2 }\right|_{\ell=0} &= 2J_A\sin\theta + R^2\sin\theta\cdot\frac{2(1-x^2)}{R^2}
\end{split}
\end{equation}

We now write $J$ up to second order:
\begin{equation}
J = \left[R^2 + 2xR\ell + 2(1-x^2)\ell^2 + O(\ell^3)\right]\sin\theta
\end{equation}

\subsection{Approximating $r^n J_A$}
This approximation is fairly important, because it is precisely the integrand of the $G_l$ and $I_{ij}$ integrals. We proceed by the same way:
\begin{equation}
\frac{d}{d\ell}\left[r^n J_A\right]_{\ell=0} =
\left[\frac{dr^n}{d\ell}J_A + r^n\frac{dJ_A}{d\ell}\right]_{\ell=0}
\end{equation}

First, notice $r^n|_{\ell=0} = R^n$, precisely because $r|_{\ell=0} = R$. That can also be seen by direct substitution of $\ell=0$ in the expression for $r^n$. Now, we seek to calculate the first derivative. We've already calculated the derivatives, which can be seen from equations (\ref{suspended_hemispherical_JA_Taylor}) and (\ref{suspended_hemispherical_rn_Taylor_second}).
\begin{equation}
\frac{d}{d\ell}\left[r^n J_A\right]_{\ell=0} = nx R^{n-1}
\end{equation}

Now, to second order:
\begin{equation}
\frac{d}{d\ell}\left[r^n J_A\right]_{\ell=0} =
\left[\frac{d^2r^n}{d\ell^2}J_A + 2\frac{dr^n}{d\ell}\frac{J_A}{d\ell} + r^n\frac{d^2J_A}{d\ell^2}\right]_{\ell=0}
\end{equation}

Again, using equations (\ref{suspended_hemispherical_JA_Taylor}) and (\ref{suspended_hemispherical_rn_Taylor_second}), we arrive at:
\begin{equation}
\begin{split}
\frac{d^2}{d\ell^2}\left[r^n J_A\right]_{\ell=0} &=
nR^{n-2}\left(nx^2 - 1\right) + R^n\frac{2(1 - x^2)}{R^2} \\
\frac{d^2}{d\ell^2}\left[r^n J_A\right]_{\ell=0} &=
R^{n-2}\left[n\left(nx^2 - 1\right) + 2(1 - x^2)\right] \\
\frac{d^2}{d\ell^2}\left[r^n J_A\right]_{\ell=0} &=
R^{n-2}\left[n^2 x^2 - n + 2 - 2x^2\right] \\
\frac{d^2}{d\ell^2}\left[r^n J_A\right]_{\ell=0} &=
R^{n-2}\left[\left(n^2 - 2\right) x^2 - \left(n-2\right)\right] \\
\end{split}
\end{equation}

Therefore, the final answer:
\begin{equation} \label{suspended_hemispherical_rnJA_second_order}
r^n J_A = R^n + nx R^{n-1}\ell +
\left[\left(n^2 - 2\right) x^2 - \left(n-2\right)\right] R^{n-2} \ell^2
+ O(\ell^3)
\end{equation}

\subsection{G-Integrals}
To solve the integrals, recall equation (\ref{suspended_hemispherical_Gl_exact}), written below:
\begin{equation}
G_{2l+1} = 4\pi \int_{\ell/r_0}^1 r^{-2l+1} P_{2l+1}(x) x J_A dx
\end{equation}

Our goal would be to Taylor expand the integrand $r^{-2l+1} P_{2l+1}(x) x J_A(x)$. Because the expansion is around $\ell = 0$, the terms $xP_l(x)$ won't contribute, because they do not depend explicitly on $\ell$. The remaining terms $r^{-l+2} J_A$ is a case of (\ref{suspended_hemispherical_rnJA_second_order}), for $n=-l+2$. Therefore, making the substitutions, one can get:
\begin{equation}
G_l = 4\pi \int_{\ell/r_0}^1  P_l(x) x
\left\{R^n + nx R^{n-1}\ell +
\left[\left(n^2 - 2\right) x^2 - \left(n-2\right)\right] R^{n-2} \ell^2\right\}
dx
\end{equation}

Therefore, one is left with the following integrals to compute:
\begin{equation} \label{suspended_hemispherical_Gl_hllR_integral_solution}
\begin{split}
G_l &= 4\pi R^{-l+2} \int_{\ell/r_0}^1  P_l(x) x dx \\
&+ 4\pi (-l+2) R^{-l+1}\ell \int_{\ell/r_0}^1  P_l(x) x^2 dx \\
&+ (l^2 - 4l + 2)\frac{4\pi}{R^l}\ell^2 \int_{\ell/r_0}^1 x^3 P_l(x) dx \\
&+ \frac{4l\pi}{R^l}\ell^2 \int_{\ell/r_0}^1 x P_l(x) dx \\
\end{split}
\end{equation}

\subsubsection{Moments of Legendre Polynomials}
Making reference of (\ref{suspended_hemispherical_Gl_hllR_integral_solution}), one requires to calculate first, second and third moments of the legencre polynomials, that is:
\begin{equation}
\int x^m P_n(x) dx, \quad m\in\{1, 2, 3\}
\end{equation}

A recurrency relation can be useful in computing integrals:
\begin{equation}
(2n+1) P_n(x) = \frac{d}{dx}\left[P_{n+1}(x) - P_{n-1}(x)\right]
\quad\implies\quad
\int P_n(x) dx = \frac{P_{n+1}(x) - P_{n-1}(x)}{2n+1} + C
\end{equation}

To calculate first moments of the Legendre polynomials, all one has to do, is to integrate by parts:
\begin{equation}
\begin{split} \label{Legendre_first_moment}
\int xP_n(x) dx &= x \int P_n(x)dx - \int
\frac{P_{n+1}(x) - P_{n-1}(x)}{2n+1}dx \\
\int xP_n(x) dx &= x\frac{P_{n+1}(x) - P_{n-1}(x)}{2n+1} - \frac{1}{2n+1}\left[\frac{P_{n+2}(x) - P_n(x)}{2(n+1)+1} - \frac{P_n(x) - P_{n-2}(x)}{2(n-1)+1}\right] + C \\
\int xP_n(x) dx &= x\frac{P_{n+1}(x) - P_{n-1}(x)}{2n+1} - \frac{1}{2n+1}\left[\frac{P_{n+2}(x) - P_n(x)}{2n+3} - \frac{P_n(x) - P_{n-2}(x)}{2n-1}\right] + C \\
\end{split}
\end{equation}

Once the first moment is calculated, it can be used to calculate the second moment, again using integration by parts:
\begin{equation}
\begin{split}
\int x^2 P_n(x) dx &= x^2 \int P_n(x)dx - \int 2x\frac{P_{n+1}(x) - P_{n-1}(x)}{2n+1}dx \\
\int x^2 P_n(x) dx &= x^2 \int P_n(x)dx - \frac{2}{2n+1}\left[\int x P_{n+1}(x)dx - \int xP_{n-1}(x) dx\right] \\
\end{split}
\end{equation}

The third moment can be found by the same procedure. An analytical exact solution for any moment can be found, though, the expressions start to grow in size pretty quickly.

\subsubsection{Value of $G_1$}
Given a general expression will be needlessly complicated, we'll focus on $l=1$, that is, $G_1$. The Legendre polynomials for these orders are $P_1(x) = x$. We'll proceed by calculating moments of $P_1$.
\begin{equation}
\begin{split}
\int x P_1(x)dx &= \int x\cdot x dx = \frac{x^3}{3} + C \\
\int x^2 P_1(x)dx &= \int x^2\cdot x dx = \frac{x^4}{4} + C \\
\int x^3 P_1(x)dx &= \int x^3\cdot x dx = \frac{x^5}{5} + C \\
\end{split}
\end{equation}

Using expression (\ref{suspended_hemispherical_Gl_hllR_integral_solution}), we find for $G_1$:
\begin{equation}
\begin{split}
G_1 &= 4\pi R^1 \left.\frac{x^3}{3}\right|_{\ell/r_0}^1
+4\pi \ell \left.\frac{x^4}{4}\right|_{\ell/r_0}^1
-\frac{4\pi}{R}\ell^2 \left.\frac{x^5}{5}\right|_{\ell/r_0}^1
-3\frac{4\pi}{R}\ell^2 \left.\frac{x^4}{4}\right|_{\ell/r_0}^1
\\
G_1 &= 4\pi R \frac{1}{3}\left[1 - \left(\frac{\ell}{r_0}\right)^3\right]
+ 4\pi \ell \frac{1}{4}\left[1 - \left(\frac{\ell}{r_0}\right)^4\right]\\
&- \frac{4\pi}{R} \ell^2 \frac{1}{5}\left[1 - \left(\frac{\ell}{r_0}\right)^5\right]
+ \frac{4\pi}{R} \ell^2 \frac{1}{3}\left[1 - \left(\frac{\ell}{r_0}\right)^3\right] \\
\end{split}
\end{equation}

Or, rewriting:
\begin{equation} \label{suspended_hemispherical_G1_ellllR}
\begin{split}
G_1 &= \frac{4}{3}\pi R\left[1 - \left(\frac{\ell}{r_0}\right)^3\right]
+ \pi\ell\left[1 - \left(\frac{\ell}{r_0}\right)^4\right] \\
&- \frac{4\pi}{5R} \ell^2\left[1 - \left(\frac{\ell}{r_0}\right)^5\right]
+ \frac{4\pi}{3R} \ell^2\left[1 - \left(\frac{\ell}{r_0}\right)^3\right] \\
\end{split}
\end{equation}

\subsection{I-Integrals}
Substituting $r^{-i-j} J_A$ as in equation (\ref{suspended_hemispherical_rnJA_second_order}) into (\ref{suspended_hemispherical_Iij_exact}) for $n=-i-j$, one finds the expression for $I_{ij}$, valid for $i+j$ even, is:
\begin{equation} \label{suspended_hemispherical_Iij_ellllR}
\begin{split}
I_{ij} &=
\frac{4\pi}{R^{i+j}}
\int_{\ell/r_0}^1 P_i(x) P_j(x) dx \\
&-(i+j)\frac{4\pi}{R^{i+j+1}}\ell
\int_{\ell/r_0}^1 x P_i(x) P_j(x) dx \\
&+\left[(i+j)^2 - 2\right]\frac{4\pi}{R^{i+j+2}}\ell^2
\int_{\ell/r_0}^1 x^2 P_i(x) P_j(x) dx \\
&+\left[(i+j) + 2\right]\frac{4\pi}{R^{i+j+2}}\ell^2
\int_{\ell/r_0}^1 P_i(x) P_j(x) dx \\
\end{split}
\end{equation}

The required Legendre moments are more complicated to calculate explicitly:
\begin{equation}
\int x^m P_i(x) P_j(x) dx
\end{equation}

Thus, we'll directly evaluate $I_{00}$ and $I_{11}$. We already know that $I_{01} = I_{10} = 0$, because $1+0$ is odd. Moments as before can be calculated easily:
\begin{equation}
\begin{split}
\int x P_0(x) P_0(x) dx &= \int x dx = \frac{x^2}{2} + C \\
\int x^2 P_0(x) P_0(x) dx &= \int x^2 dx = \frac{x^3}{3} + C \\
\int x P_1(x) P_1(x) dx &= \int x\cdot x^2 dx = \frac{x^4}{4} + C \\
\int x^2 P_1(x) P_1(x) dx &= \int x^2\cdot x^2 dx = \frac{x^5}{5} + C \\
\end{split}
\end{equation}

Using expression (\ref{suspended_hemispherical_Iij_ellllR}) to calculate $I_{00}$, one gets:
\begin{equation}
I_{00} = 4\pi \cdot\left. 1\right|_{\ell/r_0}^1
-0
-2\cdot \frac{4\pi}{R^2}\ell^2\left.\frac{x^3}{3}\right|_{\ell/r_0}^1
+2\cdot \frac{4\pi}{R^2}\ell^2\cdot\left. 1\right|_{\ell/r_0}^1
\end{equation}

Which evaluates into:
\begin{equation} \label{suspended_hemispherical_I00_ellllR}
I_{00} = 4\pi \cdot\left[1 - \frac{\ell}{r_0}\right]
-\frac{8\pi}{3R^2}\ell^2\left[1 - \left(\frac{\ell}{r_0}\right)^3\right]
+\frac{8\pi}{R^2}\ell^2\left[1 - \frac{\ell}{r_0}\right]
\end{equation}

The same thing can be done for $I_{11}$:
\begin{equation}
I_{11} = \frac{4\pi}{R^2}\left. \frac{x^3}{3}\right|_{\ell/r_0}^1
-2\cdot \frac{4\pi}{R^3}\ell\left. \frac{x^4}{4}\right|_{\ell/r_0}^1
+2\cdot \frac{4\pi}{R^4}\ell^2\left.\frac{x^5}{5}\right|_{\ell/r_0}^1
+4\cdot \frac{4\pi}{R^4}\ell^2\left.\frac{x^3}{3}\right|_{\ell/r_0}^1
\end{equation}

Which evaluates into:
\begin{equation} \label{suspended_hemispherical_I11_ellllR}
I_{11} = \frac{4\pi}{3R^2}\left[1 - \left(\frac{\ell}{r_0}\right)^3\right]
-\frac{2\pi}{R^3}\ell\left[1 - \left(\frac{\ell}{r_0}\right)^4\right]
+\frac{8\pi}{5R^4}\ell^2\left[1 - \left(\frac{\ell}{r_0}\right)^5\right]
+\frac{16\pi}{3R^4}\ell^2\left[1 - \left(\frac{\ell}{r_0}\right)^3\right]
\end{equation}

\end{document}